\documentclass[11pt,reqno]{amsart}
\usepackage{amsfonts,amssymb,amsmath,amsopn,amsthm,graphicx}
\usepackage{amsxtra, mathrsfs}
\usepackage{colordvi}
\usepackage[usenames,dvipsnames]{color}
\usepackage{amsfonts,amssymb,amsbsy,amsmath,amsthm,dsfont}
\usepackage{mathtools}
\usepackage[colorlinks=true]{hyperref}
 
\usepackage{enumitem}

\newcommand\version\today

\numberwithin{equation}{section}

\newtheorem{theorem}{Theorem}[section]
\newtheorem{corollary}[theorem]{Corollary}
\newtheorem{lemma}[theorem]{Lemma}
\newtheorem{proposition}[theorem]{Proposition}

\theoremstyle{definition}
\newtheorem{definition}[theorem]{Definition}
\newtheorem{remark}[theorem]{Remark}
\newtheorem{remarks}[theorem]{Remarks} 
\newtheorem{assumption}[theorem]{Assumption}

\newcounter{theoremi}[theorem]

\newcommand{\itemthm}{\refstepcounter{theoremi} {\rm(\roman{theoremi})}{~}}

\newcounter{assumptions}

\newcounter{smalllist}

\newcounter{listi}

\newcounter{smallenum}
\newenvironment{SE}{\begin{list}{{\rm\arabic{smallenum})}}{%
\setlength{\topsep}{0mm}\setlength{\parsep}{0mm}\setlength{\itemsep}{0mm}%
\setlength{\labelwidth}{2em}\setlength{\leftmargin}{2em}\usecounter{smallenum}%
}}{\end{list}}
\newenvironment{SEalph}{\begin{list}{{\rm\alph{smallenum})}}{%
\setlength{\topsep}{0mm}\setlength{\parsep}{0mm}\setlength{\itemsep}{0mm}%
\setlength{\labelwidth}{2em}\setlength{\leftmargin}{2em}\usecounter{smallenum}%
}}{\end{list}}
\newcommand{\Ra}{\big\rangle}
\newcommand{\La}{\big\langle}
\newcommand{\la}{\langle}
\newcommand{\ra}{\rangle}
\newcommand{\veps}{\varepsilon}

\newcommand{\id}{\mathds{1}}

\newcommand{\s}{\mathbb{S}}

\newcommand{\eps}{\varepsilon}
\newcommand{\N}{\mathbb{N}}
\newcommand{\R}{\mathbb{R}}
\newcommand{\Z}{\mathbb{Z}}
\newcommand{\x}{\la x \ra }
\newcommand{\rt}{{\rm curl}}
\newcommand{\C}{\mathbb{C}}

\newcommand{\D}{\mathscr{D}}

\newcommand{\calD}{\mathcal{D}}

\newcommand{\h}{\mathcal{H}}
\newcommand{\Q}{\mathcal{Q}}
\newcommand{\calQ}{\mathcal{Q}}

\newcommand{\U}{\mathcal{U}}

\newcommand{\calA}{\mathcal{A}}
\newcommand{\calB}{\mathcal{B}}
\newcommand{\calC}{\mathcal{C}}

\DeclareMathOperator{\im}{Im}

\DeclareMathOperator{\re}{Re}

\DeclareMathOperator{\supp}{supp}
\DeclareMathOperator{\sgn}{sgn}

\DeclareMathOperator*{\esssup}{ess\,sup}

\newcommand{\wti}{\widetilde  }
\newcommand{\curl}{\text{curl} }
\newcommand{\ol}[1]{\overline{#1}}

\newcommand{\Oh}{O}

\title[Absence of positive eigenvalues of magnetic Schr\"odinger operators]{Absence of positive eigenvalues of magnetic Schr\"odinger operators}

\author {Silvana Avramska-Lukarska}

\address {Silvana Avramska-Lukarska, Department of Mathematics, Institute for Analysis, Karlsruhe Institute of Technology, 76128 Karlsruhe, Germany}

\email {silvana.avramska-lukarska@kit.edu}

\author {Dirk Hundertmark}

\address{Dirk Hundertmark, Department of Mathematics, Institute for Analysis, Karlsruhe Institute of Technology, 76128 Karlsruhe, Germany, and
Department of Mathematics
University of Illinois at Urbana-Champaign
1409 W. Green Street
Urbana, Illinois 61801-2975  }

\email {dirk.hundertmark@kit.edu}

\author {Hynek Kova\v{r}\'{\i}k}

\address {Hynek Kova\v{r}\'{\i}k, DICATAM, Sezione di Matematica, Universit\`a degli studi di Brescia, Via Branze 38 - 25123, Brescia, Italy}

\email {hynek.kovarik@unibs.it}

\begin{document}

\begin{abstract}
We study sufficient conditions for the absence
of positive eigenvalues of magnetic Schr\"odinger operators in 
$\R^d,\, d\geq 2$. In our main result we prove the absence of eigenvalues above certain threshold energy which depends explicitly on the magnetic and electric field.  A comparison with the examples of Miller--Simon shows that our result is sharp as far as the decay of the magnetic field is concerned. As applications, we describe several consequences of the main result for two-dimensional Pauli and Dirac operators, and two and three dimensional Aharonov--Bohm operators. 
\end{abstract}

\maketitle
{\hypersetup{linkcolor=black}
\setcounter{tocdepth}{1}
\tableofcontents}

\section{Introduction and description of main results} \label{sec-intro}
The question of the absence of positive eigenvalues of Schr\"odinger operators has a long history. In 1959 Kato proved that the operator $-\Delta+V$ in $L^2(\R^d)$ has no positive eigenvalues if $V$ is continuous and such that 
\begin{equation} \label{cond-kato}
V(x) = o(|x|^{-1})\qquad |x|\to \infty, 
\end{equation}
by deriving suitable lower bounds on solutions of the Schr\"odinger equation. His lower bound showed that for positive energies these solutions decay so slowly at infinity that they are not normalizable, see \cite{kato}.   
It is known that condition \eqref{cond-kato} is essentially optimal since there exist oscillatory potentials of the Wigner-von Neumann type, decaying as $|x|^{-1}$, which produce positive eigenvalues of the associated Schr\"odinger operator, see \cite{s2, wvn} or \cite[Ex.~VIII.13.1]{rs4}. 

Kato's result was generalized by Simon \cite{s2}, who considered, for $d=3$, potentials of the class $L^2+L^\infty$ which are smooth outside a compact set and allow there a decomposition $V=V_1+V_2$ with $V_1 =o(|x|^{-1}), \ V_2(x) =o(1)$, and
\begin{equation} \label{cond-simon}
\omega_0= \limsup_{|x|\to\infty}\,  x\cdot \nabla V_2(x) < \infty. 
\end{equation}
Under these conditions Simon proved the absence of eigenvalues of $-\Delta+V$ in the interval $(\omega_0,\infty)$. Note that $\omega_0 \geq 0$ since $V_2(x) \to 0$ as $|x|\to\infty$. 
Indeed, if $\omega_0$ were negative then $V_2$ could not have a finite limit for $|x| \to\infty$.   
Later it was shown by Agmon \cite{ag} that under similar assumptions the operator $-\Delta+V$, in any dimension, has no eigenvalues in the interval $(\omega_0/2,\infty)$. 

The use of virial identites to exclude positive 
eigenvalues for specific potentials $V$, such as the 
Coulomb potential, has a long tradition 
in theoretical physics. Rigorous results can be found in 
\cite{weidmann} and \cite{alb}, the latter 
includes also magnetic operators, with strong 
regularity conditions on the magnetic field $B$ 
\emph{and the associated vector potential} $A$, the latter 
being \emph{not invariant} under gauge transformations. 
By exploiting a clever exponentially weighted virial 
identity, Froese, Herbst, and the Hoffmann--Ostenhofs 
proved the absence of all positive eigenvalues of 
$-\Delta+V$ under relative compactness 
conditions on $V$ and $x\cdot \nabla V$, \cite{fh,fhhh} in 
the non--magnetic case. 
Their conditions on the regularity and decay on $V$ and 
$x\cdot \nabla V$ were 
still global but much more general than the pointwise 
conditions of Kato, Simon, and Agmon, or the 
approaches based on virial identities.  
The use of virial identities before the work \cite{fhhh} is 
nicely explained in \cite{ek}.

Yet another approach to the problem is based on Carleman estimates 
in $L^p$-spaces. This method allows to further weaken the regularity 
and decay conditions and to include rough  potentials, see the works 
of Jerison and Kenig  \cite{jk}, Ionescu and Jerison \cite{jk}, and 
the article \cite{kt} by Koch and Tataru. 

\smallskip

Much less is known about the absence of positive eigenvalues for magnetic Schr\"odinger operators of the form  
\begin{equation} \label{mag-op} 
H= (P  - A)^2 +V, \qquad P=-i\nabla, 
\end{equation}
in particular in dimension two. In the above equation $A$ stands for a magnetic vector potential satisfying  ${\rm curl}\, A =B$. The results obtained by Koch and Tataru in \cite{kt} cover also Schr\"odinger operators with magnetic fields. But they impose decay conditions on the vector potential $A$ which are \emph{not gauge invariant} and which imply, in the case of dimension two,
that the total flux of the magnetic field must vanish. Therefore they cannot be applied to two-dimensional Schr\"odinger operators with magnetic fields of non-zero 
flux.

Certain implicit conditions for the absence of eigenvalues of the operator  \eqref{mag-op} in $\R^2$ were recently found by Fanelli, Krej\v ci\v r\'{\i}k and Vega in  \cite{fkv}, see also \cite{fkv2}. However, their result guarantees absence of {\it all} eigenvalues of $H$, not only of the positive ones. Consequently the hypotheses needed in \cite{fkv} include some smallness  conditions on $V$ and $B$ which are not necessary for the absence of positive eigenvalues only. 
In \cite{is} Ikebe and Saito proved a limiting absorption principle for $H$ under certain  pointwise decay conditions on $V$ and $B$, see 
Remark \ref{rem-is}.

A particular case with $V=0$ was considered by Iwatsuka in \cite{iwa}. He proved that if $B$ is smooth, non-constant and translationally invariant in one direction, then the spectrum of $H$ in  dimension two is purely absolutely continuous. We refer to \cite[Sec.~6.5]{cfks} for further reading on this subject. 

\smallskip

In this work we develop quadratic form methods which are an effective tool to rule out positive eigenvalues for a large class of magnetic Schr\"odinger operators while at the same time allowing
the existence of negative eigenvalues, which one does not want to rule out a priori. 
In addition, intuition from physics and experience from the rigorous study of Schr\"odinger operators without magnetic fields clearly show that while eigenvalues depend on global properties of the potential and the magnetic field, energies in the essential spectrum depend only on asymptotic properties. Thus, the nonexistence of eigenvalues embedded in the essential spectrum should depend \emph{only on the asymptotic behavior} of the potential and the magnetic field, as long as the local behavior of the potential and magnetic field is not so singular such 
that it destroys the self--adjointness of the magnetic Schr\"odinger operator.  
Our results make this intuition rigorous: the \emph{local behavior} of the magnetic field and the potential is \emph{largely irrelevant} for the non-existence of positive eigenvalues.   Our results also  cover cases where the magnetic field decays so slowly that no choice of vector potential satisfies the conditions in \cite{kt}.  
\smallskip

In dimension two we identify the magnetic field with a 
scalar function which, in turn, can be interpreted 
as a vector field in $\R^3$  perpendicular to the plane $\R^2$.
In general dimension the magnetic field is closed two--form. Hence
it can be identified with an antisymmetric matrix--valued 
function on $\R^d$ satisfying, in the sense of distributions,  the following condition:
 \begin{equation} \label{maxwell-gen}
\partial_{j} B_{k,i} + \partial_{k} B_{i,j} +\partial_{i} B_{j,k} =0 \qquad \forall \ i,j,k \in \{1,\dots, d\}.\\[4pt]
\end{equation}
Here  $B_{j,k}(x)$ denotes the entries of $B$ at a point $x\in\R^d$. If $d=3$, then $B$ is an antisymmetric $3\times 3$ matrix 
$$
\begin{pmatrix}
0& -B_3 & B_2  \\
B_3 &0   & -B_1\\
-B_2 & B_1 & 0
\end{pmatrix}
$$
which is in turn identified with a vector field $B=(B_1, B_2, B_3)$. Equation \eqref{maxwell-gen} thus coincides with the usual divergence free condition
\begin{equation} \label{maxwell}
\nabla \cdot B = 0 \qquad  [ \, d=3 \, ] \\[3pt] ,
\end{equation}
dictated by Maxwell's equations.

\smallskip
\noindent 
It is well known that as soon as $B$ satisfies \eqref{maxwell-gen}  and certain mild regularity conditions, then there exists a vector 
potential $A$, a one--form, such that $B=\curl A$ or $B=dA$, with 
the exterior derivative.

\subsection{The method}  Let us briefly describe our method and its most important novel features. As already mentioned above we build upon the technique invented 
by R.~Froese and I.~Herbst and M.~and Th.~ Hoffmann--Ostenhof \cite{fhhh} and further developed in \cite{fh}.  The latter is based on  weighted virial identities which require working with dilations and their generator. For non-magnetic Schr\"odinger operators this is facilitated by the fact that the momentum operator $P$ has very simple
commutation relations with dilations. In particular, the domain of $P$ is invariant under dilations. This is not true anymore for the magnetic operators, since the vector potential spoils the dilation invariance of the domain of  $P-A$.   

\smallskip

One of the crucial new features of our approach shows that to overcome this difficulty one has to work with a vector potential $A$ in the Poincar\'e gauge and exploit 
its connection with the dilations and the virial theorem. This connection, which enables us to 
develop a quadratic form version of the magnetic virial theorem, is explained in Section \ref{sec-prelim}. We also show that the rather different conditions of Kato and Agmon--Simon are, in fact, just two sides of the same coin. Kato's condition for the absence of positive eigenvalues can be easily recovered from the quadratic form version of the virial of the potential,  see Section \ref{sec:virial-Kato-form} for details. 

\smallskip

Moreover, the use of the Poincar\'e gauge leads to very natural decay conditions on $B$ required for the absence of positive eigenvalues.   well-known example by Miller and Simon, see Section \ref{sec-examples}, shows that these conditions are sharp. In particular, it follows from the Miller-Simon example that 
no choice of the gauge can provide better decay conditions on $B$.

\subsection{A typical result}  

 In order to describe a typical result with general and easy 
 to verify conditions on the magnetic field $B$ and the potential $V$, we need some more notation. We denote by $L^p=L^p(\R^d)$, $1\le p\le \infty$ the usual scale of Lebesgue spaces. Moreover, we need their locally uniform versions 
 \begin{align}
 	L^p_{\text{loc,unif}}= \Big\{ V:  \sup_{x\in\R^d}\int_{|x-y|\le 1} |V(y)|^p\, dy <\infty \big\}
 \end{align}
 with norms 
 \begin{align}
 	\|V\|_{L^p_{\text{loc,unif}}}\coloneqq   \sup_{x\in\R^d}\Big(\int_{|x-y|\le 1} |V(y)|^p\, dy\big)^{1/p}
 \end{align}
  when $1\le p<\infty$ and the obvious modification for $p=\infty$. Clearly these spaces are nested, that is, $ 	L^q_{\text{loc,unif}}\subset 	L^p_{\text{loc,unif}}$ when 
  $1\le p\le q\le \infty$. 
Moreover, we need 
 \begin{definition}[Vanishing at infinity locally uniformly (in $L^p$)] A function 
 $V\in  L^p_{\text{loc,unif}}$ with 
 \begin{align}
 	\lim_{R\to\infty} \|\id_{\ge R}V\|_{L^p_{\text{loc,unif}}}=0
 \end{align}
 \emph{vanishes at infinity locally uniformly in} $L^p_{\text{loc,unif}}$. 
 
 Here $\id_{\ge R}$ is the characteristic function of the set 
 $\{x\in\R^d:  |x|\ge R\}$.  In fact, we will only need the $p=1,2$ versions of vanishing locally uniformly in $L^p$ at infinity. 
 \end{definition}

	Given a magnetic field $B$ and a point $w\in\R^d$ let 
	$\wti{B}_{w}(x)\coloneqq B(x+w)[x]$. More precisely, $\wti{B}_{w}$ is a vector--field on 
	$\R^d$ with components 
	\begin{align} \label{B-tilde}
		(\wti{B}_{w})_j(x)\coloneqq (B(x+w)[x])_j=
		\sum_{m=1}^d B_{j,m}(x+w)\, x_m \, ,
		\quad j=1,\ldots, d\, .
	\end{align}
  Using translations, we will usually assume $w=0$, in which case we will simply write $\wti{B}$.  
  In dimension two, identifying the magnetic field with a scalar, the vector field $\wti{B}_{w}$ is given by 
  $\wti{B}_{w}(x)= B(x+w)(-x_2,x_1) $ and in three dimensions it is given by the cross product $\wti{B}_{w}(x+w)= B(x+w)\wedge x$. If moreover $w=0$, then 
  the latter represents, in the classical evolution, the time derivative of the magnetic part of the Lorentz force experienced by a unit point charge.
  
  \smallskip
  
  \noindent This definition is inspired by Section 3 in \cite{joergens-weidmann}. It allows us to effectively treat magnetic fields and potentials which can have severe singularities even close to infinity.  

In order to guarantee that there is a locally square integrable 
vector potential $A$ with $dA=B$, we need 
\begin{lemma}\label{lem:exist-vec-pot}
	Given a magnetic field $B$ and $w\in\R^d$ let 
	$\wti{B}_{w}$ be given by \eqref{B-tilde} 
	and assume that 
  \begin{align*}
  	\int\limits_{|x-w|<R} |x-w|^{2-d}\Big(\log \frac{R}{|x-w|} \Big)^2\  |\wti{B}_{w}(x)|^2\, dx <\infty 
  \end{align*}
 for all $R>0$. Then there exists a vector potential $A\in L^2_{\rm loc}(\R^d, \R^d)$ with $B=dA$ in the sense of distributions. 
\end{lemma}

\smallskip
  
\noindent In the sequel, given a vector field $X$ on $\R^d$ we write $X\in L^q_{\text{loc,unif}}$ as a shorthand meaning that the Euclidean 
norm of $X$ belongs to $L^q_{\text{loc,unif}}$. 

\medskip
  
\noindent The simplest version of our results is given by 
\begin{theorem}[Simple version]\label{thm:typical}
	Given a magnetic field $B$ assume that $\wti{B}_{w}\in L^p_{\text{loc,unif}}$ for some $p>d$ 
	and some $
	w\in \R^d$. 
	Then there exists a vector potential $A\in L^2_{\text{loc}}(\R^d,\R^d)$ with $B=dA$.  
	Moreover, let $V$ be a potential with $V\in L^q_{\text{loc,unif}}$ for some $q>d/2$ that allows a splitting $V=V_1+V_2$ such that 
	$x V_1\in L^{q_1}_{\text{loc,unif}}$ for some $q_1>d$ and 
	$x\cdot\nabla V_2 \in L^{q_2}_{\text{loc,unif}}$ for some $q_2>d/2$ and assume  that 
	$\wti{B}$ and $xV_1$ vanish at infinity locally uniformly in $L^2$ and $V$, $V_1$, and $x\cdot\nabla V_2$  vanish at infinity locally uniformly in $L^1$.
	  
	Then the magnetic Schr\"odinger operator $(P-A)^2+V$, defined via quadratic form methods,  has  no positive eigenvalues. 
\end{theorem}
\begin{remarks} 
    	 \itemthm The decay condition on $x V_1$, respectively $x\cdot \nabla V_2$, are  generalizations, in terms of local $L^p$ conditions, of the pointwise conditions of Kato \cite{kato}, respectively Agmon \cite{ag} and Simon  \cite{s2}.   For a generalization using only natural quadratic form conditions, see Theorems \ref{thm:typical-form} and \ref{thm-abs} below. \\[2pt]
	\itemthm  Even in this simplest version the conditions on $B$ and $V$ allow for strong local singularities and the decay condition at infinity is rather mild: for example, if one splits $V$ in such a way that $V_1$ is compactly supported. Then $xV_1$ is zero outside a compact set, so clearly vanishing at infinity. The condition 
	$x V_1\in L^{q_1}_{\text{loc,unif}}$ for some $q_1>d$ allows for rather large local singularities. In particular, the virial $x\cdot\nabla V$ has only to exist in a neighborhood of infinity in order to be able to apply Theorem \ref{thm:typical}. One can also include a long range part of $V$ in $V_1$. 
		 Moreover, since $|\wti{B}_{w}(x)|\lesssim |B(x+w)||x|$, the magnetic field can have strong local singularities, in particular at $w$. 
	The decay of the  magnetic field $B$ has to be faster than 
	$\la x-w\ra^{-1}$, which is, at leas in dimension two, in line of what one expects 
	from the Miller--Simon examples, see Section \ref{ssec-miller-simon}. 
\end{remarks}

\noindent Let us now briefly describe our main results in full generality. 

\subsection{Full quadratic form version: absence of all positive eigenvalues}  It turns out that the absence 
of positive eigenvalues depends, in a sense, \emph{only} on the behavior of $\wti B,\, x V$ and $x\cdot\nabla V$ \emph{at infinity} with respect to the operator $(P-A)^2$. The latter are to be understood in a weak sense according to the following

\begin{definition}[Vanishing at infinity] \label{def-vanishing}
We say that a potential $W$ vanishes at infinity with respect to $(P-A)^2$ if for some $R_0>0$ its quadratic form domain $\calQ(W)$ contains all $\varphi\in\calD(P-A)$ with $\supp(\varphi)\in \U_{R_0}^c$ and for $R\ge R_0$ there exist positive  $\alpha_R, \gamma_R$ with 
	$\alpha_R, \gamma_R \to 0$ as $R\to\infty$ such that 
	\begin{align}
		|\la \varphi, W\varphi \ra | \le 
		\alpha_R \|(P-A)\varphi\|_2^2 + \gamma_R\|\varphi\|_2^2 \quad \text{for all } \varphi\in\calD(P-A) 
		\text{ with } \supp(\varphi)\subset \U_R^c
	\end{align}  
	 Here 	$\U_R=\{x\in\R^d:\, |x|<R\}$ and   
	 $\U_R^c=\R^d\setminus \U_R$ is its complement.
\end{definition}
\noindent By monotonicity  we may assume, without loss of 
generality,  that  $\alpha_R$ and $\gamma_R$ are decreasing in $R\ge R_0$.

\noindent We then have
 
\begin{theorem}\label{thm:typical-form}
	Given a magnetic field $B$ assume that it fulfills 
	the condition of Lemma \ref{lem:exist-vec-pot} 
	for some $w\in\R^d$ and that  $\wti{B}_{w}^2$ given 
	by \eqref{B-tilde} is relatively form bounded 
	and vanishes at infinity with respect to $(P-A)^2$. Moreover, assume that the potential $V$ is form small and vanishes at infinity with respect to 
	$(P-A)^2$ and allows for a splitting $V=V_1+V_2$, 
	such that 
	$|xV_1|^2$ and $x\cdot\nabla V_2$ are also form small 
	and vanish at infinity with respect to $(P-A)^2$. 
	
	Then the magnetic Schr\"odinger operator $(P-A)^2+V$, defined via quadratic form methods, has essential spectrum $[0,\infty)$ and no positive eigenvalues.  
\end{theorem}

\begin{remark}\label{rem-comments} Some comments concerning Theorem \ref{thm:typical-form}:
		\itemthm 
	We only need relative form boundedness  of $\wti{B}_{w}^2$ 
	with respect to $(P-A)^2$.  Its relative form bound 
	does not have to be less than one.	\\[2pt]
	\itemthm 
 	While  the conditions on the potential $V$ and the magnetic field $B$ with respect to 
 $(P-A)^2$ might be difficult to check, the diamagnetic inequality
 \begin{equation}  \label{diamagnetic-inequality}
\big | P |\varphi | \big | \leq \big | (P-A)\varphi \big | \quad \text{a.e.} \qquad \text{for all } \varphi\in\calD(P-A), 
 \end{equation}
 see e.g.~\cite{kato-2},  shows that it is enough to check them with respect to the non-magnetic kinetic energy $P^2$, see \cite{ahs}.\\[2pt]
    \itemthm 
    	 One can again absorb strong local singularities of the potential in a suitable choice of $V_1$. 
    	Thus the local behavior of the potential $V$ and the magnetic field $B$ is \emph{largely irrelevant} for the non-existence of positive eigenvalues. Moreover,  the virial $x\cdot\nabla V_2$ has to exist only in a weak quadratic form sense, see 
     Lemma \ref{lem:xgradv} and the discussion in  Section \ref{sec:virial-Kato-form}.  \\[2pt]
    \itemthm  An inspection of the proof shows that 
    		in Theorem \ref{thm:typical-form} it is enough to assume 
    		that $x\cdot\nabla V$ is bounded from above at infinity 
    		by zero, see Definition \ref{def-bounded infinity} below 
    		for the precise meaning.  
    		Classically the force is given by 
    		$F=-\nabla V$. Thus  $x\cdot F= -x\cdot\nabla V$ is negative, i.e., the force is \emph{confining}, if $x\cdot\nabla V$ is positive, otherwise the force is \emph{repulsive}, i.e., it pushes the particle further to infinity. 
    		Thus in order to prevent localization of a quantum 
    		particle only the positive part of $x\cdot\nabla V$ 
    		should have to be small at infinity.  \\[2pt]
	\itemthm\label{rem-fks} 
 	We would like to stress that unlike many other results on the absence of positive eigenvalues for magnetic Schr\"odinger operators that we are aware of, with the exception of \cite{fkv} and \cite{is}, we impose only conditions on the magnetic field $B$ and not directly on the vector potential $A$.  Decay and regularity conditions on the vector potential $A$ 
 	are not invariant under gauge transformations and thus unphysical. 
    The conditions of \cite{fkv}, on the other hand, are quite restrictive.
	For example, in \cite{fkv} the authors need that various 
	\emph{global}  quantities related to the magnetic field $B$ 
	and to the potential $V$ are absolute form bounded with respect to 
	$(P-A)^2$, i.e. without allowing for lower order 
	terms in the respective bounds  and they need 
	an explicit smallness condition for the various \
	constants involved in their bounds. Consequently, the resulting assumptions turn out to be so 
	strong that they rule out existence of any eigenvalue.
	 
	However, for a large class of physically relevant potentials and magnetic fields one expects that the corresponding magnetic Schr\"odinger operator has negative eigenvalues, while it typically should not have positive eigenvalues, at least when the magnetic field and the potential vanish in a suitable sense at infinity. This is exactly what our Theorem \ref{thm:typical-form} and its generalizations below provide.\\[2pt]
	\itemthm In order to prove invariance of the essential spectrum, one usually assumes that the potential $V$ is relatively 
		$(P-A)^2$ form compact. We do not assume this! 
		In fact, we show in Theorem \ref{thm-abs} that if the 
		potential  $V$  is form small, i.e.,  form bounded with 
		relative bound $<1$, and vanishes at infinity with respect 
		to $(P-A)^2$, then 
	$\sigma_{\text{ess}}((P-A)^2+V) = \sigma_{\text{ess}}((P-A)^2)$. This shows invariance of the essential spectrum under a large class of perturbations. In particular, it confirms the physical intuition that local singularities, as long as they do not destroy form smallness, cannot influence the essential spectrum, at least as a set. 
	For example, one can have a potential with local Hardy type singularity and even a sequence of suitably decreasing Hardy type singularities moving to infinity.  
	Moreover, using ideas of Combesure and Ginibre \cite{cg} and Maz'ya and Verbitzky \cite{mv}, we can  allow perturbations with rather strong oscillations, both locally and at infinity.   \\[2pt]
\itemthm Theorem \ref{thm:typical-form} above is the most general formulation of our results, when one considers magnetic fields and potentials vanishing at infinity, in a suitable sense.  We can allow for much ore general condition  on the potential $V$ and the magnetic field $B$, see  the following section and Section \ref{ssec-ass} below for more general assumptions.  	
\end{remark}

\subsection{Full quadratic form version: absence of eigenvalues above a positive threshold} \label{ssec-full quadratic form version}
If  $\wti B^2, |xV_1|^2$ and $x\cdot\nabla V_2$ 
do not vanish at infinity  with respect to $(P-A)^2$, we can still exclude positive eigenvalues above a certain threshold. For this we need 
\begin{definition}[Bounded at infinity] 
\label{def-bounded infinity}
A potential $W$ is bounded from above at infinity 
with respect to $(P-A)^2$ if for some $R_0>0$ its 
quadratic form domain $\calQ(W)$ contains 
all $\varphi\in\calD(P-A)$ with $\supp(\varphi)\in \U_{R_0}^{\,c}$ 
and for $R\ge R_0$ there exist positive  $\alpha_R, \gamma_R$ 
with $\lim_{R\to\infty}\alpha_R=0$ and 
$\liminf_{\R\to\infty}\gamma_R<\infty$ such that 
	\begin{align}\label{eq-bounded infinity}
		\La \varphi, W\varphi \Ra  \le 
		\alpha_R \|(P-A)\varphi\|_2^2 + \gamma_R\|\varphi\|_2 \quad \text{for all } \varphi\in\calD(P-A) 
		\text{ with } \supp(\varphi)\subset \U_R^c .
	\end{align}
	 By monotonicity  we may assume, without loss of 
	generality,  that  $\alpha_R$ and $\gamma_R$ are 
	decreasing in $R\ge R_0$ in which case we set  		
	$\gamma^+_\infty(W)\coloneqq \lim_{R\to\infty}\gamma_R= \inf_{R}\gamma_R$, the asymptotic bound upper bound of $W$ (at infinity). 
	
	A potential $W$ is bounded from below at infinity with respect to $(P-A)^2$ if $-W$ is bounded from above at infinity. We set $\gamma^-_\infty(W)= \gamma^+_\infty(-W)$. 
	
	A potential $W$ is bounded at infinity with 
	respect to $(P-A)^2$ if $\pm W$ are bounded from above at infinity. We set 
	$$
	\gamma_\infty(W):=\max(\gamma^+_\infty(W), \gamma^-_\infty(W)), 
	$$
	the asymptotic bound of  $W$ (at infinity).  
	+
	We say that a quadratic form $q$, 
	not necessarily given by a locally integrable potential 
	$W$, is bounded from above at infinity w.r.t.\ $(P-A)^2$ 
	if, for all large enough $R>0$,  its domain  
	$\calD(q)$ contains all $\varphi\in \calD(P-A)$ with 
	$\supp(\varphi)\subset\U_R^c$ and a bound of the form 
	\eqref{eq-bounded infinity} holds with 
	$\La \varphi, W\varphi\Ra$ replaced by $q(\varphi)$.  
    We define $\gamma^+_\infty(q)$ similarly 
	as for a potential $W$ and set 
	$\gamma^-_\infty(q)\coloneqq \gamma^+_\infty(-q)$ and 
	$\gamma_\infty(q)\coloneqq \sup(\gamma^+_\infty(q),\gamma^-_\infty(q))$. 

\end{definition}
Using the diamagnetic inequality, one can replace $(P-A)^2$ by $P^2$ in the definition of the asymptotic bounds $\gamma^+_\infty(W)$ 
and $\gamma_\infty(W)$.
We split $V=V_1+V_2$ and set 
\begin{equation}\label{eq-asymptotoc bounds}
  \begin{split}
	\beta^2 \coloneqq 	\gamma_\infty\big(\wti{B}^2\big), \quad 
	\omega_1^2 \coloneqq 	\gamma_\infty\big((xV_1)^2\big) , \quad 
	\omega_2 \coloneqq 	\gamma^+_\infty\big(x\cdot\nabla V_2\big) .
  \end{split}
\end{equation}
 Of course, $\gamma^+_\infty\big(x\cdot\nabla V_2\big)$ is 
a-priori only defined when the distributional derivative 
$x\cdot\nabla V_2$ is given by a nice enough function. In the 
general case, we replace the formal expression 
$\La \varphi, x\cdot\nabla V_2\varphi\Ra$ by the 
quadratic form $q_{x\cdot\nabla V_2}$ associated 
with this distribution. See Lemma 
\ref{lem:xgradv}, and the discussion in 
Section \ref{sec:virial-Kato-form} for the precise 
meaning of this quadratic form.

Under mild regularity conditions the magnetic 
Schr\"odinger operator $(P-A)^2 +V$ has $[0,\infty)$ 
as its essential spectrum and our main result, \
Theorem \ref{thm-abs}, implies that it has no 
eigenvalues larger than
\begin{equation} \label{edge}
\Lambda(B,V) = \Lambda 
  := \frac{1}{4}\left(
  				  \beta+\omega_1 +\sqrt{(\beta+\omega_1)^2+2\omega_2}
  				\right)^2 \, .
\end{equation}
While the $\beta,\omega_1$, and $\omega_2$ might be difficult to compute directly from the definition it is easy to see  
\begin{equation} \label{limits}
\beta \le  \limsup_{|x| \to \infty}  \,  |\wti B(x)| , \qquad \omega_1\le  \ \limsup_{|x| \to \infty}  \, |x|\,  |V_1(x)|,  \qquad \omega_2 \le  \limsup_{|x|\to\infty}\,  x\cdot \nabla V_2(x)\,.
\end{equation}
once the limits are well-defined and finite. 
We would like to point out that Theorem \ref{thm-abs} can be applied also in situations in which the limits in \eqref{limits} might not be defined. 
Morally, $\gamma_\infty(W)$ is the bounded part of $W$ at infinity, modulo terms which are small uniformly locally in $L^1(\R^d)$: 
If $W$ is in the Kato--class outside a compact set, which, in particular, is the case if $W\in L^p_{\text{loc,unif}}$ outside of a compact set, 
and  if $W-W_b$ vanishes at infinity locally uniformly in $L^1(\R^d)$ for some bounded function $W_b$, then 
\begin{equation}
	\gamma_\infty(W)\le \|W_b\|_\infty\, ,
\end{equation}
and a similar bound holds for $\gamma^+_\infty(W)$, see Section \ref{sec-pointwise}.

\subsection{Relation to previous works}
\label{rem-is}
 If $B=0$, then by choosing $V_1=V$ and $V_2=0$ we obtain a generalization of the result of Kato \cite{kato}. On the other hand, by choosing $V_1$ such that $V_1(x) = o(|x|^{-1}),$ 
 and setting $V_2=V-V_1$ we get $\Lambda=\omega_0/2$, see equation \eqref{cond-simon}, and recover thus the results of Agmon \cite{ag} and Simon \cite{s2}. Moreover, Theorem \ref{thm-abs} extends all the above mentioned results to magnetic Schr\"odinger operators with magnetic fields which decay fast enough so that $\beta=0$, see Appendix \ref{sec-pointwise} for more details.

Vice-versa, if $V=0$, then we have $\Lambda= \beta$ which is in agreement with the well--known example by Miller and Simon \cite{ms}, cf.~Section \ref{sec-examples} if one corrects a calculation error in their examples. The Miller--Simon examples show that our bound on $\Lambda$  is sharp. 
\medskip

\noindent It is tempting to  split $V= sV + (1-s)V$ and to optimize the resulting expression for the threshold  energy \eqref{edge}
with respect to $0\le s\le 1$.  
This minimization problem can be explicitly done. It turns out that the minimum is always given by the minimum of the two extreme cases $s=0$ and $s=1$, see Corollary \ref{cor-optimized threshold}.

\medskip 

\noindent Ikebe and Saito proved in \cite{is} a limiting absorption principle, and hence also the absence of eigenvalues of $H$ under the condition that $V$ allows the same  decomposition as above with $|V_1(x)| \leq C\, |x|^{-1-\delta}, \ |V_2(x)| \leq C\, |x|^{-\delta}, \ |x\cdot \nabla V_2(x)| \leq C\, |x|^{-\delta}$, and that $B$ is continuous and  satisfies 
$|B(x)| \leq C\, |x|^{-1-\delta}$. Here $C$ and  $\delta$ are positive constants. Note that these pointwise conditions are covered by Theorem \ref{thm-abs}. Indeed, if $V$ and $B$ satisfy these upper bounds, then 
$\beta=\omega_1=\omega_2=0$, see \eqref{limits}.

\begin{remark}
In \cite{chs} it was proved that if the magnetic fields has the form 
$$
B(x) = \frac{b(\theta)}{r}\, , \quad x = (r \cos\theta, \, r \sin\theta), \quad b\in L^\infty(\s^1),  
$$
then the operator $H_{A,}$ has no eigenvalues above $\| b\|^2_{L^\infty(\s^1)}$. Note that in this case $\Lambda= \| b\|^2_{L^\infty(\s^1)}$.
\end{remark}

\begin{remark}
One of the authors of the present paper established in \cite{kov} dispersive estimates for the propagator $e^{-itH}$ in weighted $L^2-$spaces under the condition that $H$ has no positive eigenvalues, see \cite[Assumption 2.2]{kov}. Theorem \ref{thm-abs} implies that the latter assumption can be omitted. 
\end{remark}

\subsection{Essential spectrum} 
In Section \ref{sec-ess} we establish new sufficient conditions on $B$ under which 
$$
\sigma_{\text{ess}}((P -A)^2)=[0,\infty). 
$$
Roughly speaking we  require that $B(x)\to 0$ not uniformly, but only along a certain path connecting to infinity,  see Theorem \ref{thm-ess-free magnetic} and Definition \ref{def-vanishing somewhere} for details. For example, in $\R^2$ it suffices that $B(x)\to 0$ in a sector of positive opening angle. As a consequence of this result we show that under the assumptions stated in Section \ref{ssec-ass} we have $\sigma_{\text{ess}}((P-A)^2)=[0,\infty)$ , cf.~Corollary \ref{cor-ess-free magnetic}. We also show that if the potential $V$ is form small and vanishes at infinity w.r.t~$(P-A)^2$, then 
$\sigma_{\text{ess}}((P-A)^2+V) = \sigma_{\text{ess}}((P-A)^2)  $,
see Theorem \ref{thm-ess-V vanishing}. For this one usually assumes that 
$V$ is \emph{relative form compact} w.r.t.~$(P-A)^2$ which is 
a considerably stronger assumption, excluding, for example, 
Hardy--type singularities.  
Our result  proves invariance of the essential spectrum under a conditions which includes all  physically relevant examples, even exotic onest with strong singularities or oscillations. 

\smallskip

\subsection{Organization of the paper}  
The article is organized as follows. In Section \ref{sec-notation} we prove some preliminary results on the properties of the Poincar\'e gauge and its 
relation to magnetic Schr\"odinger operators. In Section \ref{sec-prelim} we establish a magnetic virial theorem together with a weighted version, which is our key technical tool. The main 
results are stated and proved in Section \ref{sec-abs}. In Section \ref{sec-examples} we present various examples of applications including Pauli and Dirac operators. Auxiliary material is collected in Appendices.

\section{Magnetic Schr\"odinger operators and the Poincar\'e gauge} \label{sec-notation}
First let us fix some notation. 
Given a set $M$ and two functions $f_1,\, f_2:M\to\R$, we write $f_1(x) \lesssim f_2(x)$ if there exists a numerical constant $c$ such that $f_1(x) \leq c\, f_2(x)$ for all $x\in M$. The symbol $f_1(x) \gtrsim f_2(x)$ is defined analogously. 
Moreover, we use the notation 
$$
f_1(x) \sim  f_2(x)  \quad \Leftrightarrow \quad f_1(x) \lesssim f_2(x) \ \wedge \ f_2(x) \lesssim f_1(x),
$$
and
\begin{equation} 
\lim_{|x|\to \infty} f(x) = L \quad \Leftrightarrow \quad \lim_{r\to\infty} \,  \esssup_{|x|=r} f(x) = L.
\end{equation}
The quantities $\limsup_{|x|\to \infty} f(x)$ and $\liminf_{|x|\to \infty} f(x)$ are defined in a similar way.  
We will use $\partial_j=\frac{\partial}{\partial x_j}$ for the usual partial derivatives in the weak sense, i.e., as distributions.  

\smallskip

\noindent For any $u\in L^r(\R^d)$ with $1\le r\le \infty$ we will use the shorthands
$$
\|u\|_r : =\|u\|_{L^r(\R^d)} \qquad \text{and} \qquad 
\|T\|_{r\to r} := \|T\|_{L^r(\R^d)\to L^r(\R^d)} 
$$
The space $L_{\text{loc}}(\R^d)$ is the space of all complex valued functions 
$f$ such that $f \id_K\in L^r(\R^d)$ for all compact sets $K\subset \R^d$. Here $ \id_K$ stands for the indicator function of $K$. By $L_{\rm loc}^r(\R^d,\R^d)$ we denote the space of all vector fields $v$ which are locally in $L^r$, that is, $|v|  \coloneqq (\sum_{j=1}^d v_j^2)^{1/2}\id_K$ is $L^r_\text{loc}(\R^d)$.   
Given measurable complex valued functions $f,g\in L^2(\R^d)$ we denote by 
$$
\La f, g \Ra= \int_{\R^d} \overline{ f(x)}\cdot g(x)\, dx 
$$
the usual scalar product on $L^2(\R^d)$. 
By the symbol 
$$
\U_R(x) = \{ y\in \R^d\, : \, |x-y| < R\}
$$
we denote the ball of radius $R$ centered at a point $x\in\R^d$. If $x=0$, we abbreviate 
$
\U_R = \U_R(0) $.

\smallskip
\subsection{The magnetic Schr\"odinger operator} \label{sec:magnetic operator}
\noindent Given a magnetic vector potential $A\in L^2_{\text{loc}}(\R^d,\R^d)$, the magnetic Sobolev space is defined by 
\begin{equation} \label{wB}
\h^1_A(\R^d)\coloneqq \calD(P-A) = \big\{ u\in L^2(\R^d)\, : \ (P-A)\, u \in L^2(\R^d) \big\},
\end{equation}
equipped with the graph norm 
\begin{equation}
\| u\|_{\h^1_A} =  \Big( \|(P-A) u\|_2^2  +  \| u\|_2^2 \Big)^{1/2}\, .
\end{equation}
Here  $P=-i\nabla $ is the momentum operator. It is well-know 
that  
\begin{align}
	q_{A,0}(\varphi,\psi)\coloneqq 
		\la (P-A)\varphi, (P-A)\psi \ra 
\end{align}
is a closed sesqui--linear form on $\h^1_A\times \h^1_A$, for any magnetic vector 
potential $A\in L^2_{\text{loc}}(\R^d,\R^d),$ and that 
$\calC_0^\infty(\R^d)$ is dense in $\h^1_A=\calD(P-A)$, 
and $\calC_0^\infty(\R^d)\times\calC_0^\infty(\R^d)$ is 
dense in $\h^1_A\times \h^1_A$. 
see \cite[Thm.~2.2]{s1}.  
By a slight abuse of notation, given a sesqui--linear form 
$q$ with domain $Q\times Q$,  we will use the 
notation $q(\varphi)= q(\varphi,\varphi)$, $\varphi\in Q$, 
for the associated quadratic form. Hence     
\begin{align}
	q_{A,0}(\varphi)= q_{A,0}(\varphi,\varphi)\coloneqq 
		\la (P-A)\varphi, (P-A)\varphi \ra = \|(P-A)\varphi\|_2^2
\end{align}
is a closed quadratic form on $\h^1_A$ for any magnetic vector 
potential $A\in L^2_{\text{loc}}(\R^d,\R^d),$ and that 
$\calC_0^\infty(\R^d)$ is dense in $\calD(P-A)$.  
We will only consider symmetric sesqui--linear forms 
$q:\Q\times Q\to \C$, i.e., 
$q(\varphi,\psi)= \overline{q(\psi,\varphi)}$ for all 
$\varphi,\psi\in Q$. Thus the associated quadratic forms 
will be real--valued. 

Since every closed positive quadratic form on a Hilbert 
space corresponds to a unique self-adjoint positive operators, the 
quadratic form $q_{A,0}$ defines an operator, which we denote by  
$H_{A,0}=(P-A)^2$. 
 Note that for $u\in \calD(P-A)$ one has $Au\in L^1_{\text{loc}}(\R^d,\R^d)$. 
So we only know that $Pu\in L^1_{\rm loc}(\R^d)$ for a typical $u\in \calD(P-A)$, which is one of the sources for technical difficulties of Schr\"odinger operators with magnetic fields. Nevertheless, Kato's inequality shows 
$|\varphi|\in \calD(P)$ for any $\varphi\in \calD(P-A)$ and the diamagnetic inequality \eqref{diamagnetic-inequality}, see also \cite{hs,s2}, yields 
\begin{equation}
	|((P-A)^2+\lambda)^{-1}\varphi| \le (P^2+\lambda)^{-1}|\varphi|
\end{equation}
for all $\lambda>0$ and $\varphi\in L^2(\R^d)$. 

\smallskip

A potential $V$ is a locally integrable,  measurable function 
$V:\R^d\to \R $. Hence its quadratic form domain $\calQ (V) = \calD (|V|^{1/2})$ contains $\calC^\infty_0(\R^d)$. 
The quadratic form $q_V$ corresponding to $V$ is given by 
\begin{align}
	q_V(\varphi,\varphi)= \La |V|^{1/2}\varphi, \sgn(V)|V|^{1/2} \varphi \Ra\, . 
\end{align}
 With a slight abuse of notation, we will often write 
 $	q_V(\varphi,\varphi)= \La \varphi, V \varphi \Ra$. 
A quadratic form $q$ with domain $\calD(q)$ is form bounded w.r.t.\ $(P-A)^2$ if its domain $\calD(q)$ 
 contains $\calD(P-A)$ and there exists 
 $\alpha,C_\alpha<\infty$ such that  
 \begin{align}\label{eq:form-bound-0}
	|q(\varphi)|  \le \alpha \|(P-A)\varphi\|_2^2 + C_\alpha \|\varphi\|_2^2 
	\quad \text{ for all} \ \varphi\in\calD(P-A)\, .
\end{align} 
The infimum 
$$
\alpha_0= \inf\{ \alpha>0: \text{ there exists }C_\alpha<\infty \text{ such that \eqref{eq:form-bound-0} holds for all } \varphi\in \calD(P-A)\}
$$
is called the (relative) form bound of $q$ with respect to 
$(P-A)^2$.   

We say that $q$ is (relative) form small w.r.t~$(P-A)^2$ if 
$\alpha_0<1$, i.e., the bound 
\eqref{eq:form-bound-0} holds for some $0\le \alpha<1$ and 
$C_\alpha<\infty$. If $\alpha_0=0$ one says that $V$ is 
infinitesimally form small w.r.t.~$(P-A)^2$. 

In a similar way this extends to other pairs of  
operators and their associated quadratic forms.
For example, a potential $V$ is form bounded w.r.t\ $(P-A)^2$ 
if the associated quadratic form  
$q_V(\varphi)= \La \varphi, V\varphi\Ra= 
\La \sgn{V}|V|^{1/2}\varphi, |V|^{1/2}\varphi\Ra$ with domain 
$\calD(q_V)= \calQ(V)$ is form bounded w.r.t.\ $(P-A)^2$.  
The potential $V$ is form small, respectively, 
infinitesimally form 
bounded,  w.r.t.\ $(P-A)^2$ if $q_V$ is form 
small, repectively, infinitesimally form bounded,  
w.r.t.\ $(P-A)^2$.

If a quadratic form $q_1$ is form small with respect to 
$(P-A)^2$, the KLMN Theorem, see e.g.~\cite[Theorem 6.24]{teschl}, \cite{rs4}, shows  that the sum 
\begin{align} 
	q_{A,q_1}(\varphi)\coloneqq 
	\|(P-A)\varphi\|_2^2 + q_1(\varphi) 
	= \La (P-A)\varphi,(P-A)\varphi \Ra  
		+ q_1(\varphi)
\end{align}
with domain $\calD(q_{A,q_1})\coloneqq \calD(P-A)$ defines 
a closed quadratic form which is bounded from below. 
It corresponds to a unique self-adjoint operator 
$H_{A,q_1}$, which is called the form sum of $(P-A)^2$ and 
$q_1$. 

In case $q_1=q_V$ is the quadratic form associated to 
a potential $V\in L^1_{\rm{loc}}(\R^d)$, we write  
\begin{align} \label{eq-magnetic quadratic form}
	q_{A,V}(\varphi)\coloneqq 
	\|(P-A)\varphi\|_2^2 + q_V(\varphi) 
	= \La (P-A)\varphi,(P-A)\varphi \Ra  
		+ \La \varphi, V\varphi\Ra
\end{align}
for the form sum and $H_{A,V}= (P-A)^2 +V$ for 
the associated operator. We will sometimes drop 
the dependence of $H_{A,V}$ and simply write $H$ for the 
full magnetic Schr\"odinger operator.  

The diamagnetic inequality implies that if a quadratic 
form $q$ is form bounded, respectively form small 
w.r.t.\ $P^2$, then it is also form bounded, respectively 
form small w.r.t.\ $(P-A)^2$ with the same constants, see 
\cite{ahs}. Except for Tiktopoulos' formula \eqref{Tiktopoulos}, the following is well--known.

\begin{lemma}\label{lem-form bounded}
Let $q$ be a (real--valued) quadratic form with domain 
$\calD(q)\supset\calD(P-A)$. Then $q$ is form bounded 
w.r.t.\ $(P-A)^2$ if and only if for any  
$\lambda>0$ the quadratic form given by 
\begin{align}	q_\lambda(\varphi) \coloneqq q\Big(\big((P-A)^2+\lambda\big)^{-1/2}\varphi\Big)
\end{align}
corresponds to a bounded linear operator $C_q(\lambda)$ 
such that 
$
	\La \varphi, C_q(\lambda)\varphi \Ra = q_\lambda(\varphi)
$ for all $\varphi\in L^2$. 
The bound \eqref{eq:form-bound-0} holds with
\begin{align*}
	\alpha= \left\| C_\lambda\right\|_{2\to 2}
	\qquad \text{and }
	\beta = \lambda  \left\| C_\lambda\right\|_{2\to 2},
\end{align*}
and the relative form bound $\alpha_0$ of $q$ w.r.t. 
$(P-A)^2$ is given by 
 \begin{align*}
	\alpha_0= \lim_{\lambda\to\infty}\left\| C_\lambda\right\|_{2\to 2} \, .
\end{align*}
If 
$q(\varphi)= q_V(\varphi)= \La \sgn(V)|V|^{1/2}\varphi, |V|^{1/2}\varphi \Ra $ for some potential $V\in L^1_{\rm{loc}}(\R^d)$, then 
\begin{align}\label{Tiktopoulos}
	C_\lambda\coloneqq ((P-A)^2+\lambda)^{-1/2}V ((P-A)^2+\lambda)^{-1/2}
\end{align}
Moreover, if $\alpha_0<1$ denote by $H_0=(P-A)^2$ and 
by $H$ the self-adjoint operator given by the 
form sum of the quadratic forms  $q_{A,0}(\varphi)= 
\La (P-A)\varphi, (P-A)\varphi\ra$  
and  $q$. Then  Tiktopoulos' formula for the 
resolvent 
\begin{align}
	(H+\lambda)^{-1}= (H_0+\lambda)^{-1/2}\left(1+ C_q(\lambda)\right)^{-1}  (H_0+\lambda)^{-1/2}
\end{align}
holds for all large enough $\lambda$. 
\end{lemma}

\begin{proof} 
This is well--known, see \cite{rs4}, \cite[Chapter II.3]{simon-thesis}, and, in particular, 
\cite[Theorem 6.30]{teschl}. 
Tiktopoulos' formula \eqref{Tiktopoulos} holds once  
$\lambda>0$ and $-\lambda\in \rho(H)$, the resolvent set 
of $H$ (i.e., the resolvents  $(H+\lambda)^{-1}$ and 
$(H_0 +\lambda)^{-1}$ are defined) and 
$\|C_q(\lambda)\|_{2\to 2}<1$. 
\end{proof}

One could extend the above setting by allowing a 
splitting $V= V_+-V_-$, where the positive and negative parts of $V$ are given by $V_\pm = \max(\pm V,0)$. The discussion in \cite{s1} shows 
that for arbitrary $V_+\in L^1_{\text{loc}}$, the quadratic form 
\begin{align}
	q_{A,V_+}(\varphi,\varphi)\coloneqq \|(P-A)\varphi\|_2^2 + \La \varphi, V_+\varphi\Ra
	= \|(P-A)\varphi\|_2^2 + \| \sqrt{V_+}\varphi \|_2^2 
\end{align}
is well defined and closed on the form domain 
$\calD(Q_{A,V_+}) = \calD((P-A))\cap \calQ(V_+) $, where $\calQ(V_+)= \calD(\sqrt{V_+}) $ and that  
$\calC^\infty_0$ is still dense in 
$\calD(Q_{A,V_+})$ in the graph norm 
$\|\varphi\|_{A,V_+}= (\Q_{A,V_+}(\varphi)+\|\varphi\|_2^2)^{1/2}$. 
Again this closed quadratic form corresponds to a unique self--adjoint 
operator $H_{A,V_+}$ and in order to define a self--adjoint operator $H_{A,V}$ via the KLMN theorem.  
It is enough to assume   that $V_-$ is form small w.r.t.~$H_{A,V_+}$.

More important for us is the observation due to Combescure and Ginibre \cite{cg} that rather singular potentials $V$ can be form bounded with respect to $P^2$, and by the diamagnetic inequality then also with respect to $(P-A)^2$.

\begin{lemma}  \label{lem-combescure-ginibre}
Assume that $V= \nabla\cdot \Sigma +W$, where $\Sigma\in L^2_{loc}(\R^d;\R^d)$, and $W$ is locally integrable. Suppose 
that $\Sigma^2$ and $W$ are form bounded 
w.r.t.~$(P-A)^2$, respectively $P^2$.
Then the quadratic form 
\begin{align}
	\La \varphi, V\varphi\Ra 
	\coloneqq 
	-2\im \La \Sigma \varphi, (P-A)\varphi \Ra
	+ \La\varphi, W\varphi \Ra
\end{align}
is also  form bounded 
w.r.t.~$(P-A)^2$, respectively $P^2$. 
\end{lemma}

\begin{proof}
For $\varphi\in\calC^\infty_0$, an integration by parts shows 
  \begin{align*}
  \La \varphi, (\nabla\cdot \Sigma)\varphi \Ra
  	= -\im \La \Sigma \varphi, P\varphi \Ra 
  	= 	 -2\im \La \Sigma \varphi, (P-A)\varphi \Ra \, .
  \end{align*}
  Thus, for all $\veps>0$  
\begin{align*}
	|\La \varphi, (\nabla\cdot \Sigma)\varphi \Ra|
  	\le  2\| \Sigma \varphi\|\, \| P\varphi \| 
  	\le \veps   \| P\varphi \|^2 + \veps^{-1} \| \Sigma \varphi\|^2  
  	\le   (\alpha+\veps)\| P\varphi \|^2 +\veps^{-1}C_\alpha \| \Sigma \varphi\|^2  
\end{align*}
when $ \| \Sigma \varphi\|^2 = \La \varphi, \Sigma^2\varphi\Ra \le    \alpha\| P\varphi \|^2 +C_\alpha \| \Sigma \varphi\|^2   $. 
The claim follows. 
\end{proof} 

\noindent Using the Young inequality, one can also relate the coefficients in the quadratic form bounds.    

\smallskip

\noindent Moreover, at least in the non--magnetic case, the beautiful work of Maz'ya and Verbitsky \cite{ms} shows that all potential $V$ which are relatively form bounded w.r.t.~$P^2$ are of this form.

\subsection{The Poincar\'e gauge} 
The magnetic field at the point $x\in\R^d$ is given by an antisymmetric two-form  $B(x): \R^d\times\R^d\to \R $, which we identify with a matrix valued 
function $B$ given by  
$$
B(x)= (B_{j,m}(x))_{j,m=1}^d,
$$ 
which is antisymmetric,  $B_{j,m}(x)= -B_{m,j}(x)$ for all $1\le j,m\le d, x\in\R^d$.

Any vector potential $A$, or more precisely a one form, generates 
a magnetic field via the exterior derivative $B=dA$, in the distributional sense.   
In matrix notation, $B_{j,m}= \partial_j A_m - \partial_m A_j$. 
In three space dimensions, one can identify the two form $B$ with a vector valued function $B=\curl A$. 

\smallskip

\noindent  For a given magnetic field $B$ and a point $w\in \R^d$ we define the vector field  $\wti{B}_{w}$ by equation \eqref{B-tilde}, and
put 
\begin{equation}  \label{eq-a}
A_{w}(x) \coloneqq \int_0^1 \wti{B}_w(tx)\, dt = \int_0^1 B(tx+w)[tx] \, dt \, , 
\end{equation}
which is the vector potential in the Poincar\'e gauge. Using translations, it is no loss of generality to assume $w=0$, in which case we will simply write $A$ for the vector potential given by \eqref{eq-a}. 
By going to spherical coordinates, one easily checks at least for nice, say continuous or even smooth, magnetic fields $B,$ that the above vector 
potential is well defined and that $dA =B$  in the sense of distributions. 

Since $B$ is antisymmetric the vector $\wti{B}(x)= B(x)[x]$ is orthogonal to $x$. 
Hence, when $w=0$ the vector potential $A$ given by \eqref{eq-a} satisfies the transversal, or Poincar\'e, gauge 
\begin{equation} \label{poincare}
x\cdot A(x) =0 \qquad \forall\ x\in\R^d ,
\end{equation} 
which will be very important in our discussion of dilations and the virial 
theorem for magnetic Schr\"odinger operators in Section \ref{sec-prelim}. It is easy to see that for $A$ given by \eqref{eq-a} one has $A\in L^2_{\text{loc} }(\R^d,\R^d)$ for bounded magnetic fields $B$ and this extends to a large class of singular magnetic fields, see  Lemma \ref{lem-L2-loc} below. Except otherwise noted, we will always use the Poincar\'e gauge in the following. For a nice discussion of the Poincar\'e gauge from a physics point of view see \cite{jackson} and from a more mathematical point of view, but still for rather regular magnetic fields,  see \cite{thaller}.

\subsection{Hypotheses} \label{ssec-ass}
We will use the following hypotheses on $B$ and $V$:

 \begin{assumption}\label{ass-B-mild-int} The magnetic field $B$ is such that for some $w\in\R^d$ and   
 \begin{align} 
 	\R^d\ni x\, \mapsto\,  |x-w|^{2-d}\, \log ^2\Big( \frac{R}{|x-w|} \Big)\ \wti{B}_{w}(x)^2 \in L^1_{\text{loc}}(\R^d)
 \end{align}	
 for all $R>0$. 
 \end{assumption}
 \noindent As already remarked, there is no loss of generality assuming 
 $w=0$ by using translations. 
Together with Lemma \ref{lem-L2-loc} the above mild integrability condition then assures that the corresponding vector potential 
	in the Poincar\'e gauge is locally square integrable, 
	which is essential in order to define the magnetic Schr\"odinger operator. 
	The magnetic field $B$ can have severe local singularities, 
	while Assumption $\ref{ass-B-mild-int}$ still holds.   
\begin{assumption} \label{ass-B-rel-bounded}
The scalar field $|\wti{B}|^{ 2} $ is relatively form bounded w.r.t.~$(P-A)^2$, where $A$ is the Poincar\'e gauge vector potential corresponding to $B$, That is, 
\begin{equation} \label{ass-B1-eq}
\La \varphi, |\wti{B}|^{ 2}  \varphi \Ra\ = \|\wti{B}  \varphi\|_2^2 \, \lesssim \,  \|(P-A)\varphi\|_2^2  +  \| \varphi\|_2^2 \qquad \forall\, \varphi\in \calD(P-A). 
\end{equation}    
 \end{assumption}

\begin{assumption} \label{ass-V-form small}
The potential $V$ is relatively form small w.r.t.~$(P-A)^2$, that is,  
there exist constants $\alpha_0 < 1$ and $\gamma>0$ such that 
\begin{equation} \label{ass-V0-eq}
 |\La \varphi,  V \, \varphi \Ra|  \leq \    \alpha_0\,  \|(P-A) \varphi\|_2^2  + \gamma \| \varphi\|_2^2 \qquad \forall\, \varphi\in \h^1_A(\R^d) .
\end{equation} 
\end{assumption}
\noindent

We also need similar conditions on the virial 
$x\cdot\nabla V$ of the potential. 
Since we don't want to impose strong differentiability conditions on $V$, one has to be a bit careful: 
The virial $x\cdot\nabla V$ is, at first,  a distribution.
When $\varphi\in \calC^\infty_0(\R^d)$, an formal integration 
by parts argument shows that 
\begin{equation}\label{eq-kato-form}
  \begin{split}
	  	q_{x\cdot\nabla V}(\varphi)
	  	  &= \La \varphi , x\cdot\nabla V\varphi \Ra  
  		  =  -d\La \varphi , V \varphi\Ra  
  			- 2\re\La xV \varphi ,\nabla \varphi \Ra \\
  		  &=   -d\La \varphi , V \varphi\Ra  
  			- 2\im\La xV \varphi ,P \varphi \Ra 
  		  =   -d\La \varphi , V \varphi\Ra  
  			- 2\im\La xV \varphi ,(P-A) \varphi \Ra 
  \end{split}
\end{equation}
since $\La xV \varphi ,A \varphi \Ra$ is real for all 
$A\in L^2_{\rm{loc}}(\R^d,\R^d)$. 
We assume that the form $q_{x\cdot\nabla V}$ extends to 
a quadratic form whose domain contains all 
$\calD(P-A)$ and, by a slight abuse of notation, will 
write $q_{x\cdot\nabla V}$ for this extension. 
A careful discussion when $q_{x\cdot\nabla V}$ is 
form bounded w.r.t.\ $(P-A)^2$ 
is given in  Lemma \ref{lem:xgradv} and in  Section 
\ref{sec:virial-Kato-form}.

\noindent For the assumptions which give us control of 
virial $x\cdot\nabla V$, we decompose the potential $V=V_1+V_2$
How one splits $V=V_1+V_2$ is quite arbitrary, as long as the conditions below are met.

\begin{assumption} \label{ass-V-split form bounded}
  If the potential is split as $V=V_1+V_2$, then $V_1, x^2V_1^2$ and $x\cdot\nabla V_2$ are relatively form bounded w.r.t.~ $(P-A)^2$. 
\end{assumption}
  
The above assumptions are all we need to prove a quadratic 
form version of the virial theorem, see Theorem 
\ref{thm:ABmagnetic-virial}. In particular, 
$\widetilde{B}_w^2$ and the virial $x\cdot\nabla V$ do 
not have to be form small but only form bounded w.r.t\ 
$(P-A)^2$, for the virial theorem to hold. 


\subsubsection*{\bf Behaviour at infinity} 
  We need to quantify the notion that the magnetic field $B$, the potential $V$ and the virial $x\cdot\nabla V$ are bounded, or even vanish, at infinity.

From physical heuristics, one expect  that `smallness' should not be measured pointwise, but only \emph{relative to} the kinetic energy 
$(P-A)^2$. 
The following conditions make this physical intuition precise.  
\begin{assumption} \label{ass-V-vanishing infinity}\textbf{(Vanishing at infinity)} The potential $V$ vanishes  
at infinity w.r.t.~$(P-A)^2$ in the sense of Definition \ref{def-vanishing}. Moreover, if we split $V=V_1+V_2$ as in Assumption \ref{ass-V-split form bounded}, then also $V_1$ vanishes at infinity w.r.t.~$(P-A)^2$ in the sense of Definition \ref{def-vanishing}.
\end{assumption}

\noindent The precise notion of being bounded at infinity w.r.t.~$(P-A)^2$ is given by  
\begin{assumption} \label{ass-bounded infinity}  \textbf{(Boundedness of the magnetic field and the virial at infinity)}
There exist positive sequences $(\eps_j)_{j}, (\beta_j)_{j}$ and $(R_j)_{j}$ with $\eps_j \to 0$ and $R_j\to\infty$ as $j\to\infty$, such that for all $\varphi\in \calD(P-A)$ with ${\rm supp}(\varphi) \subset \U_j^{\, c}=\{x\in\R^d: |x|\ge R_j\}$ 
\begin{align} \label{ass-B-inft}
\| \wti{B}  \varphi\|_2^2  \  & \leq \   \eps_j\,  \|(P-A) \varphi\|_2^2  + \beta_j^2 \, \| \varphi\|_2^2 .
\end{align} 
For the decomposition $V=V_1+V_2$ of the potential, we also assume that there exist positive sequences 
$(\omega_{1,j})_{j}$  and $(\omega_{2,j})_{j}$ 
such that for all $\varphi\in \calD(P-A)$ with ${\rm supp}(\varphi) \subset \U_j^{\, c}$
\begin{align} 
\| x\,  V_1  \varphi\|_2^2   \  & \leq \   \eps_j\,  \|(P-A) \varphi\|_2^2  + \omega_{1,j}^{\, 2} \, \| \varphi\|_2^2 \\
\La \varphi,  x\cdot \nabla V_2  \varphi \Ra\ &  \leq \   \eps_j\,  \|(P-A) \varphi\|_2^2  + \omega_{2,j}\,  \| \varphi\|_2^2  \label{ass-x nablaV_2-inft}.
\end{align}  
\end{assumption}

\noindent By monotonicity  we may assume, without loss of 
generality,  that the sequences $\beta_j$, $\omega_{1,j}$, 
and $\omega_{2,j}$ in Assumption \ref{ass-bounded infinity} are decreasing. We define
\begin{equation} \label{def-beta-omega}
\beta \coloneqq \lim_{j\to \infty} \beta_j \, , 
\qquad \omega_k\coloneqq \lim_{j\to \infty} \omega_{k,j} \,  , \quad k=1,2.
\end{equation}
The relative bounds on $\wti{B}$, etc., at infinity, which give a precise quantitative notion on how large, relative to $(P-A)^2$, the magnetic field $B$, respectively the virial $x\cdot\nabla V$, are at infinity w.r.t.~$(P-A)^2$. 
The assumptions above are partially inspired 
by Section 3 in \cite{joergens-weidmann}.
\smallskip

\noindent
\textbf{Unique continuation at infinity.} 
\noindent For a unique continuation type argument at infinity, we also need a quantitative version of relative form boundedness.  
\begin{assumption} \label{ass-bounded unique continuation}
 If $V=V_1+V_2$, then we assume  
 \begin{align} 
  \| \wti{B}  \varphi\|_2^2 +\|xV_1\varphi\|_2^2  
  	&  \leq \   \frac{\alpha_1^2}{4}\,  \|(P-A) \varphi\|_2^2  + C_1 \| \varphi\|_2^2 , \label{ass-B2-eq}\\
  \La \varphi,    x\cdot \nabla V_2 \,  \varphi \Ra\  
    &  \leq \   \alpha_2\,  \|(P-A) \varphi\|_2^2  + C_2 \| \varphi\|_2^2 , \label{ass-V1-eq} \\
  |\La \varphi,  V_1 \, \varphi \Ra|\  & \leq \   \alpha_3\,  \|(P-A) \varphi\|_2^2  + C_3 \| \varphi\|_2^2   \label{ass-xV-quantitative}
 \end{align} 
 for some $\alpha_j, C_j >0, \ j=1,2,3$, all $\varphi\in\calD(P-A)$, and 
\begin{equation} \label{alphas}
  \alpha_1 + \alpha_2 +d\alpha_3< 1. 
\end{equation} 
\end{assumption}
\noindent The factor $d$ in front of 
$\alpha_3$ comes from the Kato form of the virial $x\cdot\nabla V_1$, see Lemma \ref{lem-kato-virial}.

\smallskip

\begin{remark}	   
Let us make two comments concerning the above list of conditions.  First, all the above 
hypothesis are either physically motivated or required to be able to define the relevant objects. Secondly, the required conditions are quite weak. In Appendix \ref{sec-pointwise} we show that  
 Assumptions \ref{ass-B-mild-int}-\ref{ass-bounded unique continuation} are satisfied under certain mild and  easily verifiable regularity and decay conditions on $B$ and $V$, see in particular Proposition \ref{prop-pointwise}.
	
\end{remark}

\begin{remark}
In the conditions above, one can use the 
  	   diamagnetic inequality in order to replace 
  	   $P-A$ by the nonmagnetic momentum operator $P$ 
  	   in all relative form boundedness conditions, 
  	   see \cite{ahs}. 
\end{remark}

\subsection{Regularity of the Poincar\'e gauge map} 
Note that the Poincar\'e gauge map \eqref{eq-a} is a--priori only well-defined when the magnetic field $B$ is sufficiently regular, say, continuous. 
Our first result shows that the map \eqref{eq-a} can be continuously extended to all magnetic field satisfying Assumption \ref{ass-B-mild-int}. 

\begin{lemma} \label{lem-L2-loc}
Let $\calB$ be the vector space of vector fields $\wti{B}$ satisfying 
\begin{align*}
	\int_{\U_R}  |x|^{2-d}\, \Big(\log \frac{R}{|x|} \Big)^2\  |\wti{B}(x)|^2\,   dx <\infty
\end{align*}
for all $R>0$.  The continuous vector fields are dense in $\calB$ and the map $\wti{B}\mapsto A\coloneqq T(\wti{B})$ given by 
\begin{align*}
	A(x)= T(\wti{B})(x)\coloneqq \int_0^1 \wti{B}(tx)\, dt \quad \text{for } x\in \R^d\, ,
\end{align*}
extends to a continuous map from  $\calB$ into $L^2_{\rm loc}(\R^d,\R^d)$. 
In particular, the Poincar\'e gauge map given in \eqref{eq-a} is well defined for all magnetic fields satisfying Assumption \ref{ass-B-mild-int}.  Moreover, 
\begin{equation} 
\begin{split}\label{eq-AB}
\int_{\U_R} |x|^{2-d}\,  |A(x)|^2\, dx  \, &\leq \, 4  \int_{\U_R}  |x|^{2-d}\, \Big(\log \frac{R}{|x|} \Big)^2\  |\wti{B}(x)|^2\,   dx\, , 
\end{split}
\end{equation}
for any $R>0$. 
\end{lemma}
\begin{proof} Given $B\in\calB$ and $R>0$ let 
\begin{equation*}
	\|\wti{B}\|_{\calB,R} \coloneqq \left(\int_{\U_R}  |x|^{2-d}\, \Big(\log \frac{R}{|x|} \Big)^2\   |\wti{B}(x)|^2\,   dx
	\right)^{1/2}\, .
\end{equation*} Also let $\calA$ be the space of vector potentials $A$ for which 
\begin{equation*}
	\|| A \|_{\calA,R} \coloneqq \left( \int_{\U_R} |x|^{2-d}\,  |A(x)|^2\, dx  \right)^{1/2}
\end{equation*}
is finite for all $R>0$. This makes $\calA$ and $\calB$ locally convex metric spaces and by construction, $\calA\subset L^2_{\text{loc}}(\R^d,\R^d)$.  The metrics consistent with the topologies on $\calA$ and $\calB$ are, for example, 
\begin{align*}
	d_\calA (A_1,A_2) = \sum_{n=0}^\infty 2^{-n}\frac{\|| A_1-A_2 \|_{\calA,2^n}}{1+ \| A_1-A_2 \|_{\calA,2^n}} \quad \text{and} \quad 
	d_\calB (\wti{B}_1,\wti{B}_2) = \sum_{n=0}^\infty 2^{-n}\frac{2\| \wti{B}_1 -\wti{B}_2 \|_{\calB,2^n}}{1+ 2\| \wti{B}_1-\wti{B}_2 \|_{\calB,2^n}}\, .
\end{align*}
The standard arguments show that $\calA$ and $\calB$ are complete metric spaces, see e.g.~\cite[Sec.~V.]{rs1}. Moreover, the usual cutting and mollifying arguments show that the continuous functions are dense in $\calB$. In addition,  since the map $x\mapsto \frac{x}{1+x}$ is increasing on $\R_+$, $T(\wti{B})$ is well defined and locally bounded when $\wti{B}$ is continuous, so $T(\wti{B})\in \calA$, when $\wti{B}$ is continuous.  Assuming temporarily \eqref{eq-AB} then gives   
\begin{align*}
	d_\calA(T(\wti{B}_1), T(\wti{B}_2)) = \sum_{n=0}^\infty 2^{-n}\frac{\|| T(\wti{B}_1-\wti{B}_2) \|_{\calA,2^n}}{1+ \| T(\wti{B}_1-\wti{B}_2) \|_{\calA,2^n}}
	\le d_\calB(\wti{B}_1,\wti{B}_2)
\end{align*}
so  $T$ is uniformly continuous,  thus it extends  to a map from $\calB$ into $\calA$ which we continue to denote by $T$. 
This shows that the Poincar\'e gauge map \eqref{eq-a} is well defined for all magnetic fields $B$ satisfying  Assumption \ref{ass-B-mild-int}. 

Hence it is enough to prove the bound \eqref{eq-AB} and by density, it is enough to prove it for continuous vector fields $\wti{B}$.  
Let $g$ be a radial function, which is positive and finite for almost all $|x|<R$. Since $A(x)= \int_0^1 \wti{B}(tx)\, dt$, we have using symmetry 
\begin{align*}
	\int_{\U_R} & g(|x|) |A(x)|^2\, dx 
		= \int_0^1	\int_0^1 \int_{\U_R} g(|x|) \wti{B}(t_1x)\cdot \wti{B}(t_2x)\, dx dt_1dt_2 \\
		&= 2\iint_{0\le t_1<t_2\le 1} \int_{|x|\le R} g(|x|) \wti{B}(t_1x)\cdot \wti{B}(t_2x)\, dx dt_1dt_2 \\
		&= 2 \int_0^1 \int_0^1 \int_{\U_{tR}} g(|y|/t) t^{1-d} 
					\wti{B}(uy)\cdot \wti{B}(y)\, dy du dt  = 2 \int_{\U_R} \left(\int_{|y|/R}^1  g(|y|/t) t^{1-d}\, dt \right)
					A(y)\wti{B}(y)\, dy \\ 
		&\le 2\left( \int_{\U_R} g(|y|) |A(y)|^2\, dy \right)^{1/2}	\Big( \int_{\U_R}g(|y|)^{-1} \Big( \int_{|y|/R}^1  g(|y|/t) t^{1-d}\, dt \Big)^2 |\wti{B}(y)|^2\, dy\Big)^{1/2}			 
\end{align*}
where we also used the substitution $t_1= ut_2$ and $y=t_2x$ and then the Cauchy-Schwarz inequality. Thus as soon as  $\int_{|x|\le R} g(|x|) |A(x)|^2\, dx$ is finite, we arrive at the a-priori bound 
\begin{align}\label{eq-a-priori}
	\int_{\U_R} g(|x|) |A(x)|^2\, dx \le 4 \int_{\U_R}g(|x|)^{-1}\left(\int_{|x|/R}^1  g(|x|/t) t^{1-d}\, dt \right)^2 |\wti{B}(x)|^2\, dx\, .
\end{align}
It remains to choose $g$ in a such a way  that the integral weight on the left hand side coincides with the expression in \eqref{eq-AB}. Hence we set $g(s)= s^{2-d}$,  and calculate
\begin{align*}
	g(|x|)^{-1} \left( \int_{|x|/R}^1  g(|x|/t) t^{1-d}\, dt \right)^2 
		= |x|^{2-d} \Big(\log \frac{R}{|x|} \Big)^2 .
\end{align*}
Plugging this into \eqref{eq-a-priori} gives \eqref{eq-AB}. 
We note that  $A(x)= \int_0^1\wti{B}(tx)\, dt$ is locally bounded as long 
as $\wti{B}$ is locally bounded. Thus  for the above choice of $g$ 
\begin{align*}
	\int_{\U_R}  |x|^{2-d} \, |A(x)|^2\, dx
\end{align*}
is, as required, finite for all continuous $\wti{B}$. Hence the a-priori bound \eqref{eq-AB} holds for all continuous $\wti{B}$ and extend by density to all of $\calB$.  
\end{proof}

\noindent For future purposes we will need also a translated and generalized version of inequality \eqref{eq-AB};

\begin{corollary}
Let assumptions of Lemma \ref{lem-L2-loc} be satisfied and let $h: \R_+\to \R_+$ be a non-increasing bounded function. Then 
\begin{equation} \label{eq-AB-trans}
\int_{\U_R(x_0)} |x-x_0|^{2-d}\,  h(|x-x_0|)\, |A_{x_0}(x-x_0)|^2\, dx  \, \leq \, 4  \int_{\U_R}  |y|^{2-d}\, \Big(\log \frac{R}{|y|} \Big)^2\ h(|y|) |\wti{B}_{x_0}(y)|^2\,   dy\, , 
\end{equation}
holds for any $x_0\in\R^d$. Recall that $A_{x_0}(x-x_0)$ is given by \eqref{eq-a}.
\end{corollary}

\begin{proof}
Since $h(|y|/t) \leq h(|y|)$ for all $y\in\R^d$ and $t \leq 1$, the result follows from \eqref{eq-a-priori} upon setting $g(s)= s^{2-d}\, h(s)$ and translating.
\end{proof}

\noindent Together with the quadratic form $Q_{A,V}$ we will also need the associated sesqui-linear form
\begin{equation} \label{sesq-f}
q_{A,V}(u,v) = \La(P-A)\,  u ,  (P-A)\, v \Ra + \La  u, V v \Ra 
	= q_{A,0}(u,v) +\La u, V v \Ra  , \qquad u,v \in \h^1_A(\R^d)
\end{equation} 
and denote by $H=H_{A,V}$ the self-adjoint operator associated with $Q_{A,V}$.

\section{Dilations and the magnetic virial theorem} \label{sec-prelim}

\noindent
As already mentioned in the introduction, the aim of this section is to establish a weighted  virial theorem for weak eigenfunctions which will 
be needed later 
in the proof of absence of positive eigenvalues.


\subsection{Dilations and the Poincar\'e gauge} 
In this subsection we will study the behavior of the magnetic Schr\"odinger form $Q_{A,V}$ under the action of the dilation group. 
\smallskip

\noindent Let $D_0$ be  the operator defined on $ C_0^\infty(\R^d)$ by 
\begin{equation} 
D_0= \frac 12 \left(P\cdot x +x\cdot P\right), \qquad \D(D_0) = C_0^\infty(\R^d). 
\end{equation}

\begin{remark} 
Note that $D_0  = \frac 12 \left((P-A)\cdot x +x\cdot (P-A)\right)$, when $A$ is in the Poincar\'e gauge \eqref{poincare}. This is one of the reasons why dilations and the Poincar\'e gauge work well together. 
\end{remark}

\begin{lemma}  \label{lem-2} $D_0$ is essentially self-adjoint. 
\end{lemma} 

\begin{proof}
For $t\in\R $ define the unitary dilation operator $U_t$ by 
\begin{equation}  \label{eq:U}
(U_t f)(x) = e^{td/2} f(e^t x) \, \quad x\in\R^d.
\end{equation} 
It is easy to see that $U_t$ is unitary on $L^2(\R^d)$ and forms a group, 
$U_t U_s= U(t+s)$, for all $t,s\in\R$. In particular, the adjoint is 
given by $U_t^*= U_{-t}$. Moreover, each $U_t$ leaves $C_0^\infty(\R^d)$ 
 invariant and a direct calculation shows that $t\mapsto U_t$ is strongly differentiable on $C_0^\infty(\R^d)$ with 
\begin{equation} \label{d0}
\Big(\frac{d}{dt} U_t f\Big) \Big|_{t=0} = i D_0 f, \qquad \forall\ f\in  C_0^\infty(\R^d). 
\end{equation}

\smallskip

\noindent The claim now follows from \cite[Thm.~VIII.10]{rs1}. 
\end{proof}

\noindent We denote by $D$ the closure of $D_0$, which is self-adjoint, and by  $D_t$ the operator given by 
\begin{equation} \label{Dt}
i D_t = \frac{U_t - U_{-t}}{2t}\, .
\end{equation}
$D_t$ is bounded and symmetric. We will use it to approximate $D$ in the limit $t\to 0$.   Let $\varphi\in\calD(P)$. It is easy to check the commutation formula  
\begin{equation} \label{eq-comm-P}
	PU_t = e^tU_tP\, , 	
\end{equation}
since 
$(PU_t\varphi)(x) = -i\nabla(e^{td/2}\varphi(e^t x)) = -i e^t e^{td/2}(\nabla\varphi)(e^t x) = e^t (U_t(P\varphi))(x)$. 
In a similar way, one checks that for a multiplication operator $V$ the commutation formula 
\begin{equation} \label{eq-comm-V}
	V(\cdot)U_t = U_t V^*_t(\cdot) \qquad  \text{with} \qquad V^*_t(\cdot)\coloneqq V(e^{-t}\, \cdot)
\end{equation}
holds on its domain, i.e., for all $\varphi\in\calD(V)$ we have 
$(V(U_t\varphi))(x) = e^{td/2}V(x)\varphi(e^t x) =(U_t(V^*_t\varphi))(x)$ for almost all $x\in\R^d$. 
A similar result also holds for vector valued multiplication operators, for example, 
\begin{equation}\label{eq-comm-A}
	A(\cdot)U_t = U_t A_{-t}(\cdot)\coloneqq U_t A(e^{-t}\cdot)
\end{equation}

For the virial theorem, we want to define the commutator $[H_{A,V}, iD]$, where $D$ is the generator of dilations. Since the two operators involved are unbounded, this usually leads to involved domain considerations. Even worse, in our case we do not know the domain $\calD(H_{A,V})$ exactly, nor do we intend to know it, since we prefer to work only with quadratic forms. 
This seems to make a usable virial theorem impossible to achieve,  however, a quadratic form approach turns out to be feasible.  

Assuming that $u\in \calD(H_{A,V})$ we approximate the unbounded generator of dilations $D$ by the bounded approximations $D_t$. A slightly formal calculation, for $u\in \calD(H_{A,V})\cap \calC^\infty_0$ which might be $\{0\}$, however, gives 
\begin{equation} \label{eq:informal-1}
	\la u, [H_{A,V}, iD_t] u \ra 
	= \la H_{A,V}u, iD_t u \ra  +\la iD_t u, H_{A,V}u \ra 
	= 2\re(\la H_{A,V}u, iD_t u \ra) 
\end{equation}
since $iD_t$ is antisymmetric. Assume that $\calD(P-A)$ is invariant 
under dilations. Then $iD_tu \in \calD(P-A)$ and 
the right hand side of  \eqref{eq:informal-1} can be 
identified with  $2\re (q_{A,V}(u, iD_tu))$, where 
$q_{A,V}$ is the quadratic form given by 
\eqref{eq-magnetic quadratic form}.
Since $\calD(H_{A,V})$ is dense in $\calD(q_{A,V})=\calQ(H_{A,V})$, the latter expression extends to all of $\calQ(H_{A,V})$, the quadratic form domain of $H_{A,V}$. So we simply \emph{define} the commutator $[H_{A,V}, iD_t]$ as the quadratic form with domain $\calQ(H_{A,V})$, 
\begin{equation}\label{eq:commutator-form-1}
	\la u, [H_{A,V}, iD_t] u \ra 
		\coloneqq 2\re \big(q_{A,V}(u, iD_t u)\big),   \qquad u\in \calQ(H_{A,V})  .
\end{equation}
Moreover, we can define the commutator $[H, iD]$, again in the sense of quadratic forms, by 
\begin{equation} \label{eq:comutator-form-2}
\La u, i \, [H, D] \, u \Ra\coloneqq \lim_{t\to 0} \la u, [H_{A,V}, iD_t] u \ra \coloneqq \lim_{t\to 0} 2\re\big(q_{A,V}(u, i D_t u)\big)\, ,
\end{equation} 
provided the limit on the right hand side exists.  

In the remaining part of this section, we will deal with 
the proof that the limit 
in \eqref{eq:comutator-form-2} exists for all 
$\varphi\in\calD(P-A)$, the calculation of this limit, 
and, in particular, the claim that $\calD(P-A)$ is invariant under dilations under natural conditions on the magnetic field. 
By \eqref{eq-comm-P}, the Sobolev space $\calD(P)$ is invariant under dilations. 
To see how one can also get this for the magnetic Sobolev space $\calD(P-A)$  let $\varphi\in \calD(P-A)$. Then, as distributions, 
\begin{align}\label{eq-key-dilation-inv}
	(P-A)U_t\varphi = e^tU_tP\varphi - U_t A_t\varphi = e^tU_t(P-A)\varphi + U_t(e^t A -A_{-t})\varphi  \, .
\end{align}
Since $U_t:L^2(\R^d)\to L^2(\R^d)$ is unitary and $(P-A)\varphi\in L^2(\R^d)$, we have $e^tU_t(P-A)\varphi\in L^2(\R^d)$ for all $t\in\R $. So in order that $U_t\varphi\in \calD(P-A)$ we have to check if $(e^t A -A_{-t})\varphi \in L^2(\R^d)$. 
This is the content of the next proposition. 
\begin{proposition} \label{prop-2}
Suppose that the magnetic field $B$ satisfies Assumption \ref{ass-B-mild-int}, the vector potential $A$ corresponding to $B$ is in the Poincar\'e gauge, and $\wti{B}^{\, 2} $ is relatively form bounded w.r.t.~$(P-A)^2$. \\
  If $\varphi\in D(P-A) = \h^1_A(\R^d)$, then $(e^tA-A_{-t})\varphi\in L^2(\R^d)$ for all $t\in\R $ and the map $\R\ni t\mapsto (e^tA-A_{-t})\varphi$ is continuous. In particular, $\calD(P-A)$ is invariant under dilations.  
\end{proposition}
\noindent The main tool for the proof of Proposition \ref{prop-2} is the following  

\begin{lemma} \label{lem-a-bound}
Under the assumptions of Proposition \ref{prop-2}, if  
$\varphi\in D(P-A) = \h^1_A(\R^d)$, then  
\begin{equation} \label{at-uppb}
	\| (e^t A- A_{-t} ) \varphi \| \ 
		\leq \ 
			e^t\big(e^{C_z|t|}-1\big)\|(P-A)\varphi\| 
					+ \frac{z\,C_z}{C_z\pm 1} \big(e^{(C_z\pm 1) |t|} -1\big) \|\varphi\|
\end{equation} 
for all $t\in\R $ and $z>0$, 
where the $+$ sign holds for $t\ge 0$ and the $-$ sign for $t<0$ and  the constant $C_z$ is given by
\begin{equation*}
	C_z =  \sqrt{d}\  \big\| \wti{B} \big((P-A)^2 +z^2\big)^{-\frac 12} \big\|_{2\to 2} \, . 
\end{equation*}
\end{lemma}

\begin{remark}
  In the above bound we use the convention $\frac{C_z}{C_z-1} \big(e^{(C_z-1) |t|} -1\big) = |t| $ when $C_z=1$. 
\end{remark}
\noindent Given Lemma \ref{lem-a-bound}, the proof of Proposition \ref{prop-2} is simple.

\begin{proof}[Proof of Proposition \ref{prop-2}:]
  Given $\varphi\in \calD(P-A)$, Lemma \ref{lem-a-bound} shows that $(e^tA-A_{-t})\varphi\in L^2(\R^d)$ for all $t\in\R $ and then \eqref{eq-key-dilation-inv} shows that $U_t\varphi\in \calD(P-A)$ for all $t\in\R$. Thus $\calD(P-A)$ is invariant under dilations. 
  Moreover,  the bound \eqref{at-uppb} shows that the map $t\mapsto (e^tA-A_{-t})\varphi$ is continuous at $t=0$. 
  Since
  \begin{align}\label{eq-a-splitting-trick}
  	e^{t+s}A-A_{-(t+s)} 
  		= e^s\big(e^tA-A_{-t}\big) + e^s A_{-t} - (A_{-s})_{-t}
  			= e^s\big(e^tA-A_{-t}\big) + U_t^*\big(e^s A - (A_{-s})\big)U_t
  \end{align}
  and $U_t\varphi\in\calD(P-A)$ for any $\varphi\in\calD(P-A)$, continuity of $t\mapsto \big(e^{t}A-A_{-t}\big)\varphi $ 
  at $t=0$ implies continuity at all $t\in\R $.  
\end{proof}

\begin{proof}[Proof of Lemma \ref{lem-a-bound}: ]
First of all, it is enough to prove \eqref{at-uppb} for $\varphi\in \calC^\infty_0(\R^d)$, since this is dense in $\calD(P-A)$ in the graph norm: If  \eqref{at-uppb} holds for $\varphi\in \calC^\infty_0(\R^d)$, then given $\varphi\in\calD(P-A)$, choose a sequence $\varphi_n\in \calC^\infty_0(\R^d)$ such that $(P-A)\varphi_n\to (P-A)\varphi$ and $\varphi_n\to\varphi$. By taking a subsequence, if necessary, we can also assume that $\varphi_n\to \varphi$ almost everywhere, hence $(e^t A- A_{-t} ) \varphi_n\to(e^t A- A_{-t} ) \varphi$ almost everywhere, in particular, $|(e^t A- A_{-t} ) \varphi|=\lim_{n\to\infty} |(e^t A- A_{-t} ) \varphi_n|= \liminf_{n\to\infty} |(e^t A- A_{-t} ) \varphi_n|$ almost everywhere. Then Fatou's Lemma and \eqref{at-uppb} imply 
\begin{align*}
	\|(e^t A- A_{-t} ) \varphi\| 
		&\le \liminf_{n\to\infty} \| (e^t A- A_{-t} ) \varphi_n \|  \le e^t\big(e^{C_z|t|}-1\big)\|(P-A)\varphi\| 
					+ \frac{z\,C_z}{C_z\pm 1} \big(e^{(C_z\pm 1) |t|} -1\big) \|\varphi\|
\end{align*}
for all $\varphi\in\calD(P-A)$. Let $t\in  \R $. Since $A$ is in the Poincar\'e gauge, using the change 
of variables $t=e^{-s}$, we have
\begin{equation} \label{A-dilat}
A = \int_0^\infty e^{-s} \wti{B}(e^{-s} \cdot)\, ds 
	= \int_0^\infty e^{-s}  U_s^* \,  \wti{B}\,  U_s \, ds \, .
\end{equation}
From the definition of $A_{-t}$ and \eqref{A-dilat} we get 
\begin{align}
e^t\, A -A_{-t} & = e^t  \int_0^\infty e^{-s}  U_s^* \,  \wti{B}\,  U_s \, ds - \int_0^\infty e^{-s} \,  U_t^* U_s^* \,  \wti{B}\,  U_s  U_t\,ds 
		\nonumber\\
&=  e^t  \int_0^\infty e^{-s}\,  U_s^* \,  \wti{B}\,  U_s \, ds  -    e^t\int_t^\infty e^{t-s}  U_s^* \,  \wti{B}\,  U_s \, ds  
= e^t \int_0^t e^{-s}\,  U_s^* \,  \wti{B}\,  U_s \, ds. \label{eq-A-difference}
\end{align}
Let $\varphi\in C_0^\infty(\R^d)$. The above identity then gives 
\begin{equation} \label{vt}
v_t : = (e^t A- A_{-t} ) \varphi = e^t \int_0^t e^{-s}\,  U_s^* \,  \wti{B}\,  U_s \, \varphi\, ds. 
\end{equation} 
Define the operator $R_z: D(P-A) \to  D(H_0)^d$ by 
\begin{equation}\label{eq-Rz}
R_z := ((P-A)^2 +dz^2)^{-1} \, (P-A -iz).
\end{equation}
Here $( (P-A -iz))$ is a vector operator, which maps 
$\varphi\in\calD(P-A)$ to the vector function 
$(P-A-iz)\varphi=\big((P_j-A_j-iz)\varphi)_{j=1,\ldots,d}$. 
Then $R_z(P-A+iz)\varphi=\varphi$, so 
\begin{align*}
\wti{B}\,  U_s\, \varphi &= \wti{B}\, R_z (P-A+iz) \, U_s\, \varphi =  \wti{B}\, R_z\, U_s \big[ e^s (P-A)\, \varphi +(e^s A-A_{-s})\, \varphi +iz \varphi\big] \\
& =  \wti{B}\, R_z\, U_s \big[ e^s (P-A)\, \varphi +v_s +iz \varphi\big],
\end{align*}
which in view of \eqref{vt} implies 
\begin{equation}\label{eq-vt-expanded} 
v_t = \int_0^t e^{t-s}\, U_s^*  \wti{B}\, R_z\, U_s \big (e^s (P-A)\, \varphi +v_s +iz \varphi\big )\, ds .
\end{equation}
Note that the map $t\mapsto v_t\in L^2(\R^d)$ is continuous due to the presence of $\varphi$.
Hence, if $t\ge 0$, 
\begin{align*} 
	w(t) \coloneqq \|v_t\| \, &\leq \, K_z \int_0^t e^{t-s}\, \big (e^s \|(P-A)\, 		\varphi\|_2 +w(s) +z \|\varphi\| \big )\, ds = E(t) + K_z \int_0^t e^{t-s}\, w(s)\, ds,
\end{align*}
where
\begin{equation} \label{kz}
K_z := \|\wti{B}\, R_z\|_{2 \to 2}\, ,
\end{equation}
and
$$
E(t) = K_z \int_0^t e^{t-s}\, \big (e^s \|(P-A)\, \varphi\|_2  +z \|\varphi\|_2 \big )\, ds\, . 
$$
We will derive 
a suitable bound on  $K_z$ at the end of this proof. 
Lemma \ref{lem-gronwall} in the Appendix yields 
\begin{equation} \label{wt-uppb}
w(t) \ \leq \ E(t) +K_z \int_0^t e^{(1+K_z)(t-s)}\, E(s)\, ds\, .
\end{equation} 
Note 
\begin{align*}
	\int_0^t &e^{(1+K_z)(t-s)}\, E(s)\, ds =\\
		&= K_z\iint\limits_{0<s<s'<t} e^{(1+K_z)(t-s')}
				e^{s'} \, ds ds'\, \|(P-A)\, \varphi\|_2
			 +zK_z\iint\limits_{0<s<s'<t} e^{(1+K_z)(t-s')}
				e^{s'-s}\, \, ds ds' \, \|\varphi\|_2 \\
		&= \Big(\frac{e^t}{K_z}\big( e^{K_zt}-1) -te^t \Big)\|(P-A)\, \varphi\|_2 + z\Big( \frac{1}{K_z+1}\big( e^{(K_z+1)t} -1\big) -\big(e^t-1 \big) \Big)\|\varphi\|_2
\end{align*}
and a straightforward calculation gives 
\begin{align*}
	E(t) &= K_z te^t \|(P-A)\varphi\|_2 + zK_z \big( e^t-1\big)\|\varphi\|_2 .
\end{align*}
Inserting this into \eqref{wt-uppb} gives
\begin{equation*}
\|(e^t A- A_{-t} ) \varphi\| = w(t) \leq e^t\big( e^{K_zt}-1\big)\|(P-A)\, \varphi\|_2
			+ \frac{zK_z}{K_z+1}\big( e^{(K_z+1)t} -1\big) \|\varphi\|_2
\end{equation*}
which gives \eqref{at-uppb}, at least for $\varphi\in\calC^\infty_0(\R^d)$. If $t<0$, then setting $\tau=-t>0$, we get from \eqref{eq-vt-expanded} 
\begin{align*}
	\wti{w}(\tau) \coloneqq \|v_{-\tau}\|_2 \, &\leq \, K_z \int_0^\tau e^{s-\tau}\, \big (e^{-s} \|(P-A)\, \varphi\|_2 +w(s) +z \|\varphi\|_2 \big )\, ds  = \wti{E}(\tau) + K_z \int_0^\tau e^{s-\tau}\, \wti{w}(s)\, ds,
\end{align*}
with 
\begin{align*}
	\wti{E}(\tau)
		\coloneqq K_z \int_0^\tau e^{s-\tau}\, \big (e^{-s} \|(P-A)\, 		\varphi\|_2 +z \|\varphi\|_2 \big )\, ds 
\end{align*}
and the second Gronwall--type bound 
from Lemma \ref{lem-gronwall} now gives 
\begin{equation}  \label{wt-uppb-2}
	\wti{w}(\tau) \leq \wti{E}(\tau) + K_z \int_0^\tau e^{(K_z-1)(\tau-s)} \, \wti{E}(s) \, ds  \, .
\end{equation}
Similarly as above one calculates 
\begin{align*}
	\int_0^\tau &e^{(K_z-1)(\tau-s)}\, \wti{E}(s)\, ds =\\
		&= K_z\iint\limits_{0<s<s'<\tau} e^{(K_z-1)\tau-K_zs'}
				 \, ds ds'\, \|(P-A)\, \varphi\|_2
			 +zK_z\iint\limits_{0<s<s'<t} e^{(K_z-1)\tau-K_zs'+s}
				\, \, ds ds' \, \|\varphi\|_2\\
		&= \Big(\frac{e^{-\tau}}{K_z}\big( e^{K_z\tau}-1) -\tau e^{-\tau} \Big)\|(P-A)\, \varphi\|_2 + z\Big( \frac{1}{K_z-1}\big( e^{(K_z-1)\tau} -1\big) -\big(1-e^{-\tau} \big) \Big)\|\varphi\|_2
\end{align*}
and 
\begin{align*}
	\wti{E}(\tau)
		= K_z\tau e^{-\tau}\|(P-A)\, \varphi\|_2 +zK_z(1-e^{-\tau}) \|\varphi\|_2\, .
\end{align*}
Plugging this back into \eqref{wt-uppb-2}, using $t=-\tau<0$ we arrive at 
\begin{align*}
	\|(e^t A- A_{-t} ) \varphi\|_2 = \wti{w}(\tau) 
		\le 
			 e^{t} \big( e^{K_z|t|}-1)\, \|(P-A)\, \varphi\|_2
			 + \frac{zK_z}{K_z-1}\big( e^{(K_z-1)|t|} -1\big) \|\varphi\|_2 \, .
\end{align*}
Recalling that we can replace $K_z$ by any upper bound in the 
above arguments, this proves \eqref{at-uppb}, we only have to bound 
$K_z$. Let $\psi \in C_0^\infty(\R^d)$. From the definition \eqref{kz} one easily gets 
\begin{align*}
  K_z  & =  \|\wti{B}\, R_z\|_{2\to2} \leq  \|\wti{B}\, ((P-A)^2 +dz^2)^{-\frac 12}\|_{2 \to 2} \, \|((P-A)^2 +dz^2)^{-\frac 12} (P-A-iz)\|_{2\to2} .
\end{align*}
On the other hand, letting 
$T=((P-A)^2 +dz^2)^{-\frac 12} (P-A-iz)$ one sees 
\begin{equation}
\begin{split}\label{eq-bound-great}
  TT^* &=  
  		((P-A)^2 +dz^2)^{-\frac 12} (P-A-iz)\cdot 
  			(P-A+iz)((P-A)^2 +dz^2)^{-\frac 12} \\
  		&= ((P-A)^2 +dz^2)^{-\frac 12} ((P-A)^2+dz^2)((P-A)^2 +dz^2)^{-\frac 12} =\id \, .
\end{split} 
\end{equation}
Hence by duality $\|((P-A)^2 +dz^2)^{-\frac 12} (P-A-iz)\|_{2\to2} =\|T\|_{2\to2}=1 $ and thus  
\begin{equation*}
K_z \, \leq \,  \| \wti{B} (H_0 +z^2)^{-\frac 12} \|_{2\to 2}\,  \eqqcolon   C_z\, .
\end{equation*}
\end{proof}
\noindent  Since we have defined the commutator $[H, i D]$ as the limit of $[H, i D_t]$, see \eqref{eq:comutator-form-2} and \eqref{Dt}, 
we have to calculate the terms appearing in the latter.
The next result concerns the calculation of 
$
	\frac{d}{dt}\big( e^t\, A -A_{-t} \big)\varphi  \big|_{t=0} \, 
$ 
for $\varphi\in\calD(P-A)$. 
Recall that given a magnetic field $B$, the vector field $\wti{B}$ is given by equation \eqref{B-tilde}. 

\begin{proposition} \label{prop-deriv}
Suppose that the magnetic field $B$ satisfies Assumption \ref{ass-B-mild-int}, the vector potential $A$ corresponding to $B$ is in the Poincar\'e gauge. Suppose moreover  that $\wti{B}^{\, 2}$ is relatively form bounded w.r.t.~$(P-A)^2$, i.e.~
$\wti{B}\, \cdot \in \mathcal{L}_c \big(\mathcal{D}(P-A), L^2(\R^d)\big)$. 
Then for all $\varphi\in \calD(P-A)$ the map $\R\ni t\mapsto ( e^t\, A -A_{-t})\varphi$ is differentiable and 
\begin{equation} \label{eq-deriv}
\frac{d}{dt}\big( e^t\, A -A_{-t} \big)\varphi  \big|_{t=0}= \lim_{t\to 0} \frac 1t ( e^t\, A -A_{-t})\varphi = \wti{B} \, \varphi  
\end{equation} 
where the limit is taken in $L^2(\R^d)$. 
\end{proposition}

\begin{proof} 
Assume that for $\varphi\in \calD(P-A)$ the map 
$t\mapsto \big( e^tA-A_{-t} \big)\varphi$ is differentiable 
in $t=0$ with derivative given  by \eqref{eq-deriv}. 
Then \eqref{eq-a-splitting-trick} shows that it is also differentiable in any point $t\in\R$ with derivative 
\begin{align}\label{eq-a-deriv-general}
	\frac{d}{dt}\big( e^t\, A -A_{-t} \big)\varphi 
		&=  \big(e^tA-A_{-t}\big)\varphi + U_t^*\wti{B}U_t \varphi
\end{align}
By assumption,  $\wti{B}:\calD(P-A)\to L^2(\R^d)$ is bounded. Thus 
the right hand side of \eqref{eq-a-deriv-general} is in $L^2(\R^d)$  
by Proposition \ref{prop-2}.
Hence it is enough to show differentiability at $t=0$. We will prove, for all $\varphi\in\calD(P-A)$,  
\begin{equation} \label{eq-equiv}
\lim_{t\to 0} \frac{1}{|e^t-1|} \big\| (e^t\, A -A_{-t})\, \varphi - (e^t-1)\, \wti{B}\, \varphi \big\| = 0 \quad \text{in} \ L^2(\R^d), 
\end{equation}
which is equivalent to \eqref{eq-deriv}. First assume that $\varphi \in C_0^\infty(\R^d)$. Using \eqref{eq-A-difference} we have 
\begin{align}
\delta_t := \big( e^t\, A -A_{-t} - (e^t-1)\, \wti{B} \big)\, \varphi = \int_0^t e^{t-s} \, U_s^* \,  \wti{B} \,  U_s \varphi \, ds -  (e^t-1)\, \wti{B} \, \varphi  =   \int_0^t e^{t-s} \,  \big(U_s^* \,  \wti{B} \,  U_s - \wti{B} \big ) \varphi\,  ds . \label{vt2}
\end{align} 
With \eqref{eq-key-dilation-inv} we rewrite the integrand on the right hand side as
\begin{align*}
\big(&U_s^* \,  \wti{B} \,  U_s - \wti{B} \big ) \varphi 
	 = U_s^* \wti{B}(U_s -1) \varphi + (U_s^*-1) \wti{B} \varphi \\
	  &= U_s^* \wti{B} R_z\big((P-A+iz) U_s -(P-A+iz)\big) \varphi 
		 + (U^*_s-1) \wti{B} \varphi \\
 	&= U_s^* \wti{B} R_z\Big [U_s \big( (e^s-1) (P-A)\varphi+(e^s-1) \wti{B} \varphi +\delta_s \big) + (U_s-1) (P-A+iz) \varphi \Big]
	  - U_s^*(U_s-1) \wti{B} \varphi \ .
\end{align*}
Setting  
$
w(s) \coloneqq \|\delta_s\|_2,
$ 
and recalling $\|\wti{B}R_z\|_{2\to2}\le \sqrt{d}\|\wti{B}((P-A)^2+z^2)^{-1/2}\|\eqqcolon C_z$,  we get
\begin{align*}
	\| \big(U_s^* \,  \wti{B} \,  U_s - \wti{B} \big ) \varphi\|_2  
		& \leq C_z \Big [ |e^s-1| \big(\| (P-A)\varphi\|_2 + \|\wti{B} \varphi\|_2\big) +w(s) +\|(U_s-1) (P-A+iz)\varphi\|_2\Big ] \\
		& \quad + \|(U_s-1) \wti{B} \varphi \|_2 .
\end{align*}
This implies the integral inequalities 
\begin{align*}
w(t) & \leq E(t) + C_z \int_0^t e^{t-s}\, w(s)\, ds \quad\text{for } t\ge 0\\
\intertext{and} 
w(t) & \leq E(t) + C_z \int_0^{|t|} e^{t+s}\, w(-s)\, ds \quad \text{for } t\le  0 ,
\end{align*}
where now 
\begin{align*}
E(t) & =  \int_0^{t} e^{t-s} \Big[ C_z \|(U_s-1) (P-A+iz)\varphi\|_2 + \|(U_s-1) \wti{B} \varphi \|_2\Big] \, ds\\
& \quad +\int_0^t e^{t-s}\, (e^s-1) C_z\big(\| (P-A)\varphi\|_2 + \|\wti{B} \varphi\|_2\big)\, ds \, , 
\intertext{for $t\ge 0$, and} 
E(t) & =  \int_0^{|t|} e^{t+s} \Big[ C_z \|(U_s-1) (P-A+iz)\varphi\|_2 + \|(U_s-1) \wti{B} \varphi \|_2\Big] \, ds\\
& \quad +\int_0^{|t|} e^{t+s}\, (1-e^{-s}) C_z \big(\| (P-A)\varphi\|_2 + \|\wti{B} \varphi\|_2\big)\, ds \, ,
\end{align*}
for $t\le  0$. 
Lemma \ref{lem-gronwall} then yields the upper bounds 
\begin{align} 
	w(t) &\leq E(t) + C_z \int_0^t e^{(1+C_z)(t-s)}\, E(s)\, ds  \quad \text{for }t\ge 0 \label{w-uppb2} \\
\intertext{and} 
 \label{w-uppb2-neg-t}
	w(t) &\leq E(t) + C_z \int_0^{|t|} e^{(C_z-1)(t-s)}\, E(-s)\, ds  \quad \text{for }t\le 0 \, .
\end{align} 
To continue it is convenient to use, for $\tau\ge 0$,  
\begin{align*} 
\kappa(\tau) : =  \sup_{ |s|\leq \tau} \|(U_s-1) \wti{B} \varphi \|_2 + C_z \sup_{|s|\leq \tau} \|(U_s-1) (P-A+iz) \varphi \|_2  ,
\end{align*}
so that for $t\ge 0$ 
\begin{align*}
E(t) & \leq \int_0^t e^{t-s}\, \kappa(s)\, ds + \big(\| (P-A)\varphi\|_2 + \|\wti{B} \varphi\|_2\big) \int_0^t e^{t-s}\, (e^s-1) \, ds\,  \\
& \leq \kappa(t) (e^t-1) + \big(\| (P-A)\varphi\|_2 + \|\wti{B} \varphi\|_2\big)  (e^t-1)^2,
\end{align*} 
since $\kappa$ is increasing. Analogously, for $t\le 0$ we have 
\begin{align*}
	E(t) & =  \int_0^{|t|} e^{t+s} \kappa(s)\, ds
				+\big(\| (P-A)\varphi\|_2 + \|\wti{B} \varphi\|_2\big)\int_0^{|t|} e^{t+s}\, (1-e^{-s}) \, ds \, \\
		&\le \kappa(|t|)(1-e^{t}) 
			+\big(\| (P-A)\varphi\|_2 + \|\wti{B} \varphi\|_2\big)
				(1-e^{t})^2\, . 
\end{align*}
So by monotonicity, for $t\ge 0$,  
\begin{align*}
	\int_0^t e^{(1+C_z)(t-s)}\, E(s)\, ds 
		& \leq \Big( \kappa(t) (e^t-1) + \big(\| (P-A)\varphi\|_2 + \|\wti{B} \varphi\|_2\big) (e^t-1)^2 \Big) \int_0^t e^{(1+C_z)(t-s)}\, ds \\
				&= \frac{(e^{(1+C_z) t}-1) (e^t-1) }{1+C_z} \Big[\kappa(t) +  (e^t-1) \big(\| (P-A)\varphi\|_2 + \|\wti{B} \varphi\|_2\big)\Big] 
\end{align*} 
and, similarly, for $t\le 0$ we have   
\begin{align*}
	\int_0^{|t|} e^{(C_z-1)(t-s)}\, E(-s)\, ds 
		\le  \frac{(1-e^{(C_z-1) t}) (1-e^{t}) }{C_z-1} \Big[\kappa(|t|) +  (1-e^{t}) \big(\| (P-A)\varphi\|_2 + \|\wti{B} \varphi\|_2\big)\Big] 
\end{align*}
which in combination with \eqref{w-uppb2} and \eqref{w-uppb2-neg-t}  implies 
\begin{align}
\frac{w(t)}{|e^t-1|} & = \Big\| \frac{e^t\, A -A_{-t}}{e^t-1} \varphi -\wti{B} \varphi \Big\| \nonumber \\
&  \leq \label{eq-deriv-final-bound}
		\Big(1+ \frac{C_z|e^{(C_z\pm 1) t}-1|}{C_z\pm 1} \Big) \Big[\kappa(|t|) +  |e^t-1| \big(\| (P-A)\varphi\|_2 + \|\wti{B} \varphi\|_2\big)\Big] ,
\end{align}
where the $+$ sign holds when $t\ge 0$ and the $-$ sign when $t<0$. 
Since $P-A:\calD(P-A)\to L^2(\R^d)$ and $\wti{B}:\calD(P-A)\to L^2(\R^d)$ 
are bounded, \eqref{eq-deriv-final-bound} extends to all 
$\varphi\in \calD(P-A)$, by density. 
Since $\kappa(t) \to 0$ as $t\to 0$, this proves \eqref{eq-equiv}. 
\end{proof}

\noindent We will need a version Proposition \ref{prop-deriv} for the electric potential. Recall that  $iD_t=(U_t-U_{-t})/(2t)$, cf.~\eqref{Dt}. 
\begin{lemma} \label{lem:xgradv}
Let $A$, $B$, and $\wti{B}$ satisfy the same assumptions as in Proposition \ref{prop-deriv} and let  $V$ be any electric potential, with form domain $\calD(P-A)\subset\calQ(V)$, such that the distribution $x\cdot \nabla V$ extends to a quadratic form $q_{x\cdot\nabla V}$ 
which is form bounded with respect to $(P-A)^2$. 
Then  with $V_{-t}= U_t^* V U_t= V(e^{-t}\cdot)$ and $q_{V}$,     
respectively,  $q_{V_{-t}}$, the quadratic form  
corresponding to $V$, respectively, $V_{-t}$, we have 
\begin{align}
	\label{eq-deriv-2}
  \lim_{t\to 0} \frac{1}{t} 
    \left( q_{V}(\varphi,\psi)- q_{V_{-t}}(\varphi,\psi)\right)
  =  q_{x\cdot\nabla V}(\varphi,\psi)
\end{align}
and
\begin{align}
 \label{eq-deriv-3}
  \lim_{t\to 0} 2\re q_{V}(\varphi, iD_t\varphi)
  =  -q_{x\cdot\nabla V}(\varphi,\varphi) 
\end{align}
for all $\varphi,\psi\in\calD(P-A)$. 
\end{lemma} 

\begin{proof}
 We always have $V\in L^1_{\text{loc}}(\R^d)$. 
  Since $V_{-t}= U_t^*VU_t$, we have the identity 
  $q_{V_{-t}}(\varphi,\psi)= 
  \La \varphi, V_{-t}\psi\Ra 
  = \La U_t\varphi, VU_t\psi\Ra = q_{V}(U_t\varphi,U_t\psi)$.  
  
\noindent If $V$ is a nice differentiable function, then 
  $\frac{d}{dt}V_{-t} = -e^{-t}x\cdot\nabla V(e^{-t}\cdot) 
  = -e^{-t}U_t^* (x\cdot V) U_t $
  so 
  \begin{align}\label{eq-nice functions}
  	\frac{d}{dt}\La\varphi, V_{-t}\psi\Ra 
  = -e^{-t}\La U_t\varphi, x\cdot\nabla V U_t\psi\Ra
  =  -e^{-t} q_{x\cdot\nabla V}(U_t\varphi, U_t\psi)\, .
  \end{align}
  Given $\varphi\in\calC^\infty_0(\R^d)$, the map 
  $\calC^\infty_0(\R^d)\ni\psi\mapsto q_{V_{-t}}(\varphi,\psi)$ yields a distribution.  
  Approximating $V$ in $L^1_{\rm{loc}}$ by 
  $\calC^\infty_0$ functions and using  \eqref{eq-nice functions} shows that  the distributional derivative 
  $W_t\coloneqq \frac{d}{dt}V_{-t}$  
  is given by 
  \begin{align*}
  	\la \varphi, W_t \psi \ra 
  		= \frac{d}{dt} \la \varphi, V_{-t}\psi \ra 
  		= -e^{-t}q_{x\cdot\nabla V}(U_t\varphi,U_t\psi .
  \end{align*}
  for all 
  $\varphi,\psi\in \calC^\infty_0(\R^d)$,  
  with $q_{x\cdot\nabla V}$ the sesqui--linear form 
  corresponding to the distribution $x\cdot\nabla V$.  
  By assumption, the sesqui--linear form $q_{x\cdot\nabla V}$ extends to  
  sesqui--linear form, again denoted by $q_{x\cdot\nabla V}$, 
  which is relatively form bounded with respect to 
  $(P-A)^2$. 
We claim that for any $\varphi,\psi\in\calD(P-A)$ the map 
  \begin{align}\label{eq-virial-V-cont}
  	\R\ni s\mapsto q_{x\cdot\nabla V}(U_s\varphi, U_s\psi)  \quad \text{is continuous}. 
  \end{align}  
  Assuming this for the moment,  the fundamental theorem of calculus shows 
  \begin{align}\label{eq-virial-V-1}
  	q_V(\varphi,\psi)- q_{V_{-t}}(\varphi,\psi) 
  		= -\int_0^t \frac{d}{ds} q_{V_{-t}}(\varphi,\psi)\, ds 
  		= \int_0^t e^{-s}q_{x\cdot\nabla V}(U_s\varphi, U_s\psi)\, ds 
  \end{align}
  for any 
  $\varphi,\psi\in \calC_0^\infty(\R^d)\subset\calD(P-A)$. Since 
  $\calC^\infty_0(\R^d)$ is dense in $\calD(P_A)$ with respect to the 
  graph norm and the involved quadratic forms are form bounded w.r.t\ 
  $(P-A)^2$,  equation \eqref{eq-virial-V-1} extend to all 
  $\varphi,\psi\in \calD(P-A)$. 
  But then \eqref{eq-virial-V-1} implies 
  \begin{align*}
  	\frac{d}{dt} q_{V_{-t}}(\varphi,\psi)|_{t=0} 
  		= \lim_{t\to 0} \frac{1}{t} 
  		 		\left( 
  		 		   q_{V}(\varphi,\psi)-q_{V_{-t}}(\varphi,\psi)
  		 		 \right) 
  		= q_{x\cdot\nabla V}(\varphi,\psi)
  \end{align*}
  which proves \eqref{eq-deriv-2}. For \eqref{eq-deriv-3} we note 
  \begin{align*}
  	2t\re q_{V}(\varphi, iD_t\varphi)  
  		&= \re\big(q_{V}(\varphi, U_t\varphi)- q_{V}(\varphi, U_{-t}\varphi)\big) 
  		  = \re\big(q_{V}(U_t\varphi, \varphi)- q_{V}(\varphi, U_{-t}\varphi)\big) \\
  		&=  \re\big(q_{V_{-t}}(\varphi, U_{-t}\varphi)- q_{V}(\varphi, U_{-t}\varphi)\big)
  \end{align*}
  and 
   \begin{align*}
     q_{V_{-t}}(\varphi, U_{-t}\varphi)- q_{V}(\varphi, U_{-t}\varphi
     =  - \int_0^t e^{-s}q_{x\cdot\nabla V}(U_s\varphi, U_sU_{-t}\psi)\, ds 
  \end{align*}
  again by \eqref{eq-virial-V-1}. By a simple continuity argument this shows 
  \begin{align*}
  	2\re\la \varphi, V iD_t\varphi\ra 
  		&= -\frac{1}{t} \int_0^t e^{-s}\re q_{x\cdot\nabla V}(U_s\varphi, U_sU_{-t}\psi)\ \, ds 
  			\to  q_{x\cdot\nabla V}(\varphi,\varphi) 
  \end{align*}
  as $t\to 0$, which yields \eqref{eq-deriv-3}.

\smallskip

\noindent  It remains to prove \eqref{eq-virial-V-cont}: The sesqui--linear form $q$ being 
  relatively $(P-A)^2$ form bounded is equivalent to the fact that the sesqui-linear form 
  \begin{align*}
  	\varphi,\psi\mapsto q(\big( (P-A)^2 +dz^2\big)^{-1/2}\varphi, \big( (P-A)^2 +dz^2\big)^{-1/2}\psi)
  \end{align*}
  extends, for $z>0$, to a bounded sesqui--linear form to all 
  $(\varphi,\psi)\in L^2(\R^d)$, see e.g.~\cite{teschl}. Recalling the definition \eqref{eq-Rz} 
  for $R_z$ and \eqref{eq-bound-great}, this is equivalent to 
  \begin{align*}
  	\varphi,\psi\mapsto q(\big( R_z\varphi, R_z\psi)\eqqcolon \wti{q}(\varphi,\psi)
  \end{align*}
  being a bounded quadratic form, more precisely, extending to a bounded quadratic form on all of $L^2(\R^d)$, for all $z>0$. 
  Using sesqui--linearity, it is easy to see that for all 
  continuous maps 
  $s\mapsto \varphi_s, s\mapsto\psi_s\in L^2(\R^d)$ the map 
  $s\mapsto \wti{q}(\varphi_s,\psi_s)$ is continuous for any 
  bounded sesqui--linear form $\wti{q}$ on $L^2(\R^d)$. 
  
  \smallskip
  
  \noindent For  $\varphi,\psi\in\calD(P-A)$ we have 
  \begin{align*}
  	q(U_s\varphi, U_s\psi) = \wti{q}((P-A-iz)U_s\varphi, (P-A-iz)U_s\psi)
  \end{align*}
  and 
  \begin{align*}
  	U_s\varphi
  		&= R_z(P-A-iz)U_s\varphi 
  			=  R_zU_s\big(e^s(P-A) +(e^sA-A_{-s})-iz)  \big)\varphi \, .
  \end{align*}
  The map $s\mapsto e^s(P-A)\varphi$ is clearly continuous for all $\varphi\in \calD(P-A)$ and so is the map $s\to (e^sA- A_s)\varphi$
 by Proposition \ref{prop-2}. Thus $s\mapsto \wti{\varphi}_s\coloneqq\big(e^s(P-A) +(e^sA-A_{-s})-iz)  \big)\varphi$ is continuous for all $\varphi\in\calD(P-A)$. Using $s\mapsto U_s$ being strongly continuous and unitary, and 
 \begin{align*}
 	U_t \wti{\varphi}_t - U_s\wti{\varphi}_s = (U_t-U_s)\wti{\varphi}_t + U_s(\wti{\varphi}_t-\wti{\varphi}_s)
 \end{align*}
 one sees that the map $s\mapsto \varphi_s\coloneqq U_s\wti{\varphi}_s$ is continuous. Similarly when 
 $\varphi$ is replaced by $\psi\in\calD(P-A)$. Thus 
  \begin{align*}
  	\R\ni s\mapsto q(U_s\varphi, U_s\psi) 
  		&= 	\wti{q}(\varphi_s,\psi_s)
  \end{align*}
  is continuous, since $\wti{q}$ is a bounded sesqui--linear form. This proves \eqref{eq-virial-V-cont} and hence the lemma. 
\end{proof}


\subsection{The commutator  as a quadratic form} 
This section deals with one of our main results, 
the rigorous identification of the right hand side of \eqref{eq:comutator-form-2}.
\begin{theorem}[Magnetic virial theorem] \label{thm:magnetic-virial} 
Let $B$ and $V$ satisfy assumptions \ref{ass-B-mild-int}- \ref{ass-V-form small} and $A$ 
be the vector potential in the Poincar\'e gauge corresponding to the 
magnetic field $B$. Assume also that the distribution $x\cdot\nabla V$ 
extends to a quadratic form which is form bounded with respect to 
$(P-A)^2$.  
Then for all $\varphi\in \calD(P-A)$, the limit $\lim_{t\to0}2\re\big(q_{A,V}(\varphi, i D_t \varphi)\big) $ exists. Moreover,
\begin{equation}  \label{eq-hd}
	\begin{split}
		\La\varphi, [H, iD] \, \varphi \Ra 
			&\coloneqq \lim_{t\to 0} 2\re \big(q_{A,V}(\varphi, iD_t \varphi)\big) 
				= 2 \|(P-A) \varphi\|_2^2  + 2\, \re \La\wti{B}\varphi,(P-A)\, \varphi\Ra- \La \varphi, x\cdot\nabla V\varphi \Ra \, .
	\end{split}	
\end{equation}
\end{theorem}

\begin{proof}
Recall that, as a quadratic form,  we defined 
$\la \varphi, [H_{A,V},iD_t] \varphi\ra \coloneqq 2\re q_{A,V}(\varphi, iD_t\varphi)$, using the notation from \eqref{sesq-f}.  
See \eqref{eq:commutator-form-1} and the discussion 
before it. Once one knows this limit, and its existence, the proof is straightforward. 
Since $\psi$ is a weak eigenfunction we also have 
$\La E\psi,\varphi\Ra = q_{A,V}(\psi,\varphi)$ for all 
$\varphi\in\calD(P-A)$. Since multiplication with 
$E\in\R$ and $iD_t$ are bounded operators and 
multiplication with a constant commutes 
with any bounded operator,   
$2\re(q_{A,V}(\psi, iD_t\psi)) = 2\re(\La E\psi, iD_t\psi) 
= 2\re \La \psi, [E,iD_t]\psi\Ra =0 $.

Now we will show that the limit in \eqref{eq-hd}  exists for all 
$u\in \calD(P-A)$ and is given by the right hand side of 
\eqref{eq-hd}.  By  \eqref{eq-key-dilation-inv} 
\begin{equation} \label{kom-1}
(P-A) U_t\, u = e^tU_t (P-A)\, u + X_t\, u,
\end{equation} 
where 
\begin{equation} \label{xs}
X_t\, u =  U_t (e^t\, A-A_{-t})\, u \, ,
\end{equation}
where we recall $A_{-t} = U_{t}^* \, A\, U_t = A (e^{-t} \cdot) $.
Since 
\begin{align*}
	2t \re \big( q_{A,V}(\varphi, iD_t\varphi) \big) 
		= \re \big( q_{A,V}(\varphi, U_t\varphi) - q_{A,V}(U_{-t}\,\varphi,\varphi) \big) \, , 
\end{align*}
and 
\begin{align*}
	\begin{split}
		q_{A,0}(\varphi, U_t\varphi)
			&= \la (P-A)\varphi, U_te^t(P-A)\varphi \ra
				+ \la (P-A)\varphi, X_t\varphi\ra \, ,\\
		q_{A,0}(U_{-t}\varphi,\varphi)
			&= \la (P-A)\varphi, U_{t}e^{-t}(P-A)\varphi \ra
				+ \la X_{-t}\varphi, (P-A)\varphi\ra \, ,
	\end{split}
\end{align*}
we get 
\begin{align*}
	2\re q_{A,0}(\varphi, iD_t\varphi) 
		&= \frac{e^t-e^{-t}}{t} \la (P-A)\varphi, U_t (P-A)\varphi\ra 
			+ \la (P-A)\varphi,\frac{1}{t} X_t\varphi \ra
			- \la \frac{1}{t}X_{-t}\varphi, (P-A)\varphi \ra \\
		&\to 2\la (P-A)\varphi, (P-A)\varphi\ra 
			+ 2\re \la \wti{B}\varphi, (P-A)\varphi \ra
\end{align*}
as $t\to 0$, because by Proposition \ref{prop-deriv} we have 
$$
\lim_{t\to 0} \frac{1}{t}X_{\pm t}\,  u  = \pm\wti{B} u \quad \text{in} \ \ L^2(\R^d). 
$$
Lemma \ref{lem:xgradv} gives  
$\lim_{t\to0}\re \la \varphi, ViD_t\varphi\ra= -\la \varphi, x\cdot\nabla V\varphi\ra$ and since 
$$
q_{A,V}(\varphi, iD_t\varphi) = q_{A,0}(\varphi, iD_t\varphi)+ \la \varphi,ViD_t\varphi\ra\, ,
$$ 
this finishes the proof.
\end{proof}

\begin{remark} \label{rem-virial}
Equation \eqref{eq-hd} is known for smooth magnetic and electric fields, see e.g.~to \cite{alb}.
As for its physical  interpretation, we note that the classical virial states that 
\begin{equation} \label{virial-class}
2 \La T \Ra = - \La x\cdot F\Ra,
\end{equation} 
where $T$ denotes the kinetic energy, $F$ denotes the external force, and $\La\,  \Ra$ stands for the average over time. 
The identity \eqref{virial-class} holds for all initial conditions for which the velocity and position of the system stay bound in time, i.e., the classical version of a bound state.
 In our case $F$ is given by the 
Lorentz force, hence $F = -q \nabla V +q v \wedge B$, and therefore  
$$
x\cdot F = -q x\cdot \nabla V + q\, x\cdot (v \wedge B) =  -q\, x\cdot \nabla V + q\,  v\cdot (B\wedge x ) = -q\, x\cdot \nabla V + q\,  v\cdot \wti B,
$$
where we have used the vector identity $a\cdot (b\wedge c)=b\cdot (c\wedge a)$.
Since we have $v = \frac 1m (P-q A)$ and $T=\frac{1}{2m}(P-qA)^2$, the quantum analog of \eqref{virial-class} reads 
\begin{align*}
0=\frac{1}{m} \|(P-qA) \varphi\|_2^2 + \frac{q}{m} \,  \re  \La(P-A)\, \varphi, \wti{B}\varphi\Ra- \La \varphi, x\cdot\nabla V\varphi \Ra ,
\end{align*}
which in our system of units, where $q=1$ and $m=\frac 12$, coincides with  \eqref{eq-hd} when the commutator vanishes. 
\end{remark}

 \noindent  An immediate consequence of our magnetic virial theorem is 
\begin{corollary} 
Let the assumptions of the magnetic virial Theorem \ref{thm:magnetic-virial} above be satisfied. If $\psi\in \calD(P-A)$ is a normalised weak  eigenfunction of the magnetic Schr\"odinger operator $H_{A,V}$ corresponding to the energy $E\in \R$, in the sense that $\|\psi\|_2=1$, and 
	\begin{align}\label{eq:weak-eigfcn-1}
		E\la \varphi, \psi\ra = q_{A,V} (\varphi, \psi) 
	\end{align} 
	for all $\varphi\in \calD(P-A)$, or all $\varphi\in\calC^\infty_0(\R^d)$,  then 
	\begin{align}
		0 = 2E + 2\, {\rm Re}\, \La(P-A)\, \psi, \wti{B}\,  \psi\Ra  -  \La\psi, (2V+ x\cdot \nabla V)  \psi\Ra
	\end{align}
\end{corollary}

Now, of course, the question is for what class of potentials $V$ one 
can calculate the virial $x\cdot\nabla V$ in a simple way. 
If $x\cdot\nabla V$ is given by a function which yields a nice quadratic form, 
then $\la \varphi, x\cdot\nabla V\varphi \ra$ is given by the 
classical expression. 
On the other hand, the virial $\la\varphi, x\cdot\nabla V\varphi\Ra$ exists even if $V$ is not at all classically differentiable. 
A typical example is given in the next section.

\subsection{The Kato form of the  virial }\label{sec:virial-Kato-form} 
Our standing assumption is that the virial of the 
potential, given by the distribution 
$x\cdot \nabla V$, yields a quadratic form 
$q_{x\cdot\nabla V}$ which is form bounded w.r.t. $(P-A)^2$. 
If $x\cdot\nabla V$ is given by a function which corresponds 
to a nice quadratic form, then 
$q_{x\cdot\nabla V}$ is given by the classical expression 
$\la \varphi, x\cdot\nabla V\varphi \ra$. 
On the other hand, the virial given by the formal expression 
$\la\varphi, x\cdot\nabla V\varphi\Ra$ can exist even if 
$V$ is not at all classically differentiable. 

Our next result shows that this can be the case, even 
without any kind of differentiability of $V$. 
Lemma \ref{lem-kato-virial} result also identifies the 
quadratic form 
$q_{x\cdot\nabla V}$ with an expression similar to one 
already used by Kato in his proof of absence of 
positive eigenvalues.

\begin{lemma}\label{lem-kato-virial}
	Assume that the magnetic field $B$ satisfies assumptions \ref{ass-B-mild-int} and \ref{ass-B-rel-bounded}, $A$ is the magnetic vector-potential in the Poincar\'e gauge, and $V$ and  $|x|^2V^2$ are relatively form bounded with respect to $(P-A)^2$. Then the quadratic form corresponding to the distribution $x\cdot\nabla V$ is given by  
	\begin{equation}\label{eq:kato-virial}
		-\La \varphi , x\cdot\nabla V\varphi \Ra = -2\im \La xV\varphi, (P-A)\varphi \Ra  + d\La\varphi, V\varphi \Ra
	\end{equation}
	for all $\varphi\in\calD(P-A)$. 
\end{lemma}
\begin{remark}  Since $x^2V^2$ is form bounded w.r.t.\ $(P-A)^2$, 
  $|x|V\varphi\in L^2$ for all $\varphi\in \calD(P-A)$.
	We call \eqref{eq:kato-virial} the Kato form of the virial. 
	Kato did not consider magnetic fields and used the  pointwise conditions 
	$V$ bounded and $\lim_{x\to\infty} |x|V(x)=0$ to conclude absence 
	of positive eigenvalues for non-magnetic Schr\"odinger operators. 
	Lemma \ref{lem-kato-virial} allows us not only to extend this to magnetic Schr\"odinger operators but to replace Kato's pointwise condition  
	by a rather weak and natural smallness condition on the quadratic form 
	$\la \varphi, |x|^2V^2\varphi \ra$   at infinity.
	
	Of course, since the vector potential is in the Poincar\'e gauge $x\cdot A(x)=0$, so 
	$ \La xV\varphi, (P-A)\varphi \Ra=  \La V\varphi, x\cdot P\varphi \Ra$, hence  
	the right hand side of  \eqref{eq:kato-virial} does not depend on vector potential $A$. In fact, since $A$ is a real-valued vector function and $V$ is real-valued $ \La xV\varphi, A\varphi \Ra$ is real for any function 
		$\varphi\in \calC^\infty_0(\R^d)$. 
		Keeping $P-A$ in the right hand side of \eqref{eq-kato-form} is useful, however, see in particular the proof of \eqref{eq-punch-2-5}.
\end{remark}
\begin{proof}
  By definition, the virial is given by $-\la \varphi , x\cdot\nabla V\varphi \ra= \lim_{t\to 0}\re\la \varphi,ViD_t\varphi\ra$. We will calculate this limit  slightly differently than in Lemma \ref{lem:xgradv}. As distributions 
  \begin{align*}
  	2it D_t\varphi = \int_{-t}^t U_s iD\varphi\, ds 
  		= \int_{-t}^t U_s ix\cdot P\varphi\, ds + \frac{d}{2}\int_{-t}^t U_s \varphi\, ds 
  \end{align*}
  and 
  \begin{align*}
  	\frac{1}{|x|} \int_{-t}^t U_s ix\cdot P\varphi\, ds
  		= i \int_{-t}^t e^s U_s \big(\tfrac{x}{|x|}\cdot P\varphi\big)\, ds
  		= i \int_{-t}^t e^s U_s \big(\tfrac{x}{|x|}\cdot (P-A)\varphi\big)\, ds
  \end{align*}
  since any vector potential in the Poincar\'e gauge is transversal, 
  that is,  $x\cdot A(x)=0$ for all $x\in\R^d$. 
   Altogether, we have 
  \begin{align*}
	iD_t\varphi 
		&= \frac{i}{2t} |x|\int_{-t}^t e^s U_s \big(\tfrac{x}{|x|}\cdot (P-A)\varphi\big)\, ds\;  
			+   \frac{d}{4 t} \int_{-t}^t U_s \varphi\, ds
  \end{align*}
 at least when $\varphi\in\calC^\infty_0(\R^d)$. Thus, in this case, 
 \begin{align}\label{eq-nice}
 	\la \varphi, ViD_t\varphi \ra
 		&= i \La|x|V\varphi, \frac{1}{2t}\int_{-t}^t e^s U_s \big(\tfrac{x}{|x|}\cdot (P-A)\varphi\big)\, ds \Ra
 			+ \frac{d}{2}\la V\varphi, \frac{1}{2t}\int_{-t}^t U_s\varphi\, ds \ra\, .
 \end{align}
 Since $\tfrac{x}{|x|}\cdot(P-A)\varphi\in L^2(\R^d)$ for all $\varphi\in\calD(P-A)$, the maps $s\mapsto U_s(\tfrac{x}{|x|}(P-A)\varphi)$ and 
 $s\mapsto U_s\varphi$ are continuous. Moreover, 
 the map $s\mapsto U_s\varphi$ is continuous in the graph norm corresponding to $P-A$ for any $\varphi\in\calD(P-A)$ by a similar argument as in the proof of Lemma \ref{lem:xgradv}. 
 Also $|x|V\varphi\in L^2(\R^d)$ for any $\varphi\in\calD(P-A)$, since $xV$ is relatively $P-A$ bounded, that is, $|x|^2V^2$ is relatively $(P-A)^2$ form bounded, by assumption. 
 But then \eqref{eq-nice} also extends to all 
 $\varphi\in\calD(P-A)$ by continuity.  
 
 Since for $\varphi\in\calD(P-A)$ the map $s\mapsto U_s$ is continuous in the graph norm of $P-A$, we also have $\tfrac{1}{2t}\int_{-t}^t U_s\varphi\, ds\to \varphi$ in the graph norm. In addition,  
 $\tfrac{1}{2t}\int_{-t}^t e^s U_s \big(\tfrac{x}{|x|}\cdot (P-A)\varphi\big)\, ds \to \tfrac{x}{|x|}\cdot (P-A)\varphi$ in $L^2(\R^d)$ as $t\to 0$. Then \eqref{eq-nice} yields  
 \begin{align*}
 	\lim_{t\to 0} \la \varphi , ViD_t\varphi \ra 
 		&= i\la |x|V\varphi, \tfrac{x}{|x|}\cdot (P-A)\varphi \Ra
 			+ \frac{d}{2}\la V\varphi, \varphi \ra\, 
 			= i\la xV\varphi,  (P-A)\varphi \Ra
 			+ \frac{d}{2}\la V\varphi, \varphi \ra
 \end{align*}
  which, taking real parts, finishes the proof of Lemma 
  \ref{lem-kato-virial}.
\end{proof}

\begin{remark}
  Slightly informally, an alternatively way to  
  derive \eqref{eq:kato-virial} is as follows: 
  For $u,w\in\calC^\infty_0(\R^d)$, which is dense in the domain of 
  $P-A$, the quadratic form $\La u , x\cdot\nabla V w \Ra $ is given as a distribution by  
  \begin{equation}\label{eq-virial-locality}
   \begin{split}
  	\La u , x\cdot\nabla Vw \Ra 
  		&= \La u , x\cdot\nabla (V w) - Vx\cdot\nabla w\Ra 
  			= -\La \nabla\cdot (x u) ,  Vw\Ra
  			   - \La V u , x\cdot\nabla w\Ra \\
  		&=  -d\La u , V w\Ra  
  			- \La V u ,x\cdot\nabla w \Ra 
  				-\La x\cdot \nabla u , V w\Ra  \\
  		&=   -d\La u , V w\Ra  
  			- i\big(\La xV u ,(P-A) w \Ra 
  				-\La (P-A) u , xV w\Ra\big)  
  	\end{split}
  \end{equation}
  since the vector potential $A$ is in the Poincar\'e gauge and $P=-i\nabla$. 
  Under the conditions on $V$ this extends to all 
  $\varphi\in \calD(P-A)$. 
\end{remark}

\begin{corollary}\label{cor-mixed-virial}
  Assume that the magnetic field $B$ satisfies assumptions \ref{ass-B-mild-int} 
  and \ref{ass-B-rel-bounded}, $A$ is the magnetic vector-potential in the 
  Poincar\'e gauge, and the potential $V$ splits as $V=V_1+V_2$ 
  where  $V_1$ and  $|x|^2V_1^2$ are relatively form bounded with respect to 
  $(P-A)^2$ and the distribution $ x\cdot\nabla V_2$ extend to a quadratic form which is form bounded with respect to $(P-A)^2$. Then 
  \begin{equation}
  	-\La \varphi , x\cdot\nabla V\varphi \Ra
  		= -2\im \La xV_1\varphi, (P-A)\varphi \Ra  + d\La\varphi, V_1\varphi \Ra
  			- \La \varphi , x\cdot\nabla V_2\, \varphi \Ra
  \end{equation}
 for all $\varphi\in\calD(P-A)$. 
\end{corollary}
\begin{proof} 
Simply combine Lemma \ref{lem:xgradv} and Lemma \ref{lem-kato-virial}. 
\end{proof}

\subsection{The exponentially weighted magnetic virial}
 
The proof of our main result, see Theorem \ref{thm-abs} below, is based on finding two different expressions for the commutator $\la e^F\psi, [H, iD]e^F\psi\ra$, when $F$ is a suitable weight function and $\psi$ is a weak eigenfunction, see \eqref{eq:weak-eigfcn-1}. This is done in  
\begin{lemma} \label{lem-com-ef}
  Assume that the magnetic field $B$ and the electric potential $V$ satisfy assumptions \ref{ass-B-mild-int}, \ref{ass-B-rel-bounded}, 
  and \ref{ass-V-form small}, and $A$ is the vector potential corresponding to $B$ in the Poincar\'e gauge. 
  Moreover assume that the distribution $x\cdot\nabla V$ extend to a quadratic form, which is form bounded with respect to $(P-A)^2$. 
  Let $F:\R^d\to \R$ be a smooth and 
  bounded radial function, such that  $\nabla F(x) = g(x)x$, and assume 
  that $g\ge0$ and  
  that the functions $\nabla(|\nabla F|^2)$, $(1+|\cdot|^2) g$,  $x\cdot\nabla g$ and 
  $(x\cdot\nabla)^2g$ are bounded. 
Let $\psi\in\calD(P-A)$ be a weak eigenfunction of the 
  magnetic Schr\"odinger operator $H_{A,V}$, i.e., 
  $E\la \varphi,\psi\ra= q_{A,V}(\varphi,\psi)$ for some $E\in\R $ and 
  all $\varphi\in\calD(P-A)$, where $q_{A,V}$ is the sesqui--linear form corresponding to the magnetic Schr\"odiner operator $H_{A,V}$ and set $\psi_F = e^F\, \psi$. Then 
\begin{align}
\La\psi_F, i \, [H, D] \, \psi_F \Ra& =   \La\psi_F, \big(E+|\nabla F|^2\big)\, \psi_F \Ra+ 2\, {\rm Re}\,  \La(P-A)\, \psi_F, \wti{B}\,  \psi_F\Ra
-2\, {\rm Im}\,  \La(P-A)\, \psi_F,  x V_1\, \psi_F\Ra \nonumber \\ 
& \quad + \|(P-A)\psi_F\|_2^2 \, + \, \La\psi_F, (d\,  V_1-V ) \psi_F\Ra- \La\psi_F,  x\cdot \nabla V_2\, \psi_F\Ra, \label{eq-psiF-1}
\end{align} 
and 
\begin{align}
\La\psi_F, i \, [H, D] \, \psi_F \Ra& = -4 \, \|\sqrt{g}\, D\, \psi_F\|_2^2 + \La\psi_F, \big((x\cdot \nabla)^2 g -x\cdot \nabla |\nabla F|^2  \big) \psi_F \Ra    \, .\label{eq-psiF-2} 
\end{align} 
\end{lemma}

\smallskip

\begin{remark}
 Of course, $\La\varphi,  (V+x\cdot \nabla V)\, \varphi\Ra$ 
is given by the sum $ q_{V}+ q_{x\cdot\nabla V} $ of the quadratic forms.
Rearranging the terms in the derivation of \eqref{eq-psiF-1}
a little bit also shows that 
\begin{align*}
\La\psi_F, i \, [H, D] \, \psi_F \Ra& =   \La\psi_F, 2\big(E+|\nabla F|^2\big)\, \psi_F \Ra + 2\, {\rm Re}\,  \La(P-A)\, \psi_F, \wti{B}\,  \psi_F\Ra
	- \La\psi_F,  (2V+x\cdot \nabla V)\, \psi_F\Ra\, .
\end{align*}
  Thus \eqref{eq-psiF-1} and \eqref{eq-psiF-2}
  are a quadratic form version of the bounds of  \cite{fhhh}, 
  in which the authors considered only the nonmagnetic case.

  To get an idea why the bounds from Lemma \ref{lem-com-ef} are useful for  
  excluding eigenfunctions for positive energies $E>0$, think of 
  $\La\psi_F, \big(E+|\nabla F|^2\big)\, \psi_F \Ra $, respectively 
  $ -4 \, \|\sqrt{g}\, D\, \psi_F\|^2$, as the main terms in 
  \eqref{eq-psiF-1} and \eqref{eq-psiF-2}, and the other terms as lower 
  order. Then  \eqref{eq-psiF-1} and \eqref{eq-psiF-2} contradict each other when $E>0$ unless $\psi=0$. 
\end{remark}

\noindent Before we prove Lemma \ref{lem-com-ef}
 we first collect some auxiliary results. 
 First note that as distributions, 
\begin{equation} \label{eq-aux-1}
  (P-A)\psi_F = e^F(P-A)\psi -ie^F\nabla F \psi\, .
\end{equation}
Hence since $F$ and $\nabla F$ are bounded we have 
$\psi_F\in \calD(P-A)$ for any $\psi\in \calD(P-A)$, 
so $\La\psi_F, i \, [H, D] \, \psi_F \Ra$ is well-defined. 
Secondly, note that the operators $\nabla F\cdot P$ and $P\cdot \nabla F$ are well defined on $\calD(P-A)$. 
Indeed, since $F$ is radial we have $\nabla F=gx$ for some function $g$ depending only on $|x|$. This implies $\nabla F \cdot A=0$, see also \eqref{poincare}. 
Hence, as distributions, 
\begin{equation}\label{eq-useful-1}
	\nabla F\cdot P u 
		= g x\cdot P u = gx\cdot(P-A)u \in L^2(\R^d) 
\end{equation}
for all $u\in\calD(P-A)$. Similarly, 
\begin{equation}\label{eq-useful-2}
 \begin{split}
	P\cdot \nabla F\, u
		&= P\cdot(gx)u= g P\cdot xu -i(x\cdot\nabla g)u
			= gx\cdot (P-A) u - igd u -i(x\cdot\nabla g)u \in L^2(\R^d) \, ,\\
	\la x\ra^{-1}D\, u &= \frac{1}{2\la x\ra}\big( x\cdot P + P\cdot x \big)u 
		= \frac{x}{\la x\ra}\cdot P\, u -\frac{i}{2\la x\ra} u 
		= \la x\ra^{-1}x\cdot (P-A)\, u -\frac{i}{2\la x\ra} u \in L^2(\R^d)  ,\\
	gD\, u &= \frac{g}{2}\big( x\cdot P + P\cdot x \big)u 
		= gx\cdot P\, u -\frac{ig}{2} u 
		= gx\cdot (P-A)\, u -\frac{ig}{2} u \in L^2(\R^d)  ,\\
	\sqrt{g}D\, u 
		&= \sqrt{g}x\cdot (P-A)\, u -\frac{i\sqrt{g}}{2} u \in L^2(\R^d)  , \\
	\langle x\rangle  gD\, u 
	&= \langle x\rangle gx\cdot (P-A)\, u -\frac{i\langle x\rangle g}{2} u \in L^2(\R^d)  ,
 \end{split}
\end{equation}
and
\begin{equation}\label{eq-useful-3}
  	D_{\nabla F}\, u \coloneqq \frac{1}{2}\big(\nabla F\cdot P +P\cdot \nabla F\big)u 
  		= g Du -\frac{i}{2}(x\cdot\nabla g)u \in L^2(\R^d) 
\end{equation}
for all $u\in\calD(P-A)$, by the assumptions on $g$. Note also that $D_{\nabla F}$ is symmetric. 
The next results are needed also later, so we single them out.
  \begin{lemma}\label{lem-energy-boost}
  Under the conditions of Lemma \ref{lem-com-ef} we have 
  \begin{equation}\label{Fmp}
  q_{A,V}(u,v) = q_{A,V}(e^{-F} u, e^{F} v) + 2i \la D_{\nabla F}\,u,v\ra  + \La\nabla F \, u, \nabla F\, v\Ra  
  \end{equation}
  for all $u,v\in \calD(P-A)$. 	
  In particular, if   $\psi$ is a weak eigenfunction 
  corresponding to the energy $E$ of the magnetic 
  Schr\"odinger operator $H_{A,V}$, then 
  \begin{equation}\label{eq-energy boost}
		q_{A,V}(\psi_F,\psi_F) 
			= \La \psi_F,(E+|\nabla F|^2)\psi_F  \Ra
  \end{equation}
\end{lemma}
\begin{proof}
	A straightforward calculation using the above equations and \eqref{eq-aux-1} yields 
\begin{equation}
  \begin{split}	\label{eq-twisted quadratic form}
	q_{A,0}(e^{-F}u, e^{F}v)
		&= \La (P-A+i\nabla F)u, (P-A-i\nabla F)v \Ra \\
		&= q_{A,0}(u,v) -i\big( 
				\La\nabla Fu, (P-A)v \Ra + \La (P-A)u,\nabla Fv \Ra  
				\big) 
			-\La \nabla F u, \nabla F v \Ra \\
		&= q_{A,0}(u,v) -2i 
				\La D_{\nabla F}\, u, v \Ra -\La \nabla F u, \nabla F v \Ra \, .
  \end{split}
  \end{equation}
In particular, since $\La e^{-F}\, u, Ve^F\, v\Ra = \La u, V v\Ra$ and $q_{A,0}(u,v) =q_{A,0}(u,v) + \La u, V v\Ra$ this gives  \eqref{Fmp}. 
If $\psi$ is a weak eigenfunction of $H_{A,V}$ then  $q_{A,V}(\psi,v)= E\La \psi,v \Ra$ for all $v\in\calD(P-A)$.  Since $ D_{\nabla F}$ is symmetric, $\La D_{\nabla F}\psi_F, \psi_F \Ra$ s real and \eqref{Fmp} implies 
  \begin{align*}
  	q_{A,V}(\psi_F,\psi_F) 
  		&=\re q_{A,V}(\psi_F,\psi_F) 
  			= \re q_{A,V}(\psi_F ,e^F \psi_F) + \re \La \nabla F \psi_F, \nabla F \psi_F  \Ra \\
  		&=  \re E\La\psi_F,e^F \psi_F\Ra + \re \La \nabla F \psi_F, \nabla F \psi_F  \Ra
  			= \re \La \psi_F, (E+ |\nabla F|^2) \psi_F  \Ra  \qedhere
  \end{align*}
\end{proof}

\begin{proof}[Proof of Lemma \ref{lem-com-ef}] 
From \eqref{eq-aux-1} we know that $\psi_F\in \calD(P-A)=\calQ(H_{A,V})$. Thus for any $\psi\in \calQ(H_{A,V})$ our magnetic virial Theorem \ref{thm:magnetic-virial} shows 
\begin{align*}
	\La\psi_F, i \, [H, D] \, \psi_F \Ra
		&=2 q_{A,0}(\psi_F,\psi_F)  + 2\, \re \La\wti{B}\psi_F,(P-A)\, u\Ra- \La \psi_F, x\cdot\nabla V\psi_F \Ra \, .
\end{align*}
with $q_{A,0}(\psi_F,\psi_F)=\La (P-A) \psi_F, (P-A)\psi_F\Ra $. 
If $\psi$ is a weak eigenfunction of $H_{A,V}$ with energy $E$, then  
\begin{align*}
	\La\psi_F, i \, [H, D] \, \psi_F \Ra
		&=q_{A,V}(\psi_F,\psi_F) -\La \psi_F, V\psi_F \Ra 
			+q_{A,0}(\psi_F,\psi_F)  \\
		&\phantom{===}+ 2\, \re \La\wti{B}\psi_F, (P-A)\, \psi_F\Ra- \La \psi_F, x\cdot\nabla V\psi_F \Ra \\
		&= \La\psi_F,(E+|\nabla F|^2)\psi_F) 
			+q_{A,0}(\psi_F,\psi_F)  \\
		&\phantom{===}	+ 2\, \re \La\wti{B}\psi_F,  (P-A)\, \psi_F\Ra
			- \La \psi_F, (V+x\cdot\nabla V)\psi_F \Ra 
\end{align*}
by \eqref{eq-energy boost}. This proves the first claim of Lemma \ref{lem-com-ef}. 
\medskip

Applying \eqref{Fmp} with $u = \psi_F$ and $v= iD_t\psi_F$ one sees 
\begin{align*}
	q(\psi_F, iD_t\psi_F ) 
		&= q(\psi,e^{F}iD_t \psi_F) 	+2i\la D_{\nabla F} \psi_F, iD_t \psi_F\ra 
			+ \La\nabla F \, \psi_F , \nabla F\, iD_t\psi_F\Ra\\
		& = E \, \La \psi_F, \,iD_t\psi \Ra-2\la D_{\nabla F} \psi_F, D_t \psi_F\ra 				+ \La \, \psi_F , |\nabla F|^2 iD_t\psi_F\Ra \, ,
\end{align*}
where we again used $q_{A,V}(\psi,v)= E\la \psi,v \ra$ for all $v\in\calD(P-A)$ and  any weak eigenfunction $\psi$ with energy $E$.  
Notice that $\la  \psi_F, \,iD_t\psi_F\ra = i \la\psi_F, D_t\psi_F\ra$ is purely imaginary since $D_t$ is symmetric, so taking the real part above shows 
\begin{equation}\label{eq:calculation-neg-virial}
	2\re q(\psi_F, iD_t\psi_F ) 
		= -4\re \la D_{\nabla F}\, \psi_F, D_t\psi_F \ra + 2\re\la  \psi_F ,  |\nabla F|^2iD_t\psi_F \ra .
\end{equation}
Lemma \ref{lem:xgradv} gives  
$2\re\la  \psi_F ,  |\nabla F|^2iD_t\psi_F \ra\to -\la \psi_F, x\cdot\nabla(|\nabla F|^2) \psi_F\ra$ as $t\to0$. Hence \eqref{eq:calculation-neg-virial} implies  \eqref{eq-psiF-2} as long as  
\begin{equation}\label{eq-neg-virial-final}
	\lim_{t\to 0}\re \la D_{\nabla F}\, \psi_F, D_t\psi_F \ra = \|\sqrt{g}D\psi_F\|_2^2 -\frac{1}{4}\la \psi_F, ((x\cdot\nabla)^2 g)\psi_F \ra .
\end{equation} 
Using $D_{\nabla F}u= gDu-\tfrac{i}{2}(x\cdot\nabla g)u$ for all  $u\in\calD(P-A)$, we get   
\begin{equation*}
	\la D_{\nabla F}\, u, D_tu \ra 
		= \La gD\, u, D_tu \Ra+ \frac{1}{2}\La (x\cdot\nabla g)u, iD_tu \Ra
\end{equation*}
and we already know from Lemma \ref{lem:xgradv} that 
$ \frac{1}{2}\re\La (x\cdot\nabla g)u, iD_tu \Ra\to -\frac{1}{4}\La u, ((x\cdot\nabla)^2 g)u \Ra$ as $t\to0$. 
Moreover, 
\begin{align*}
	\la x\ra^{-1}D_tu = \frac{1}{2t}\int_{-t}^t \la x \ra^{-1}U_s\big( D u \big)\, ds
		= \frac{1}{2t}\int_{-t}^t \frac{\la e^s x \ra}{\la x \ra}U_s\big(  \la x \ra^{-1}Du \big)\, ds
\end{align*}
initially for $u\in\calC^\infty_0(\R^d)$, but by density and since 
$\la x \ra^{-1}D:\calD(P-A)\to L^2(\R^d)$ is bounded,  
this extends to all $u\in \calD(P-A)$. 
Thus, by continuity, $ \la x\ra^{-1}D_tu\to \la x \ra^{-1}Du $ in $L^2(\R^d)$ as $t\to0$ and 
\begin{align*}
	\la gD\, u, D_tu \ra 
	&= \la \la x\ra gD\, u, \la x\ra^{-1}D_tu \ra 
		\to \la \la x\ra gD\, u, \la x\ra^{-1}Du \ra 
			=\|\sqrt{g}Du\|_2^2 
\end{align*}
as $t\to0$ for all $u\in\calD(P-A)$. 
 This completes the proof of \eqref{eq-neg-virial-final} and of the Lemma.
\end{proof}

\begin{lemma} \label{lem-lowerb}
Let $B$ and $V$ satisfy assumptions \ref{ass-B-mild-int}, \ref{ass-V-form small}, and \ref{ass-bounded unique continuation}. Assume that $\psi$ and $F$ satisfy conditions of Lemma \ref{lem-com-ef}. Then there exists $\kappa>0$ and $c_\kappa>0$ such that 
\begin{equation} \label{com-lowerb-2} 
\La\psi_F, i \, [H, D] \, \psi_F \Ra\, \geq  \, \kappa\,  \La\psi_F, |\nabla F|^2\, \psi_F \Ra- c_\kappa \|\psi_F\|_2^2 \, .
\end{equation}
\end{lemma}

\begin{proof}
In what follows the value of a constant $c$ might change from line to line. Since $\psi_F\in \h^1_A(\R^d)$, Lemma \ref{lem-com-ef}, the Cauchy-Schwarz inequality and Assumption \ref{ass-V-split form bounded} give 
\begin{align*}
\La\psi_F, i \, [H, D] \, \psi_F \Ra & \geq   \|(P-A)\psi_F\|_2^2 -2  \|(P-A)\psi_F\|_2 \big(\|\wti{B} \psi_F\|_2 +\|x V_1 \psi_F\|_2\big)\\ 
& \ \  -( \alpha_2+ d\, \alpha_3) \|(P-A)\psi_F\|_2^2 - c \|\psi_F\|_2^2\, .
\end{align*}
Therefore using \eqref{eq-energy boost} and Assumption \ref{ass-V-form small}  we find that for any $\kappa>0$
\begin{align*}
\La\psi_F, i \, [H, D] \, \psi_F \Ra & \geq (1 -\kappa) \|(P-A)\psi_F\|_2^2  + \kappa\,  \La\psi_F, |\nabla F|^2\, \psi_F \Ra- (\alpha_2 +d\alpha_3+\kappa \alpha_0) \, \|(P-A)\psi_F\|_2^2 \\
&\ \  -2  \|(P-A)\psi_F\|_2 \big(\|\wti{B} \psi_F\|_2 +\|x V_1 \psi_F\|_2\big)  - c \|\psi_F\|_2^2.
\end{align*}
On the other hand Assumption \ref{ass-V-split form bounded} implies that 
\begin{align*}
2 \|(P-A)\psi_F\|_2 \big(\|\wti{B} \psi_F\|_2+\|x V_1 \psi_F\|_2\big)   & \, \leq \, \alpha_ 1  \|(P-A)\psi_F\|_2^2 +2\, c_1 \|(P-A)\psi_F\|_2 \,  \|\psi_F\|_2 \\
&\,  \leq \, (\alpha_ 1+\kappa)\,  \|(P-A)\psi_F\|_2^2 + \frac{c_1}{\kappa} \,  \|\psi_F\|_2^2. 
\end{align*}
Hence 
\begin{align*}
\La\psi_F, i \, [H, D] \, \psi_F \Ra & \geq (1 -2\kappa-\kappa \alpha_0 -\alpha_1-\alpha_2-d\, \alpha_3) \, \|(P-A)\psi_F\|_2^2 + \kappa\,  \La\psi_F, |\nabla F|^2\, \psi_F \Ra \\
 & \quad  - (c+ \kappa^{-1} c_1)\, \|\psi_F\|_2^2,
\end{align*}
and the result follows upon setting 
$$
\kappa = \frac{1-\alpha_1-\alpha_2-d\, \alpha_3}{2+\alpha_0} >0.
$$
\end{proof}

\section{Absence of positive eigenvalues}
\label{sec-abs}
We will give the proof of absence of positive eigenvalues in two steps. The first is that putative eigenfunctions corresponding to positive energies have to decay faster than exponentially. In a second step, we prove that any such eigenfunction has to be zero.   

\subsection{Ridiculously fast decay}
\noindent We set $\x_\lambda \coloneqq \sqrt{\lambda+|x|^2}\ $ for $x\in\R^d,\, \lambda>0$.
For $\lambda=1$, we write simply $ \x_1\, =\,  \x$.  We have

\begin{proposition}[Fast decay] \label{prop-decay}
Assume that $B$ and $V$ satisfy Assumptions \ref{ass-B-mild-int}- \ref{ass-bounded infinity} and that the magnetic field $A$ corresponding to $B$ is in the Poincar\'e 
gauge. 
Furthermore, assume that $\psi$ is a weak eigenfunction of the magnetic 
Schr\"odinger operator $H_{A,V}$ corresponding to the energy $E\in\R ,$ 
and that there exist $\ol{\mu}\ge 0$ and $\lambda>0$ such that 
$x\mapsto e^{\, \ol{\mu}\, \x_\lambda }\, \psi(x) \in L^2(\R^d)$. If $E+\ol{\mu}^2>\Lambda$ 
with  $\Lambda$ given by \eqref{edge}, then 
\begin{equation} \label{exp-decay} 
x\mapsto e^{\, \mu  \x_\lambda  }\, \psi(x) \in L^2(\R^d) \qquad \forall \, \mu >0, \quad \, \forall \, \lambda>0. 
\end{equation}
\end{proposition}

Before we start with the proof, we make 
some preparations. 
Obviously it suffices to prove the statement for 
$\lambda=1$. We will  first consider the case 
$\ol{\mu}=0$, i.e., we only know that $\psi\in \calD(P-A)\subset L^2(\R^d)$. 
The choice 
\begin{equation}  \label{eq-F}
F_{\mu,\eps}(x) = \frac\mu\eps \left(1-e^{-\eps\, \x}\right)\, ,
\end{equation}
for the weight function, for some $\mu\ge 0$ and $\veps>0$, will be 
convenient. 
We have   $F_{\mu,\eps}(x) \to \mu \x$ as $\eps\to 0$. 
Also, since 
\begin{equation}\label{grad-F} 
	\nabla F_{\mu,\veps} = \mu \la x\ra^{-1} e^{-\veps\la x \ra} x\, 
\end{equation} 
we have 
\begin{equation}\label{eq-g}
	g_{\mu,\veps}(x) = \mu \la x \ra^{-1} e^{-\veps\la x \ra}\, .
\end{equation}
Moreover, let 
$$
\mu_* = \sup \left\{ \mu\geq 0\, : \,  e^{\mu \x   } \psi \in L^2(\R^d)\right\}\, ,
$$
the maximal exponential decay rate of the weak eigenfunction $\psi$. The bound \eqref{exp-decay} is 
equivalent to $\mu_*=\infty$, so we have to exclude $0\le \mu_*<\infty$.    
If $0\le \mu_*<\infty$,  then there exist sequences $\mu_n\searrow\mu_*$, $\veps_n\searrow 0$ as $n\to\infty$, i.e., both sequences are decreasing and $\mu_n\to\mu_*$, $\veps_n\to0$, as $n\to\infty$,  with 
\begin{equation}  \label{div-n}
a_n : = \| e^{ F_n}\, \psi\|_2 \ \to \ \infty \quad \text{as} \quad  n\to\infty,
\end{equation} 
where we put $F_n\coloneqq F_{\mu_n,\veps_n}$. Moreover, we let 
$g_n(x)\coloneqq g_{\mu_n,\veps_n}$ and define
$\varphi_n$  by
\begin{equation} \label{fin}
\varphi_n = \frac{ e^{ F_n}\, \psi}{\|e^{ F_n}\, \psi\|}\, . 
\end{equation}
Since $
F_n(x) \, \leq \, \mu_n \x $,
the function  $e^{F_n}$ is bounded uniformly in $n\in\N$ on compact subsets of 
$\R^d$. This  implies that for any compact subset $K\subset \R^d$ 
one has 
\begin{equation*}
	\La\varphi_n, \id_K \varphi_n \Ra\to  0 \quad \text{as } n\to\infty
\end{equation*}
where $\id_K$ is the characteristic function of $K$. 
In turn, this implies that for any bounded function $W$ with $W(x)\to 0$ as $x\to\infty$ one has
\begin{equation} \label{eq-local-vanishing-2}
	\La\varphi_n, W \varphi_n \Ra\to  0 \quad \text{as } n\to\infty. 
\end{equation}
The last equation is the central point of the argument used in the proof of Proposition \ref{prop-decay}. It will allow us to show that in the virial identity
applied to $\varphi_n$ certain terms vanish as $n\to \infty$. This turns crucial when applying  \ref{lem-com-ef} to derive a contradiction in the proof of  Proposition \ref{prop-decay}. 

\begin{lemma}\label{lem-punch-1}
	Let $F_n$, $g_n$, $\psi$, and $\varphi_n$ be given as above.  If $0<\mu_*<\infty$, then 
	\begin{equation}\label{eq-punch-1}
		\lim_{n\to\infty} \la e^{F_n}\psi, \veps_n\la x \ra  e^{F_n}\psi \ra =0\, .
	\end{equation}
	Moreover, if $0\le \mu_*<\infty$, then 
	\begin{align}
		\lim_{n\to\infty}\La \nabla F_n\varphi_n,\nabla F_n  \varphi_n\Ra 
			&= \mu_*^2  \label{eq-punch-2} 
	\end{align} 
	and 
	\begin{align}
		\lim_{n\to\infty}\La \varphi_n, \big( (x\cdot\nabla)^2g_n-x\cdot \nabla |\nabla F_n|^2\big) \varphi_n\Ra 
			&=   0  \label{eq-punch-3}	
	\end{align} 
\end{lemma}
\begin{remark}
	If $\mu_*>0$, then $\psi$ decays exponentially and since $F_n$ is bounded for fixed $n\in\N$ we have  $\la  e^{F_n}\psi, \la x \ra  e^{F_n}\psi\ra <\infty$ for all $n$. 
\end{remark}

\begin{lemma}\label{lem-punch-2}
		Let $0\le \mu_*<\infty$ and $F_n$, $g_n$, and $\varphi_n$ be given as above. If the potential $V$ is relative form small and vanishes at infinity w.r.t~$(P-A)^2$, i.e satisfies assumptions \ref{ass-V-form small} and  
		\ref{ass-V-vanishing infinity}, then 
  	\begin{align}
		\lim_{n\to\infty}\la \varphi_n, V \varphi_n\ra 
			&= 0  \label{eq-punch-2-1}\\
		\lim_{n\to\infty}\la (P-A)\varphi_n, (P-A) \varphi_n\ra 
			&= E+\mu_*^2 \, . \label{eq-punch-2-2}
	\end{align}
	Moreover, if the magnetic field $B$ satisfy assumptions \ref{ass-B-rel-bounded}, and \ref{ass-bounded infinity}, then 
	\begin{align}
		\limsup_{n\to\infty} |\La\wti{B}\,  \varphi_n,(P-A)\, \varphi_n \Ra| 
			&\le \beta (E+\mu_*^2)^{1/2} \label{eq-punch-2-3}\, .
	\end{align}
	and if one splits $V=V_1+V_2$, with $V_1$ and $V_2$ satisfying assumptions \ref{ass-V-split form bounded} and 
	\ref{ass-bounded infinity} then 
		\begin{align}
 			\limsup_{n\to\infty} \La \varphi_n, x\cdot\nabla V \varphi_n\Ra 
			&\le  2\omega_1(E+\mu_*^2)^{1/2} + \omega_2 \label{eq-punch-2-5} .
	\end{align}
Here $\beta,\omega_1$, and $\omega_2$ measure the strength 
  of the magnetic field and the virial of the potential near 
  infinity. 
\end{lemma}
\begin{remark}
  For the proof of similar results in \cite{fhhh}, the assumption that $V$ and $x\cdot\nabla V$ are relatively form compact with respect to $P^2$ is made. Thus they only deal with potentials which are relatively form bounded with relative bound zero. They also do not consider conditions on the  Kato form of the virial $x\cdot\nabla V$. 
\end{remark}

\noindent We will prove these two Lemmas later in this section. 

\begin{proof}[Proof of Proposition \ref{prop-decay}]
Assume that $0\le \mu_*^2<\infty$. It is easy to check that $F_n$ and $g_n$ satisfy the assumptions of the 
exponentially weighted magnetic virial Lemma \ref{lem-com-ef}. 
Thus Lemma \ref{lem-com-ef} and 
Lemma \ref{lem-punch-1} show
\begin{equation}\label{eq-contradiction-1}
	\limsup_{n\to\infty} \La\varphi_n, i \, [H, D] \, \varphi_n \Ra\le 0\, . 
\end{equation}
On the other hand the first equality from Lemma \ref{lem-com-ef} together with Lemma \ref{lem-punch-2} shows 
\begin{equation}
  \begin{split}\label{eq-contradiction-2}
	\liminf_{n\to\infty}  \La\varphi_n, i \, [H, D] \, \varphi_n \Ra
		&\ge 2(E+\mu_*^2) -2(\beta+\omega_1) (E+\mu_*^2)^{1/2} - \omega_2 \\
		&=2\left[ \left(\sqrt{E+\mu_*^2}-\frac{\beta+\omega_1}{2}\right)^2 -\left( \frac{\beta+\omega_1}{2}\right)^2 -\frac{\omega_2}{2}\right] >0
  \end{split}
\end{equation}
if $\sqrt{E+\mu_*^2}> \frac{1}{2}\big(\beta+\omega_1 + \sqrt{(\beta+\omega_1)^2 +2\omega_2}\big)= \sqrt{\Lambda}$. 
Clearly, \eqref{eq-contradiction-1} and \eqref{eq-contradiction-2} contradict each other. 
Thus $\mu_*=\infty$,  which is equivalent to   \eqref{exp-decay}. 
\end{proof}
\noindent It remains to prove Lemmas \ref{lem-punch-1} and \ref{lem-punch-2}.  
\begin{proof}[Proof of Lemma \ref{lem-punch-1}]
Clearly, for any $\delta>0$
\begin{align*}
	\La \varphi_n, \veps_n\la x\ra \varphi_n\Ra 
		= \La \varphi_n, \id_{\{\veps_n\la x \ra< \delta\}}\, \varphi_n\Ra 
			+\La \varphi_n, \id_{\{\veps_n\la x \ra \ge  \delta\}}\,   \veps_n\la x\ra \varphi_n\Ra
		\le \delta +  \La \varphi_n, \id_{\{\veps_n\la x \ra > \delta\}}\, \veps_n\la x\ra \varphi_n\Ra
\end{align*}
One easily checks that the mapping $t\mapsto \frac{1-e^{-t}}{t}$ is decreasing on $(0,\infty)$. Thus 
\begin{equation}
	\gamma_\delta\coloneqq \sup_{t\ge \delta}\frac{1-e^{-t}}{t} = \frac{1-e^{-\delta}}{\delta}<1
\end{equation}
which shows 
\begin{align*}
	F_n = \frac{\mu_n\la x \ra}{\veps_n\la x \ra}(1-e^{-\veps_n\la x \ra}) \le \mu_n\gamma_\delta \la x \ra \quad \text{for all } x \text{ with }\veps_n\la x \ra\ge\delta\, .
\end{align*}
Given $\delta>0$ choose any $\kappa$ with $\gamma_\delta<\kappa<1$. If $0<\mu_*<\infty$ then 
$\psi$ decays exponentially with rate $\kappa\mu_*<\mu_*$,  
by the definition of $\mu_*$. Thus  
\begin{equation*}
	\limsup_{n\to\infty}  \La e^{F_n}\psi, \id_{\{\veps_n\la x \ra > \delta\}}\, \la x\ra  e^{F_n}\psi \Ra 
		\le  \limsup_{n\to\infty} \La e^{\mu_n\gamma_\delta \la x \ra}\psi, \la x\ra  e^{\mu_n\gamma_\delta \la x \ra}\psi \Ra <\infty 
\end{equation*} 
since, $\mu_n\gamma_\delta\to \gamma_\delta\mu_* <\kappa\mu_*$ as $n\to\infty$.  In view of \eqref{div-n} this implies \eqref{eq-punch-1}.
\medskip

\noindent For the proof of the remaining part of Lemma \ref{lem-punch-1}, we note that from  \eqref{grad-F} one gets 
\begin{equation}\label{eq-|grad F|^2}
	|\nabla F_n|^2 = \mu_n^2 \big( 1-\la x \ra^{-2} \big)e^{-2\veps_n\la x \ra}\, .
\end{equation} 
Since $\varphi_n$ is normalized this gives   
\begin{equation}\label{eq-energy-shift-1}
  \begin{split}
	\mu_n^2-\La \nabla F_n \varphi_n, \nabla F_n \varphi_n\Ra 
		&= \La \varphi_n, \big(\mu_n^2- |\nabla F_n|^2\big) \varphi_n\Ra \\
		&= \mu_n^2 \left( \La \varphi_n, 
							\big(1- e^{-2\veps_n \la x \ra}\big)
						\varphi_n\Ra 
						+ \La \varphi_n, 
							\la x\ra^{-2} e^{-2\veps_n \la x \ra}	
						  \varphi_n\Ra 
					\right)
  \end{split}
\end{equation}
Recall that $\mu_n\searrow\mu_*$. If $\mu_*=0$, then \eqref{eq-energy-shift-1} shows
\begin{align*}
	\left| \mu_n^2-\La \nabla F_n \varphi_n, \nabla F_n \varphi_n\Ra  \right| \le 2\mu_n^2 \to 0 \quad \text{as } n\to\infty\, . 
\end{align*}
If $0<\mu_*<\infty$, then using $0\le 1- e^{-2\veps_n \la x \ra} \le 2\veps_n \la x \ra$ in  \eqref{eq-energy-shift-1} gives  
\begin{align*}
	\left| \mu_n^2-\La \nabla F_n \varphi_n, \nabla F_n \varphi_n\Ra  \right| 
		&\le \mu_n^2\left(   2 \La  \varphi_n, 
							\veps_n \la x \ra
						\varphi_n\Ra 
					+  \La \varphi_n, 
							\la x\ra^{-2}
						\varphi_n\Ra 
					\right) 
			\to 0 \quad \text{as } n\to\infty 
\end{align*}
due to \eqref{eq-punch-1} and \eqref{eq-local-vanishing-2}. 
This proves \eqref{eq-punch-2}.
\smallskip

Using  the definitions of $F_n$ and $g_n$ a relatively short calculation shows  
\begin{equation}\label{eq-virial-lower-order}
	\left| (x\cdot\nabla)^2g_n-x\cdot \nabla |\nabla F|^2\right|
		\lesssim 
			\mu_n(\mu_n+1)
			\left[ \la x \ra^{-2} + \la x \ra^{-1}+ \veps_n \la x \ra + \veps_n^2 \la x \ra 
			\right] e^{-\veps_n\la x \ra}
\end{equation}
Since $0\le t\mapsto t e^{-t}$ is bounded, \eqref{eq-virial-lower-order} implies, if $\mu_*=0$, 
\begin{align*}
	\left|
		\La \varphi_n,\big( (x\cdot\nabla)^2g_n-x\cdot \nabla |\nabla F|^2\big) \varphi_n\Ra
	\right|
		\lesssim \mu_n(\mu_n+1)
 \to 0  \quad \text{as } n\to\infty 
\end{align*}
If $0<\mu_*<\infty$, then \eqref{eq-virial-lower-order} shows 
\begin{align*}
	\left|
		\La \varphi_n,\big( (x\cdot\nabla)^2g_n-x\cdot \nabla |\nabla F|^2\big) \varphi_n\Ra
	\right|
		\lesssim 
			\La \varphi_n,\Big(\la x \ra^{-2} +\la x \ra^{-1}\Big)  \varphi_n\Ra 
			+ \La \varphi_n, \veps_n\la x \ra  \varphi_n\Ra 
			 \to 0  \quad \text{as } n\to\infty 
\end{align*}
using again \eqref{eq-punch-1} and \eqref{eq-local-vanishing-2}. This proves \eqref{eq-punch-3}. 
\smallskip
\end{proof}

\noindent In the proof of Lemma \ref{lem-punch-2} we need the following auxiliary tool. 

\begin{lemma}\label{lem-kinetic energy bounded}
	Assume that the potential $V$ is relatively form bounded w.r.t~$(P-A)^2$. Then for any family of real-valued bounded function $\xi_j\in\calC^\infty_0(\R^d)$, $j\in I$, for which 
	$\sup_{j\in I}\|\xi_j\|_\infty$ and $\sup_{j\in I} \|\nabla\xi_j\|_{\infty}$ are finite, we have  
	\begin{equation}
		\sup_{j\in I}\sup_{n\in\N} \|(P-A)\xi_j\varphi_n\| <\infty\, .
	\end{equation}
	where $\varphi_n$ is the sequence  defined in \eqref{fin}.
	Moreover, if $\xi\in\calC^\infty_0(\R^d)$ is a real-valued function with compact support, then 
	\begin{equation}
		\limsup_{n\to\infty} \|(P-A)\xi\varphi_n\| = 0\, .
	\end{equation}
\end{lemma}
We give the proof of this Lemma after the 
\begin{proof}[Proof of Lemma \ref{lem-punch-2}]
 One easily checks that if $\xi$ is an infinitely often differentiable cut--off function with bounded derivative, then $\xi\varphi\in\calD(P-A)$ for any $\varphi\in\calD(P-A)$.
 
 Let $\chi_{l}:[0,\infty)\to \R_+$, $l=1,2$, be infinitely often 
 differentiable on $(0,\infty)$ with $\chi_1(r)=1$ for $0\le r\le 1$, $\chi_1(r)>0$ for $r\le 3/2$,   $\chi_1(r)=0$ for 
 $r\ge  7/4$,  
 and $\chi_2(r)=0$ for $r\le 5/4$,   $\chi_2(r)>0$ for $r\ge  3/2$,    $\chi_2(r)=1$ for $r\ge  2$. 
 Then $\inf_{r\ge 0}(\chi_1^2(r)+\chi_2^2(r))>0$ and thus 
 \begin{align*}
 	\xi_1 \coloneqq \frac{\chi_1}{\sqrt{\chi_1^2+\chi_2^2}}\, , 
 	\quad 
 	\xi_2 \coloneqq \frac{\chi_2}{\sqrt{\chi_1^2+\chi_2^2}} 
 \end{align*}
are infinitely often differentiable with bounded derivatives and $\xi_1^2+\xi_2^2=1$. Given $R\ge 1$ we set 
\begin{align*}
	\xi_{<R}(x)\coloneqq \xi_1(|x|/R) ,  
	\quad 	\xi_{\ge R}(x)\coloneqq \xi_2(|x|/R) 
\end{align*}
which yields a family of infinitely often differentiable real-valued localization functions on $\R^d$ with bounded derivatives. 
Note that $\xi_{<R}$ has compact support and $\supp(\xi_{\ge R})\subset  \U_R^{\, c}=\{x\in\R^d: |x|\ge R\}$.
By construction, we have 
\begin{align*}
  \La \varphi_n, V\varphi_n \Ra 
  	= \La \xi_{<R}^2\, \varphi_n, V\varphi_n \Ra +  \La \xi_{\ge R}^2\, \varphi_n, V\varphi_n \Ra
\end{align*}
and, recalling that $V$ is form bounded with respect to $(P-A)^2$,  we have for fixed $R\ge 1$ 
\begin{align*}
	 |\La \xi_{<R}^2\, \varphi_n, V\varphi \Ra|
	 & = |\La \xi_{<R}\, \varphi_n, V\xi_{<R}\, \varphi_n \Ra|
	 	\lesssim \|(P-A)\xi_{<R}\, \varphi_n\|_2^2+ \|\xi_{<R}\, \varphi_n\|_2^2
	 	\to 0\, , \text{ as } n\to\infty 
\end{align*}
by Lemma \ref{lem-kinetic energy bounded} and \eqref{eq-local-vanishing-2}, since $\xi_{<R}$ 
has compact support. Since $V$ vanishes at infinity w.r.t.~$(P-A)^2$, 
there exist $\alpha_R, \gamma_R$ with 
$\alpha_R, \gamma_R \to 0$ as $R\to\infty$ such that 
\begin{align*}
	|\La \xi_{\ge R}^2\, \varphi_n, V\varphi_n \Ra|
		= |\La \xi_{\ge R}\, \varphi_n, V\xi_{\ge R}\, \varphi_n \Ra| 
		\le  \alpha_R \|(P-A)\xi_{\ge R}\, \varphi_n\|_2^2 + \gamma_R\|\xi_{\ge R}\varphi_n\|_2^2 \, .
\end{align*}
Lemma \ref{lem-kinetic energy bounded} then shows 
\begin{align*}
		\limsup_{n\to\infty}|\La \xi_{\ge R}^2\, \varphi_n, V\varphi_n \Ra| 
		  \lesssim   \alpha_R  + \gamma_R 
		  \to 0\, , \text{ as } R\to\infty\, ,
\end{align*}
which proves \eqref{eq-punch-2-1}.

\noindent Moreover, from Lemma \ref{lem-energy-boost}, we get  
\begin{align*}
	\La (P-A)\varphi_n, (P-A)\varphi_n \Ra 
		&= E + \La \nabla F_n \varphi_n, \nabla F_n \varphi_n \Ra - \La\varphi_n, V\varphi_n \Ra \\
		&\to E+\mu_*^2 \quad \text{as } n\to\infty 
\end{align*}
using also \eqref{eq-punch-2-1} and  \eqref{eq-punch-1}. 
This  proves \eqref{eq-punch-2-2}. 
\smallskip

\noindent For $\wti{B}^2$ one can argue exactly the same way as above for $V$ to see that for fixed $R$ 
\begin{align*}
	\limsup_{n\to\infty}  \La\varphi_n,|\wti{B}|^2\varphi_n\Ra
	\le 	\limsup_{n\to\infty}  \La \xi_{\ge R}^2\,  \varphi_n,\, |\wti{B}|^2 \varphi_n\Ra 
	\le C \veps_R +\beta_R^2 
\end{align*}
where we also used Assumption \ref{ass-bounded infinity} and put 
$C=\sup_{j\in\N}\limsup_{n\to\infty}\|(P-A)\xi_j\varphi_n\|_2^2 $, which due to Lemma \ref{lem-kinetic energy bounded} is finite. 
Since $\veps_R\to 0$ and $\beta_R\to\beta$, as $R\to\infty$, we get  
\begin{align*}
	\limsup_{n\to\infty} \|\wti{B}\varphi_n\| \le \beta \, , 
\end{align*}
Because of $ |\La\wti{B}\,  \varphi_n,(P-A)\, \varphi_n \Ra| 
\le  \|\wti{B}\,  \varphi_n\| \|(P-A)\, \varphi_n \| $ and \eqref{eq-punch-2-2} this proves \eqref{eq-punch-2-3}.

\noindent If the potential splits as $V=V_1+V_2$ with $V_1,V_2$ satisfying assumptions \ref{ass-V-split form bounded} and \ref{ass-bounded infinity}, then one can argue exactly as above to see that 
\begin{align*}
	\limsup_{n\to\infty} |\La xV_1\varphi_n, (P-A)\varphi_n\Ra| 
		&\le  \ \omega_1
\intertext{and }
		\limsup_{n\to\infty} |\La \varphi_n, x \cdot\nabla V_2 \varphi_n\Ra| 
		   &\le  \ \omega_2\, .
\end{align*}
Moreover, if $V_1$ and $(xV_1)^2$ are form bounded w.r.t.~$(P-A)^2$ and  
$\varphi\in\calD(P-A)$ with $\supp(\varphi)\subset \{|x|\ge R\}$, then 
\begin{align*}
	 |\La \varphi,  V_1 \, \varphi\Ra| 
	 	&=   |\La |x|^{-1}\varphi, |x| V_1 \varphi\Ra| 
	 		\le \| |x|^{-1}\varphi\| \||x| V_1 \varphi\|
	 		\lesssim R^{-1}\|\varphi\| \left( \|(P-A)\varphi\|_2^2 + \|\varphi\|_2^2 \right)^{1/2}\, ,
\end{align*}
so $V_1$ vanishes at infinity w.r.t.~$(P-A)^2$. 
Thus $\lim_{n\to\infty}\La\varphi_n, V_1\, \varphi_n\Ra=0$ and 
using the mixed form of the virial from Corollary \ref{cor-mixed-virial} yields 
  \begin{equation*}
  	\limsup_{n\to\infty} 
  		\La \varphi , x\cdot\nabla V\varphi \Ra
  		\le 2\omega_1(E+\mu_*^2)^{1/2} +\omega_2\, .
  		\qedhere
  \end{equation*}
\end{proof}

\begin{remark}
Note that $\Lambda < \beta +\omega$ as soon as $\omega>0$. 
\end{remark}

Now we give the 
\begin{proof}[Proof of Lemma \ref{lem-kinetic energy bounded}]
Let $\psi\in \calD(P-A)$ be a weak eigenfunction of the 
magnetic Schr\"odinger operator $H_{A,V}$ with eigenvalue $E$ 
and $F_n$, $\psi_{n}= e^{F_n}\psi$ and $\varphi_n=\psi_{n}/\|\psi_{n}\|$ as in \eqref{fin}. 
In particular, we have $\sup_n\|\nabla F_n\|\le \sup_n\mu_n<\infty$. 
Since $V$ is relatively form bounded with respect to $(P-A)^2$ 
\begin{align*}
	\|(P-A)\varphi\|_2^2 = q_{A,V}(\varphi,\varphi) - \La \varphi, V\varphi \Ra 
	\le  q_{A,V}(\varphi,\varphi) 
		+\alpha_0 	\|(P-A)\varphi\|_2^2 + C\|\varphi\|_2^2
\end{align*}
for some $0\le \alpha_0<1$, $C>0$, and all $\varphi\in \calD(P-A)$. Thus 
\begin{align*}
		\|(P-A)\varphi\|_2^2 
			\le (1-\alpha_0)^{-1}\left( q_{A,V}(\varphi,\varphi)  
		 	+ C\|\varphi\|_2^2 \right)
\end{align*}
From the IMS localization formula \eqref{eq-ims} we get 
\begin{align*}
	q_{A,V}(\xi\psi_n,\xi\psi_{n}) 
		&=  \re q_{A,V}(\xi^2e^{2F_n}\psi,\psi) + \La \psi, |\nabla(\xi e^{F_n})|^2\psi \Ra  \\
		&\le 	E\|\xi\psi_n\|_2^2 + 2 \|(\nabla\xi)\psi_n \|_2^2 
			+ 2  \|(\nabla F_n)\xi \psi_n \|_2^2 	
\end{align*}
since $\psi$ is a weak eigenfunction with energy $E$. Thus 
\begin{align*}
	\|(P-A)\xi_j\varphi_n\|_2^2 
		\lesssim  	\|\xi_j\varphi_n\|_2^2  + \|(\nabla\xi_j)\varphi_n \|_2^2
\end{align*}
where the implicit constant is independent of $j\in I$ and 
$n\in\N$.  
Since $\varphi_n$ is normalized, this proves the first claim.

On the other hand, if $\xi$ has compact support then so does 
$\nabla\xi$. Thus, from \eqref{eq-local-vanishing-2} we get  
$\|\xi\varphi_n\|\to 0$ and $\|(\nabla\xi)\varphi_n\|\to 0$, 
as $n\to\infty$. Hence, 
\begin{align*}
			\|(P-A)\xi\varphi_n\|_2^2 \
			\lesssim \
				\|\xi\varphi_n\|_2^2  + \|(\nabla\xi)\varphi_n\|_2^2
				 \to 0\, ,
\end{align*}
as $n\to\infty$.
\end{proof}

\smallskip

\subsection{Absence of positive eigenvalues}
Now we are in position to prove our main result.
\begin{theorem} \label{thm-abs}
  Let $B$ and $V$ satisfy assumptions \ref{ass-B-mild-int}- \ref{ass-bounded infinity}. 
  Then the magnetic Schr\"odinger operator $H_{A,V}$ has 
  no eigenvalues in the interval $(\Lambda,\infty)$, where 
  $\Lambda$ is given by \eqref{edge}.  
    
  Moreover, if $E\le \Lambda$ is an eigenvalue of 
  $H_{A,V}$ then any weak eigenfunction $\psi$ with 
  energy $E$ cannot decay faster than 
  $e^{\sqrt{\Lambda-E}|x|}$, in the sense that if 
  $x\mapsto e^{\ol{\mu}|x|}\psi(x)\in L^2(\R^d)$ for some 
  $\ol{\mu}>\sqrt{\Lambda-E}$, then $\psi$ is the zero function.    
\end{theorem}

\begin{proof}
Let $q_{A,V}$ be the quadratic from corresponding to $H_{A,V}$ 
and assume that $E\La\varphi,\psi\Ra = q_{A,V}(\varphi,\psi)$ for all $\varphi\in\calD(q_{A,V})=\calD(P-A)$. Furthermore, assume that either $E>\Lambda$ or $E+\ol{\mu}^2>\Lambda$ for some $\ol{\mu}>0$ and $x\mapsto e^{\ol{\mu}|x|}\psi(x)\in L^2(\R^d)$. 
Then from Proposition \ref{prop-decay} we know that 
\begin{equation*}
x\mapsto e^{\, \mu  \x_\lambda  }\, \psi(x) \in L^2(\R^d) \qquad \forall \, \mu >0, \quad \, \forall \, \lambda>0. 
\end{equation*}
where $  \x_\lambda= (\lambda+x^2)^{1/2}$.

 Let $\mu >0, \eps>0, \lambda>0$, and define 
$$
F(x) = F_{\mu,\eps,\lambda}(x) =  \frac\mu\eps \left(1-e^{-\eps\, \x_\lambda}\right)\, , 
$$
so that 
$$
\nabla F_{\mu,\eps,\lambda}(x) = x g_{\mu,\eps,\lambda}(x), \quad  g_{\mu,\eps,\lambda}(x) = \frac{\mu\, e^{-\eps\, \x_\lambda}}{\sqrt{\lambda+|x|^2}}\, .
$$
Denote $\psi_{\mu,\eps,\lambda} = e^{F_{\mu,\eps,\lambda}} \, \psi$. Lemma \ref{lem-lowerb} and equation \eqref{eq-psiF-2} then give
\begin{align} \label{lowerb-3}
 \kappa\,  \La\psi_{\mu,\eps,\lambda}, |\nabla F_{\mu,\eps,\lambda}|^2\, \psi_{\mu,\eps,\lambda} \Ra &  \leq 
 \La\psi_{\mu,\eps,\lambda}, \big((x\cdot \nabla)^2 g_{\mu,\eps,\lambda} -x\cdot \nabla |\nabla F_{\mu,\eps,\lambda}|^2  \big) \psi_{\mu,\eps,\lambda} \Ra
 + C\, \|\psi_{\mu,\eps,\lambda}\|_2^2
\end{align}
for all $\mu,\eps,\lambda >0$ and  some constant $C$ independent of $\mu, \lambda$ and $\eps$.  
Moreover, a direct calculation shows  
\begin{equation}
\lim_{\eps\to 0} \, x\cdot \nabla |\nabla F_{\mu,\eps,\lambda}(x)|^2  = 2\lambda \mu^2 \x_\lambda^{-1} (1\, -\x_\lambda^{-2}) \, >0\, 
\end{equation} 
and
\begin{equation}
\lim_{\eps\to 0} \, (x\cdot \nabla)^2  g_{\mu,\eps,\lambda}(x) = -2\lambda \mu \x_\lambda^{-3} |x|^2 <0 \, .
\end{equation}
Since 
$$
\lim_{\eps\to 0} F_{\mu,\eps,\lambda}(x) := F_{\mu,\lambda}(x) = \mu \x_\lambda\, ,
$$
in view of Proposition \ref{prop-decay} we can pass to limit $\eps\to 0$ in \eqref{lowerb-3} to obtain
\begin{equation} \label{lowerb-4}
 \kappa\,  \mu^2\, \La\psi_{\mu,\lambda}, \frac{|x|^2}{\lambda+|x|^2}\, \psi_{\mu,\lambda} \Ra \, \leq 
 C\, \|\psi_{\mu,\lambda}\|_2^2 \qquad \forall \, \mu, \lambda >0 , 
\end{equation}
where 
$$
\psi_{\mu,\lambda}(x) := e^{\mu \x_\lambda}\, \psi(x)\, .
$$
Using Proposition \ref{prop-decay}  again and the monotone convergence theorem we finally obtain, by letting $\lambda\to 0$, 
\begin{equation}
 \kappa\,  \mu^2\,\|\psi_{\mu}\|_2^2\, \leq 
 C\, \|\psi_{\mu}\|_2^2 \qquad \forall \, \mu >0 , 
\end{equation}
where $\psi_\mu (x)=e^{\mu |x|}\, \psi(x)$. This is of course impossible for $\mu$ large enough. 
Hence $\psi_\mu=0$.The the first part of the claim, i.e.~the absence of eigenvalues above $\Lambda$, thus follows from the case $E >\Lambda$. The second 
part of the claim is covered by the case  $E+\ol{\mu}^2>\Lambda$ for some $\ol{\mu}>0$.  
\end{proof}

\begin{remark} \label{rem-ah-cash}
Notice that in view of Corollary \ref{cor-ess-free magnetic}  we have $(\Lambda,\infty) \subseteq \sigma_{\text{ess}}(H)$. Hence Theorem  \ref{thm-abs} excludes the presence of all embedded eigenvalues of $H$ strictly larger than $\Lambda$. 

On the other hand, the possibility of $\Lambda$ being an eigenvalue of $H$ cannot be in general excluded. Indeed, if $B$ is continuous and compactly  supported with $|\int_{\R^2} B |> 2\pi$, and if $V=-B$, then by the Aharonov-Casher theorem, see e.g.~\cite[Sec.~6.4]{cfks}, $\Lambda=0$ is an eigenvalue of $H=(P-A)^2-B$. Sufficient conditions for the absence of positive eigenvalues of the Pauli operator are proved in Section \ref{ssec-pauli} , see \ref{cor-pauli}.

\end{remark}


\section{Examples}  
\label{sec-examples} 
We recall a couple of examples  which show that the decay assumptions on $B$ and $V$ stated in Theorems \ref{thm:typical},
\ref{thm:typical-form}, and \ref{thm-abs}, and Proposition \ref{prop-pointwise} cannot be improved. 

\subsection{Miller-Simon revisited} \label{ssec-miller-simon}
In \cite{ms} Miller and Simon considered, in dimension two, the case $V=0$ and radial magnetic field $B(x) =b(r), \, r=|x|$. They proved that 

\medskip

\begin{SE}
  \item\label{MS-1} If $b(r) = r^{-\alpha} + \Oh(r^{-1-\veps})$ with 
   	$0<\alpha<1$ and $\veps>0$ then the spectrum of $H$ is dense pure point,
  \item\label{MS-2} If $b(r) = b_0\, r^{-1} + \Oh(r^{-1-\veps})$ for some $\veps>0$ then the spectrum of $H$ is dense pure point in $[0, b_0^2)$ and absolutely continuous in $[b_0^2, \infty)$,
  \item \label{MS-3} If $b(r) = \mathcal{O}(r^{-\alpha})$ with $\alpha >1$ then  the spectrum of $H$ is purely absolutely continuous in $(0,\infty)$. 
\end{SE}
\begin{remark} 
   Note that $\beta= \limsup_{|x|\to \infty}|\wti{B}(x)|= +\infty$ in the case (1). On the other hand, Theorem \ref{thm-abs} guarantees the absence of eigenvalues in the interval 
   $(b_0^2, \infty)$ for the case (2), where  
   $\beta=b_0$, and in the interval $(0,\infty)$ for the case (3), even for non--radial magnetic fields. 
   In particular, the Miller--Simon examples show that our result on absence of eigenvalues is \emph{sharp}. 
\end{remark}
Since there is a calculation error in the original Miller-Simon paper and also in the book \cite{cfks}, we sketch their argument: Assume that the radial magnetic field $b$ is reasonable, e.g., bounded and use $x,y$ as coordinates in $\R^2$ and $r=(x^2+y^2)^{1/2}$. 
The first observation of Miller and Simon is that if the magnetic field, radial or not, $B$ goes pointwise to zero at infinity, 
then $\sigma_{\text{ess}}((P-A)^2)=[0,\infty)$ (this is sharpened in 
Theorem \ref{thm-ess-free magnetic}). 
For radial magnetic fields we have $\wti{B}(x)= (-y,x)b(r)$, so the Poincar\'e gauge the magnetic vector potential is 
\begin{align*}
	A(x,y)= (-y,x) \int_0^1 b(tr)t\, dt = \frac{(-y,x)}{r} h(r)
\end{align*}
with $h(r)=r^{-1}\int_0^r b(s)s\, ds$. Expanding $(P-A)^2$ one sees 
\begin{align*}
	(P-A)^2= (P_x-A_x)^2 + (P_y-A_y)^2 
		= P^2 + h(r)^2 - 2\frac{h(r)}{r} L
\end{align*}
where $L=xP_y-yP_x$ is the angular momentum in the plane. 
It is well-known that $L$ has eigenvalues $(0,\pm1,\pm2,\ldots)$ and it commutes with $P^2$ and with the radial potential $h(r)^2$. 
So restricted to the angular momentum channel $\{L=m\}$, the operator $(P-A)^2$ is given by 
\begin{align*}
	H_m\coloneqq (P-A)^2\big|_{\{L=m\}}= \big(P^2+ V_m\big)|_{\{L=m\}}
	\quad \text{with } V_m(r)= h(r)^2 - \frac{2mh(r)}{r}
\end{align*} 
Due to the angular momentum barrier the divergence of $V_m$ for small $r$ when $m\not= 0$ is irrelevant.

If $b_0 = \lim_{r\to\infty}r\, b(r)=\infty$, then $h(r)\to \infty$, so $V_m$ is trapping and all operators $H_m$ have discrete spectrum. 
But if also $b(r)\to 0$ as $r\to\infty$, then $\sigma_{\text{ess}}(H_A)=[0,\infty)$, so $(P-A)^2$ has necessarily 
dense point spectrum in $[0,\infty)$, proving the first claim (\ref{MS-1}) above. 

If $b_0 = \lim_{r\to\infty}r\, b(r)<\infty$, then $h(r)\to b_0$ and $V_m(r)\to b_0^2$ 
as $r\to\infty$, so $H_m$ has only discrete spectrum below $b_0^2$ for any $m\in\Z$. Since $b(r)\to 0$ for $r\to\infty$, the operator has essential spectrum $[0,\infty]$, which must be dense point spectrum in $[0,b_0^2]$. 
(or any reasonable choice of radial magnetic field $b$, the effective potential $V_m$ is smooth with decaying derivatives for large $r$, so the spectrum of $H_m$ above $b_0^2$ is absolutely continuous for all $m\in\Z$. Thus $(P-A)^2$ has absolutely continuous spectrum in $(b_0^2,\infty)$.
\begin{remark}
	In \cite{ms} the choice of the vector potential contains a wrong factor of $1/2$ and in the example in \cite{cfks} there is a mistake in the calculation of the magnetic field. 
	Thus in their examples they concluded incorrectly that the effective potential has the asymptotic  $V_m(r)\to b_0^2/4$ for large $r$.  
\end{remark}

\medskip

\smallskip

\noindent

\subsection{Wigner-von Neumann  potential} Suppose that $B=0$. Wigner and von Neumann showed that the operator  $-\Delta +V$ in $L^2(\R^3)$ with the radial potential
\begin{equation} \label{wvn}
V(r)= -\frac{32 \sin r \big[ g(r)^3\cos r -3g^2(r)\sin^3 r+g(r)\cos r+\sin^3(r)\big]}{(1+g(r)^2)^2}\, , \qquad g(r) = 2r -\sin(2r),
\end{equation} 
has eigenvalue $+1$, see \cite{s2, wvn} and \cite[Ex.~VIII.13.1]{rs4}.  As pointed out in \cite{s2} for large $r$ 
\begin{equation} \label{wvn-2}
V(r) =-\frac{8 \sin(2r)}{r} + \mathcal{O}(r^{-2}). 
\end{equation} 
Theorem \ref{thm-abs} implies that $-\Delta +V= -\Delta +V_1+V_2$ has no eigenvalues larger than 
$$
\Lambda = \frac 12 \Big( \omega_1 +\omega_2 +\sqrt{\omega_1^2 +2\omega_1 \omega_2} \ \Big),
$$
with $\omega_1$ and $\omega_2$ defined by equation \eqref{limits}. 
We can thus optimize the splitting $V=V_1+V_2$ in order to minimize $\Lambda$. A quick calculation using \eqref{wvn-2} shows that the optimal choice is $V_1=0,\,  V_2=V$, see also Lemma \ref{lem-bang-bang}. With this choice we get $\Lambda = 8$ which coincides with \cite[Thm.~4]{ag}. Note that \cite[Thm.~2]{s2} implies absence of eigenvalues in the interval $(16,\infty)$.

\smallskip

\noindent The Wigner-von Neumann  example was further generalized in \cite{au} where  Arai and Uchiyama constructed, for each $|k|>2$, 
bounded radial potentials which are asymptotically of the form 
\begin{align}
	V(x) = \frac{k\sin(2|x|)}{|x|} + \Oh(|x|^{-1-\veps}) 
	\quad \text{as } |x|\to\infty
\end{align} 
for some $\veps>0$ such that $P^2+V$ has eigenvalue $1$. 
In these  examples also $x\cdot\nabla V$ is bounded and 
$\omega_1=\limsup_{|x|\to\infty}(x\cdot\nabla V(x))_+ = 2|k|$.  
Thus we can conclude that $P^2+V$ has no eigenvalues $E>|k|^2/2$.


\subsection{Aharonov Bohm vector potentials} 
\label{ssec-AB}
In two dimensions the prototypical  Aharonov Bohm magnetic vector potential 
is given by 
\begin{align}
	A^{ab}(x,y)= \frac{(-y,x)}{x^2+y^2}\ B_0 \, ,
\end{align} 
for some $B_0\in\R$. This yields a locally square integrable vector potential on $\R^2\setminus\{0\}$, it corresponds to a singular 
magnetic field, which is concentrated 
in zero, i.e., $B=\partial_x A^{ab}_y-\partial_y A^{ab}_x= 0$ in $\R^2\setminus\{0\}$, but for any smooth curve $S$ circling once 
around zero, the line integral along $S$ is given by  
\begin{align*}
	\int_{S} (A_xdx+A_ydy) = 2\pi B_0
\end{align*}
that is, the `magnetic field' corresponding to $A$ has total flux $2\pi B_0$. The corresponding magnetic Schr\"odinger 
operator $H^{ab}_0$ is now defined as the closure of the  
quadratic form $q_{ab,0}$ defined first on 
$\calC^\infty_0(\R^2\setminus\{0\})$ as 
\begin{align*}
	q_{ab,0}(\varphi,\varphi) = \La (P-A^{ab})\varphi,(P-A^{ab})\varphi ) \Ra
\end{align*}
and for any potential $V$ which is  form small w.r.t.~$H^{ab}_0$, the operator $H^{ab}_V$ is defined as the form sum 
\begin{align*}
	q_{ab,V}(\varphi,\varphi) \coloneqq q_{ab,0}(\varphi,\varphi)+ \La \varphi,V\varphi\Ra \, .
\end{align*}
For such type of singular magnetic Schr\"odinger operators we still have a virial theorem and a result on absence of positive eigenvalues for the following simple reasons:

For dilation, it makes no difference if one works on $\R^2$ or on 
$\R^2\setminus\{0\}$. Thus  we can still use dilations to derive a virial theorem. In fact, this is easy.  The first thing one has to check if $\calD(P-A^{ab})$ is invariant under dilations.
Recall equation \eqref{eq-key-dilation-inv}, which for the Aharonov Bohm vector potential reads 
\begin{align}\label{eq-key-dilation-Aharonov-Bohm}
	(P-A^{ab})U_t\varphi 
	  = e^tU_tP\varphi - U_t A^{ab}_t\varphi = e^tU_t(P-A^{ab})\varphi + U_t(e^t A^{ab} -A^{ab}_{-t})\varphi  \, 
	  =  e^tU_t(P-A^{ab})\varphi 
\end{align}
since, the Aharonov Bohm vector potential is homogeneous of degree $-1$, we have  $e^t A^{ab} -A^{ab}_{-t}=0$ for all $t>0$. That is, the Aharonov Bohm magnetic momentum operator 
$P-A^{ab}$ has the same commutation properties with dilations as the free momentum $P$, which drastically simplifies the analysis! 

\begin{theorem}[Aharonov Bohm magnetic virial theorem] \label{thm:ABmagnetic-virial} 
Let $A^{ab}$ be the Aharonov Bohm vector potential and $V$ 
satisfy assumptions \ref{ass-V-form small}. 
Assume also that the distribution $x\cdot\nabla V$ 
extends to a quadratic form which is form bounded with respect to 
$(P-A^{ab})^2$.  
Then for all $\varphi\in \calD(P-A^{ab})$, the limit $\lim_{t\to0}2\re\big(q_{ab}(\varphi, i D_t \varphi)\big) $ exists. Moreover,
\begin{equation}  \label{eq-AB-virial}
	\begin{split}
		\La\varphi, [H^{ab}_V, iD] \, \varphi \Ra 
			&\coloneqq \lim_{t\to 0} 2\re \big(q_{ab,V}(\varphi, iD_t \varphi)\big) 
				= 2 \|(P-A^{ab}) \varphi\|_2^2  
					- \La \varphi, x\cdot\nabla V\varphi \Ra \, .
	\end{split}	
\end{equation}
\end{theorem}
This is proven exactly as Theorem \ref{thm:magnetic-virial}, 
the extra term from the magnetic field disappears because of 
the scaling of  the Aharonov Bohm vector potential. 
Of course, this theorem then also implies absence of positive eigenvalues under the same conditions on the potential $V$ as in Theorem \ref{thm-abs}, now with $\beta=0$. For the Aharonov--Bohm Hamiltonian $H^{ab}_V$ there are no eigenvalues $E$ with 
\begin{align}\label{eq-AB threshold}
 	E> \frac{1}{4}\left(
  				  \omega_1 +\sqrt{\omega_1^2+2\omega_2}
  				\right)^2
\end{align}

\begin{remarks}
  \itemthm One can also allow for an angular dependence in the Aharonov--Bohm type potential as in \cite{lw}.\\[2pt]
  \itemthm In addition to the Aharonov--Bohm potential, one can 
  		also allow for an additional regular magnetic field 
  		$B$ satisfying assumptions \ref{ass-B-mild-int} 
  		and \ref{ass-B-rel-bounded}. One has to modify the 
  		right hand sides of \eqref{eq-AB-virial} and of 
  		\eqref{eq-AB threshold}. \\[2pt]   	
\end{remarks}


\subsection{ Pauli and magnetic Dirac operators} 
\label{ssec-pauli}
\noindent In this section we state two consequences of Theorem \ref{thm-abs} and Proposition \ref{prop-pointwise}. Let $B:\R^2\to \R$ be given and  consider the Pauli operator 
$$
P(A) =  \begin{pmatrix}
(i\nabla+A)^2+B& 0  \\
0 &   (i\nabla +A)^2-B
\end{pmatrix}
 $$
in $L^2(\R^2, \C^2)$. It is well-known that the operator $P(B)$ is non-negative, and that if $|\int_{\R^2} B | > 2\pi$, then zero is an eigenvalue of $P(B)$, see also Remark \ref{rem-ah-cash}. 

\begin{corollary} \label{cor-pauli}
Assume that $B\in L^p_{\rm loc}(\R^2)$ for some $p>2$ and that $B(x) = \mathcal{O}(|x|^{-1})$ as $|x|\to \infty$. Let $A\in L^2_{\rm loc}(\R^2)$ be such that ${\rm curl}\, A=B$. 
Then the operator $P(A)$ has no eigenvalues in the interval $(4\beta^2,\infty)$, with $\beta$ given by \eqref{limits}. 

\smallskip

\noindent If moreover there exists a compact set $K\subset\R^2$ such that $B\in C^1(\R^2\setminus K)$, then the operator $P(A)$ has no eigenvalues in the interval 
$\big (\Lambda_P, \infty \big)$, with 
\begin{equation}\label{lambda-p}
   \Lambda_P:= \min \Big\{4\beta^2, \frac 14 (\beta+\omega +\sqrt{(\beta+\omega)^2 +2\omega} )^2\Big\} 
\end{equation}
and
$$
\omega = \max\Big\{ \limsup_{|x|\to\infty}\,  x\cdot \nabla B(x), \, - \liminf_{|x|\to\infty}\,  x\cdot \nabla B(x) \Big\}\, .
$$
\end{corollary}

\smallskip

\begin{proof}
The first part of the statement follows from Theorem \ref{thm-abs} and Proposition \ref{prop-pointwise} applied to the components of the Pauli operator with the splitting $V_1=\pm B, \, V_2=0$. The second part follows from the first part and from the application of Theorem \ref{thm-abs} and Proposition \ref{prop-pointwise} with the splitting $V_1=0, \, V_2=\pm B$.
\end{proof}

\begin{remark} A couple of comments are in order:
	\itemthm  The example of Miller and Simon \cite{ms}, 
		 see Section \ref{ssec-miller-simon},  
	applies to two-dimensional Pauli operators as well. In particular, a quick inspection shows that if $B(r) = b_0\, r^{-1} + \mathcal{O}(r^{-2})$, then the spectrum of $P(A)$ is dense pure point in $[0, b_0^2)$ and absolutely continuous in $[b_0^2, \infty)$. Note that in this case Corollary \ref{cor-pauli} guarantees the absence of eigenvalues for $P(A)$ in the interval $(b_0^2, \infty)$, so this result is \emph{sharp}. \\[2pt]
	\itemthm  Under the hypotheses of Corollary \ref{cor-pauli} the essential spectrum of $P(A)$ coincides with $[0,\infty)$, see Corollary \ref{cor-ess-free magnetic} below. \\[2pt]
	\itemthm    Using the main results of our paper, 
		one can significantly relax the regularity 
		assumption on the magnetic field $B$. 
		We leave this to the interested reader. \\[2pt]
	\itemthm Absence of positive eigenvalues of the Pauli operator in $\R^3$ will be treated elsewhere. 
\end{remark}

\medskip 

\noindent The second application of Theorem \ref{thm-abs} concerns magnetic Dirac operators in $L^2(\R^2, \C^2)$ which in the standard representation have the form 
\begin{equation} \label{dirac-op}
\mathbb{D} = \begin{pmatrix}
m&  {\rm Q}\\
{\rm Q}^*  &   -m
\end{pmatrix}
\, ,
\qquad {\rm Q}= (P_1-A_1) +i(P_2-A_2)\, ,
\end{equation}
where $m$ is the mass of the particle. We have

\medskip

\begin{corollary} \label{cor-dirac}
Let $B$ satisfy the assumptions of Corollary \ref{cor-pauli} and let $A\in L^2_{\rm loc}(\R^2)$ be such that ${\rm curl}\, A=B$. 
Then the Dirac operator $\mathbb{D}$ defined on $\mathcal{D}(P-A)$ has no eigenvalues in 
$$
\big(-\infty, -\sqrt{\Lambda_P +m^2} \ \big) \, \cup \, \big( \sqrt{\Lambda_P +m^2}, \, \infty\big ).
$$

\begin{proof}
Note that 
\begin{equation} 
\mathbb{D}^2 = P(A) + m^2\id
\end{equation}
  in the sense of sesqui-linear forms on $\mathcal{D}(P-A)\otimes\mathcal{D}(P-A)$. Hence if $\mathbb{D} \psi= E \psi$ for some $\psi\in\mathcal{D}(P-A)\otimes\mathcal{D}(P-A)$, then $\psi$ is a weak eigenfunction of $P(A)$ relative to eigenvalue $E^2 -m^2$. In view of Corollary \ref{cor-pauli} we thus have $E^2-m^2 \leq \Lambda_P$.
\end{proof}

\end{corollary}

\begin{remark}
Sufficient conditions for the absence of the entire point spectrum of Pauli and Dirac operators with electromagnetic fields where recently found in \cite{cfk}, see also Remark \ref{rem-comments}.\ref{rem-fks}. 
\end{remark}

\medskip


\section{The essential spectrum}
\label{sec-ess} 

\noindent We have the following dichotomy. 

\begin{lemma}[Dichotomy] \label{lem-ess-dichotomy}
Let $A\in L^2_{\rm loc}(\R^d, \R^d)$. Then either $\inf\sigma((P-A)^2) >0$ or $\sigma((P-A)^2) = [0, \infty)$. 
\end{lemma}
\begin{remark}\label{rem-ess-landau}
  The Landau Hamiltonian, where the vector potential $A$ corresponds to a constant magnetic field, provides an example where 
  $\inf\sigma((P-A)^2) >0$, see \cite{Landau-lifshitz}. 
\end{remark}
\begin{proof} Write $H_0=(P-A)^2$. 
It suffices to prove the implication 
\begin{equation} \label{enough}
0 \in \sigma(H_0) \quad \Rightarrow \quad \sigma(H_0) = [0, \infty).
\end{equation} 
Let $D(H_0)$ denote the domain of $H_0$. To prove \eqref{enough} suppose that $0 \in \sigma(H_0)$. Hence there exits a sequence $\{\widetilde\varphi_n\}_{n\in\N} \subset D(H_0)$ such that $\|\widetilde\varphi_n\|_2 =1$ for all $n\in\N$ and 
\begin{equation} \label{phi-tilde}
\| H_0\,  \widetilde\varphi_n\|_2 \ \to \ 0 \qquad n\to\infty.
\end{equation}
Now we define 
\begin{equation} \label{phi}
\phi_n(x) = e^{i k\cdot x}\, \widetilde\varphi_n(x), 
\end{equation}
where $k\in\R^d$ is arbitrary. Then $\phi_n\in D(H_0)$ for every $n\in\N$, and we have
\begin{align*}
(P-A) \, \phi_n(x) & = e^{i k\cdot x}\, (P-A+k)\, \widetilde\varphi_n(x),
\end{align*}
and
\begin{align*}
H_0\, \phi_n(x)  &= (P-A)^2 \, \phi_n(x)  = e^{i k\cdot x}\, H_0 \, \widetilde\varphi_n(x) +2 e^{i k\cdot  x} k\cdot (P-A) \widetilde\varphi_n(x) +|k|^2 \phi_n(x) .
\end{align*}
with the derivatives meant in the sense of distributions.
Since $\|\widetilde\varphi_n\|_2 =1$, it follows that $H_0\, \phi_n\in L^2(\R^d)$. Hence $\phi_n\in D(H_0)$. Moreover the above calculations and the Cauchy-Schwarz inequality show 
that 
\begin{align*}
\| (H_0-|k|^2) \, \phi_n\|_2 & \leq \| H_0\,  \widetilde\varphi_n\|_2 +2 |k|\, \|(P-A) \widetilde\varphi_n\|_2 \, = \,  \| H_0\,  \widetilde\varphi_n\|_2 +2 |k|\, \sqrt{\La\widetilde\varphi_n, H_0\,  \widetilde\varphi_n \Ra}\\
& \leq \, \| H_0\,  \widetilde\varphi_n\|_2 + 2 |k|\, \| H_0\,  \widetilde\varphi_n\|_2^{1/2}. 
\end{align*}
By \eqref{phi-tilde} we thus have $\| (H_0-|k|^2) \, \phi_n\|_2 \to 0$ for any $k\in\R^d$. Hence $[0, \infty) \subseteq  \sigma_{\text{ess}}(H_0)$ and since $H_0\geq 0$, we conclude that $\sigma(H_0) = \sigma_{\text{ess}}(H_0)=[0, \infty)$. 
\end{proof}

\noindent Next we formulate a condition on $B$ under which $\sigma((P-A)^2) = [0,\infty)$ for any locally square integrable vector potential $A$ with $B=dA$. 

\begin{definition}[Vanishing somewhere at infinity] \label{def-vanishing somewhere}
We say that the magnetic field $B$ \emph{vanishes somewhere at infinity} if there exist sequences $\{R_n\}_{n\in\N}\subset \R$ and $\{x_n\}_{n\in\N}\subset \R^d$  such that 
$R_n \to \infty, \ |x_n| \to \infty$ as $n\to \infty$, and  
\begin{equation} \label{vanish} 
\lim_{n\to \infty}\,  R_n^{-d}\!\! \int_{\U_{R_n}} \left ( \frac{|y|}{R_n} \right)^{2-d}\,  \left(1-\frac{|y|}{R_n}\right)^2\, \left (\log \frac{R_n}{|y|} \right)^2\, \big | \wti B_{x_n} (y)\big |^2\, dy =0.
\end{equation} 
\end{definition}

\smallskip

\begin{remark}\label{rem-B-pointwise vanishng}
This vanishing condition is quite weak. For example, if $d=2$ and if $B$ decays uniformly in a cone $S_\omega$ with an opening angle $\omega\in (0, \pi)$, meaning that $\sup_{ |\theta | < \omega} |B(r,\theta)| \to 0$ as $r\to\infty$, then $B$ vanishes somewhere at infinity. Indeed, given a sequence  $R_n \to \infty$ one can choose $x_n=(x^n_1,0)$ with $x^n_1$ growing fast enough, depending on $B$ and $R_n$, such that $ \U_{R_n}(x_n) \subset S_\omega$ for all $n$, and such that 
  \begin{align*}
  	  R_n^{-d}\!\! \int_{\U_{R_n}(0)}  \left (\log \frac{R_n}{|y|} \right)^2\, \big | \wti B_{x_n}(y)\big |^2\, dy 
  	  & \ \lesssim\   R_n^2\, \Big(\,  \sup_{\U_{R_n}(x_n)} |B|^2\Big)  \,    
  	    \to 0\, , \qquad n\to\infty.
  \end{align*}
  Also, we do not require that the magnetic 
  field $B=dA$ exists  as a classical vector field outside 
  the sequence of balls $\U_{R_n}(x_n)$.   
\end{remark}

\smallskip

\begin{theorem} \label{thm-ess-free magnetic} 
Suppose that $A$ is a locally square integrable magnetic 
vector potential such that the magnetic field $B=dA$ 
vanishes somewhere at infinity in the sense of 
Definition \ref{def-vanishing somewhere}. Then  
$$
\sigma((P-A)^2) = \sigma_{\rm ess}((P-A)^2) =[0,\infty). 
$$
\end{theorem}
\begin{remark}
	In case that the magnetic field goes to zero pointwise 
	at infinity, the above result was already shown by Miller 
	and Simon, \cite{ms, cfks}. As pointed out in \cite{ms} 
	the invariance of the essential spectrum is quite 
	remarkable, since the the vector potential $A$ corresponding 
	to the magnetic field $B$ might not have any decay at 
	infinity, i.e., the magnetic kinetic energy $(P-A)^2$ is 
	not a small perturbation of the non--magnetic kinetic 
	energy $P^2$, in general.  
\end{remark}

\begin{proof}
Let $R_n$ and $x_n$ be the sequences defined in Definition \ref{def-vanishing somewhere} and let 
\begin{equation} \label{An}
A_n(x) = \int_0^1  B(x_n +t(x-x_n))\,[t(x-x_n)]  \, dt
\end{equation}
be the vector potential related to $B$ via the Poincar\'e gauge centred at $x_n$. Then $\rt \, A_n=\rt\, A =B$ for all $n\in\N$, and 
therefore there exits  a scalar gauge field 
$\chi_n \in H^1_{\rm loc}(\R^2;\R)$ with $\nabla\chi_n\in L^2(\R^2)$ such that 
\begin{equation}
A_n = A -\nabla \chi_n  ,
\end{equation}
and for all $\varphi\in L^2(\R^2)$ with  $(P-A_n)\varphi\in L^2(\R^2)$ we have $e^{i\chi}\varphi\in\calD(P-A)$ 
and $(P-A)e^{i\chi_n}\varphi= e^{i\chi_n}(P-A_n)\varphi$, 
see \cite{leinfelder}. To simplify the notation we denote 
$$
\U_n=\U_{R_n}(x_n).
$$ 

Due to the Dichotomy Lemma \ref{lem-ess-dichotomy} we only have to 
show that $0\in\sigma((P-A)^2)$. To this end we will construct a 
sequence $\{\phi_n\}_{n}\subset \calD(P-A)$ with $\supp(\phi_n\in\U_n)$ and  $\|\phi_n\|_2=1$ such that
\begin{equation} \label{enough2}
 \|(P-A)\, \phi_n\|_2^2  \ \to \ 0 \qquad n\to \infty.
\end{equation}
We choose $\phi_n = e^{i\chi_n}\,  \varphi_n,$ where
$$
\varphi_n(x) = C_d\,  R_n^{-\frac d2}  \Big(\frac{|x-x_n|}{R_n}\Big)^{\frac{2-d}{2}}\,  \left(1-\frac{|x-x_n|}{R_n}\right)_+\, ,
$$ 
where the constant  $C_d$ depends only on $d$ and is chosen such that $\|\phi_n\|=\|\varphi_n\|=1$. 
Then by the above gauge invariance  
\begin{equation} \label{p-an}
 \|(P-A)\, \phi_n\|_2^2 
   = \|(P-A_n)\, \varphi_n\|_2^2 \, 
   \le   
   \big(\|P\varphi_n\|+ \|A_n\varphi_n\|\big)^2.
\end{equation}
We have
\begin{align*}
 \|P \varphi_n\|_2^2 \ \lesssim\  R_n^{-2} \ \to \ 0 \qquad n\to\infty.
\end{align*}
Hence by setting $h(s) = (1-s/R)^2_+$ in \eqref{eq-AB-trans} we obtain, in view of \eqref{An},
\begin{align} 
\| A_n  \varphi_n\|_2^2&\,  \lesssim\, R_n^{-2} \int_{\U_n}  \left(1-\frac{|x-x_n|}{R_n}\right)^2\,  |x-x_n|^{2-d}\, |A_n(x)|^2\, dx\nonumber  \\
& \leq\,  4 R_n^{-d}  \int_{\U_{R_n}(0)}   \left(\frac{|y|}{R_n}\right)^{2-d}\, \left(1-\frac{|y|}{R_n}\right)^2\,  \log^2(R_n/|y|)\,  \big | \wti B_{x_n}(y)\big |^2\,   dy. \label{4-ineq}
\end{align}
Thus the assumption that $B$ vanishes somewhere at infinity  
implies $\| A_n \varphi_n\|_2 \, \to 0\, $ as $n\to \infty$. 
By \eqref{p-an} this shows 
\begin{align*}
  \|(P-A)\, \phi_n\|_2^2 \to \ 0 
\end{align*}
as $n\to\infty$, which proves \eqref{enough2}. Since $\|\phi_n\|_2=\|\varphi_n\|_2=1$ for all $n\in\N$, it follows that $0\in \sigma(H_0)$ and  applying Lemma  \ref{lem-ess-dichotomy} then gives
$\sigma_{\text{ess}}((P-A)^2)=[0,\infty)$.  
\end{proof}

\noindent To prove that $B$ vanishes somewhere at infinity it is convenient to impose an additional condition on $B$: 
\begin{assumption} \label{ass-ess}
Suppose that there exist $	\kappa>0$ and sequences $\{x_n\}\subset \R^d, \{R_n\}\subset  \R_+, \{\alpha_n\} \subset  \R_+$ and $\{\gamma_n\}\subset\R_+$ such that $|x_n|\to \infty, \, R_n\to \infty, \alpha_n\to 0, \, \gamma_n\to 0$, and such that
\begin{equation} \label{ass-B1-add}
\La \varphi , \, | \cdot -x_n|^\kappa\, \big| \wti B_{x_n} (\cdot - x_n)\big|^2\, \varphi \Ra \, \leq \, \alpha_n \|(P-A)\varphi\|_2^2 +\gamma_n \|\varphi\|_2^2
\end{equation}    
for all $\varphi\in D(P-A)$ with supp$\, \varphi\subset \U_{R_n}(x_n)$.
\end{assumption}

\begin{corollary} \label{cor-ess-free magnetic}
Suppose that the magnetic field satisfies Assumptions
\ref{ass-B-mild-int} and \ref{ass-ess}.
Then for any locally square integrable magnetic vector 
potential $A$ with $dA=B$ we have 
\begin{equation} \label{eq-ess-free}
 \sigma_{\text{ess}} ((P-A)^2) =  [0,\infty)\, .
\end{equation} 
\end{corollary} 

\begin{proof}
Let $\wti R_n, x_n, \alpha_n$ and $\gamma_n$ be the sequences given by Assumption \ref{ass-ess}. We define 
$$
u_n = \wti R_n^{\, -\frac d2}\,  \left(\frac{|x-x_n|}{\wti R_n}\right)^{\frac{2-d+\kappa}{2}}\, \left(1-\frac{|x-x_n|}{\wti R_n}\right)_+\,  \log_+(\wti R_n/|x-x_n|)\,  \, .
$$
and
$$
\wti C_n =  \wti R_n^{\, \kappa}\,  \La u_n, |\cdot-x_n|^{-\kappa} \, \big| \wti B_{x_n} (\cdot - x_n)\big|^2\, u_n\Ra\, .
$$
Note that $u_n \in L^2(\R^d), |\nabla u_n| \in L^2(\R^d)$ and 
$$
\|u_n\|_2 = \mathcal{O}(1), \qquad \|P\, u_n\|_2 =  \mathcal{O}(\wti R_n^{\, -1}) \qquad n\to \infty. 
$$
Hence by \eqref{ass-B1-add} and \eqref{An}
\begin{align}\label{cn-tilde}
\wti C_n & \leq \,  2 \wti R_n^{\, \kappa}\,  \alpha_n \|P\, u_n\|_2^2 + 2  \wti R_n^{\, \kappa}\,  \alpha_n \| A_n u_n\|_2^2 + \wti R_n^{\, \kappa}\, \gamma_n \|u_n\|_2^2
\lesssim \,  \wti R_n^{\, \kappa-2}\,  \alpha_n + \wti R_n^{\, \kappa}\, \gamma_n  + \wti R_n^{\, \kappa}\,  \alpha_n \| A_n u_n\|_2^2.
\end{align}
Now, the bound $\sup_{0<s<1} s^\kappa |\log s|^2 <\infty$ in combination with inequality  \eqref{eq-AB-trans} implies 
\begin{align*}
\| A_n u_n\|_2^2 & =  \wti R_n^{\, -d} \int_{\U_{\wti R_n}} \left(\frac{|x-x_n|}{\wti R_n}\right)^{2-d+\kappa}\, \left(1-\frac{|x-x_n|}{\wti R_n}\right)^2 \, \Big( \log_+(\wti R_n/|x-x_n|)\Big)^2\, |A(x_n+y)|^2\, dy \\
& \lesssim \wti R_n^{\, -d} \int_{\U_{\wti R_n}} \left(\frac{|x-x_n|}{\wti R_n}\right)^{2-d}\, \left(1-\frac{|x-x_n|}{\wti R_n}\right)^2 \,  |A(x_n+y)|^2\, dy\\
& \lesssim \wti R_n^{\, -d} \int_{\U_{\wti R_n}} \left(\frac{|y|}{\wti R_n}\right)^{2-d}\, \left(1-\frac{|y|}{\wti R_n}\right)^2 \, \Big( \log_+(\wti R_n/|y|)\Big)^2\,  |B_{x_n}(y)|^2\, dy = \wti C_n\, .
\end{align*}
Inserting this into \eqref{cn-tilde} gives 
\begin{equation*} 
 \wti C_n \, \lesssim \,   \wti R_n^{\, \kappa-2}\,  \alpha_n + \wti R_n^{\, \kappa}\, \gamma_n  + \wti R_n^{\, \kappa}\,  \alpha_n \, \wti C_n\, .
\end{equation*}
At this point we redefine 
\begin{equation} \label{rn-defin}
R_n  := \min \big\{  \wti R_n, \, \alpha_n^{-\frac{1}{2\kappa}}, \, \gamma_n^{-\frac{1}{2\kappa}}\, \big\} ,
\end{equation}
and 
$$
\varphi_n(x)  =  R_n^{-\frac d2}\,  \left(\frac{|x-x_n|}{R_n}\right)^{\frac{2-d}{2}}\, \left(1-\frac{|x-x_n|}{R_n}\right)_+\,  \log_+(R_n/|x-x_n|)\, .
$$
Note that $R_n\to \infty$ as $n\to\infty$. 
Repeating the above bounds with $\wti R_n$ replaced by $R_n$  and $u_n$ replaced by $\varphi_n$ then leads to 
\begin{align*}
C_n &  := \La \varphi_n, \, \wti B_{x_n}(\cdot -x_n)\, \varphi_n\Ra \, \lesssim\, R_n^{\, \kappa-2}\,  \alpha_n +  R_n^{\, \kappa}\, \gamma_n  + R_n^{\, \kappa}\,  \alpha_n \, C_n\, \lesssim\, R_n^{-2}\,  \sqrt{\alpha_n}  + \sqrt{\gamma_n}   +  \sqrt{\alpha_n}  \ C_n,
\end{align*}
where, in the second step, we have used $R_n^{ \kappa}\,  \alpha_n  \leq  \sqrt{\alpha_n} $ and $R_n^{ \kappa}\,  \gamma_n  \leq  \sqrt{\gamma_n} $, which follows from \eqref{rn-defin}. Hence $C_n\to 0$ as $n\to \infty$, and since 
$$
C_n =  R_n^{-d}\! \int_{\U_{R_n}} \left ( \frac{|y|}{R_n} \right)^{2-d}\, \left(1-\frac{|y|}{R_n}\right)^2\, \left (\log \frac{R_n}{|y|} \right)^2\, \big | \wti B_{x_n} (y)\big |^2\, dy,
$$
the claim follows from Theorem \ref{thm-ess-free magnetic}.
\end{proof}

\begin{remark}
The condition imposed by Assumption \ref{ass-ess} rather weak. Indeed, if we set 
$$
W_n(x) = \id_{\U_n}(x) \, |x-x_n|^{-\kappa}\, \big|\wti B_{x_n}(x-x_n)\big |^2, 
$$
then by the discussion in Section \ref{ssec-kato-class} it follows that \eqref{ass-B1-add} holds with
$$
\alpha_n = \| (P^2+\lambda)^{-1} \, W_n\|_\infty, \qquad \gamma_n =\lambda  \| (P^2+\lambda)^{-1} \, W_n\|_\infty, 
$$
and any $\lambda >0$.  Hence  Assumption \ref{ass-ess}  is satisfied whenever 
$$
\lim_{n\to \infty} \| (P^2+\lambda)^{-1} \, W_n\|_\infty = 0.
$$
\end{remark}

\begin{theorem} \label{thm-ess-V vanishing}
Suppose that $A$ is a locally square integrable magnetic vector potential and the potential $V$ is form small and vanishes at infinity w.r.t~$(P-A)^2$. Then 
\begin{equation} \label{eq-ess-1}
 \sigma_{\rm ess} (H_{A,V}) =  \sigma_{\rm ess} ((P-A)^2).
\end{equation} 
\end{theorem}

\begin{proof}
\noindent Since $V$ is form small with respect to $(P-A)^2$, the quadratic form $q_{A,V}$ is closed and bounded from  below on the form domain $\calD(P-A)$.  
Hence there exists $ s \geq 1$ such that 
the operators $H_{A,0}+s$ and $H_{A,V}+s$ are invertible in $L^2(\R^d)$. We are going to prove that the resolvent difference 
\begin{equation} \label{res-dif}
(H_{A,0}+s)^{-1} - (H_{A,V}+s)^{-1}  \quad \text{is compact in }  \ L^2(\R^d). 
\end{equation}
for some large enough $s\ge 1$, which by Weyl's theorem implies that 
the essential spectra of $H_{A,V}$ and $(P-A)^2$ coincide. 
In the following, we will abbreviate $H_0=H_{A,0}$. 
Let $C_s\coloneqq (H_0+s)^{-1/2}V(H_0+s)$, more precisely, 
$C_s$ is the bounded operator associated with the bounded form 
\begin{align*}
  \La \varphi, C_s \varphi \Ra
	\coloneqq \La (H_0+s)^{-1/2}\varphi, V(H_0+s)^{-1/2}\varphi \Ra
	  =  q_V((H_0+d\lambda)^{-1/2}\varphi,(H_0+s)^{-1/2}\varphi)\, ,
\end{align*}
and the relative form bound of $V$ w.r.t~$(P-A)^2$ is given by $\lim_{s\to\infty}\|C_s\|_{2\to2}<1$, see \cite[Theorem 6.24]{teschl}, \cite{rs4}. Choose $\lambda$ large enough, such that $\|C_\lambda\|<1$. Then Tiktopoulos' formula, \cite[Chapter II.3]{simon-thesis}, \cite[Theorem 6.25]{teschl} shows 
\begin{align*}
	 (H_{A,V}+s)^{-1} 
	 	=  (H_0+s)^{-1/2}(1-C_s)^{-1} (H_0+s)^{-1/2} \, .
\end{align*}
Hence 
\begin{align*}
	(H_0+s)^{-1}- (H_{A,V}+s)^{-1} 
	  = (H_0+s)^{-1/2}(1-C_s)^{-1}C_s (H_0+s)^{-1/2}  
\end{align*}
so we only have to show that 
\begin{align*}
	C_s (H_0+s)^{-1/2}  
	  =  (H_0+s)^{-1/2}  V (H_0+s)^{-1}  
\end{align*}
is a compact operator. For this let 
$\xi_{< R},\xi_{\ge R}$ the smooth partition from 
the proof of Lemma \ref{lem-punch-2} with $\xi_{< R}^2+\xi_{\ge R}^2=1$, $\supp(\xi_{< R})\subset \U_{2R}$,
$\supp(\xi_{\ge R})\subset \U_R^{\, c}$, and 
$\|\nabla\xi_{< R}\|_\infty, \|\nabla\xi_{\ge R}\|_\infty\lesssim R^{-1}$. With 
\begin{align}
  	 J_{<R}
  	 	& \coloneqq   (H_0+s)^{-1/2} \xi_{< R}^2 V (H_0+s)^{-1} \label{eq-compact-1}\\
  	 J_{\ge R}
  	 	&\coloneqq 
  	 	     (H_0+s)^{-1/2}  \xi_{\ge R}^2V (H_0+s)^{-1} \label{eq-compact-2}
\end{align}
we obviously have 
$
  (H_0+s)^{-1/2}  V (H_0+s)^{-1}  = J_{<R}+J_{\ge R}\, .
$

We will show that 
$\lim_{R\to\infty}\|J_{\ge R}\|_{2\to 2}=0$. So 
$(H_0+d\lambda)^{-1/2}  V (H_0+d\lambda)^{-1}  $ is 
the norm limit of $J_{<R}$ as $R\to\infty$, in particular, it is 
a compact operators if $J_{<R}$ is compact for all large $R$. 
Since  
\begin{align}\label{eq-operator norm}
	\|J_{\ge R}\|_{2\to 2}=\sup_{\|f\|=1}|\La f, J_{\ge R}f \Ra|
\end{align}
and with $\varphi=(H_0+s)^{-1/2}f$ 
\begin{align*}
	|\La f, J_{\ge R}f \Ra|
	  &= |\La \xi_{\ge R}\varphi, V\xi_{\ge R}\varphi \Ra| 
	    \le \alpha_R \|(P-A)\xi_{\ge R}\varphi\|_2^2 + \gamma_R\|\xi_{\ge R}\varphi\|_2^2 \\
	  &\le  \alpha_R \Big(\|(P-A)\varphi\|+\|\nabla\xi_{\ge R}\|\|\varphi\|\Big)^2 + \gamma_R\|\varphi\|_2^2 
	  	\lesssim 
	  	 \big( \alpha_R (1+R^{-1})^2+ \gamma_R\big)\|f\|_2^2
\end{align*}
since $(P-A)\xi_{\ge R}\varphi= \xi_{\ge R}(P-A)\varphi- i(\nabla\xi_{\ge R})\varphi$,  
$\|(P-A)\varphi\|\le \|f\|$ and $\|\varphi\|\le s^{-1}\|f\|$. 
From this and \eqref{eq-operator norm} one immediately gets 
$\|J_{\ge R}\|_{2\to 2}\lesssim \alpha_R (1+R^{-1})^2+ \gamma_R\to 0$ for $R\to\infty$. 
\smallskip

To prove that $J_{<R}$ is compact, we first note that the domain 
of $H_0=(P-A)^2$ is given by all $\varphi\in\calD(P-A)$ for 
which with $\psi=(P-A)\varphi$ the distribution $(P-A)\psi$ is  
also in $L^2(\R^d)$. Thus for all $\varphi\in \calD((P-A)^2)$ we have 
\begin{align*}
	(H_0+s)^{-1} (P-A+i\lambda)\cdot(P-A-i\lambda)\varphi 
		= 	(H_0+s)^{-1}	(H_0+d\lambda^2)\varphi = \varphi 
\end{align*}
when $s=d\lambda^2$. Moreover, when $\varphi\in\calD((P-A)^2)$ 
and $\chi$ is a bounded $\calC^2$ function such that 
$\nabla\chi$ and $\Delta\chi$ are bounded, then  
\begin{align*}
	(P-A-i\lambda)\chi\varphi
	  &= \chi(P-A-i\lambda)\varphi -i(\nabla\chi)\varphi 
	  		\in L^2(\R^d)\, ,\\
    (P-A+i\lambda)\cdot(P-A-i\lambda)\chi\varphi
      &= \chi (P-A+i\lambda)\cdot(P-A-i\lambda)\varphi 
      	- 2i (\nabla\chi)\cdot(P-A)\varphi - (\Delta\chi)\varphi \\
      & = \chi(H_0+d\lambda^2)\varphi 
      	    - 2i (\nabla\chi)\cdot(P-A)\varphi - (\Delta\chi)\varphi
      	\in L^2(\R^d)
\end{align*}
so also $\chi\varphi\in \calD((P-A)^2)$. 

Use $\varphi=(H_0+s)^{-1}f$ with $f\in L^2(\R^d)$ and choose $d\lambda^2=s$. Then the last equality yields   
\begin{align*}
	\chi(H_0+s)^{-1} f 
		&= \chi\varphi= (H_0+s)^{-1} (P-A+i\lambda)\cdot(P-A-i\lambda)\chi\varphi \\
		&= (H_0+s)^{-1} \chi f 
			-2i (H_0+s)^{-1} (\nabla\chi)\cdot(P-A)(H_0+s)^{-1}f 
			- (H_0+s)^{-1} (\Delta\chi)(H_0+s)^{-1}f\, . 
\end{align*}
Setting $\chi=\xi_{<R}^2$ one sees that $J_{<R}$ can be written as 
\begin{equation}\label{eq-sum of compact ops}
  \begin{split}
	J_{<R} =
	  C_s\, \Big(
			J_1 -  2i J_2\cdot(P-A)(H_0+s)^{-1} 
			 - J_3(H_0+s)^{-1}
			\Big)\, .
  \end{split}
\end{equation}
where we abbreviated $J_1= (H_0+s)^{-1/2} \chi $, $J_2= (H_0+s)^{-1/2} (\nabla\chi)$, and $J_3= (H_0+s)^{-1/2} (\Delta\chi)$. 

Note that $C_s$ is bounded and so are $(P-A)(H_0+s)^{-1}$ and 
$(H_0+s)^{-1}$. Moreover, since $\chi=\xi_{<R}^2$ has 
compact  support, it is well--know that the operators 
$\chi(P^2+s)^{-1/2}, (\nabla\chi)(P^2+s)^{-1/2}$, and 
$ (\Delta\chi)(P^2+s)^{-1/2}$ are compact operators on $L^2(\R^d)$, 
see  \cite[Thm.~5.7.3]{da2}, for example. 
The diamagnetic inequality and the 
Dodds--Fremlin--Pitt theorem \cite{df,pitt} then imply that the operators   
$\chi(H_0+s)^{-1/2}, (\nabla\chi)(H_0+s)^{-1/2}$, and $ (\Delta\chi)(H_0+s)^{-1/2}$ are also compact, and by duality so are 
$J_1, J_2$, and $J_3$. Thus by \eqref{eq-sum of compact ops} 
the  operator $J_{<R}$ is a compact operator for all $R>0$.  
\end{proof}

\begin{corollary} \label{cor-B and V vanishing}
Suppose that $B$ satisfies Assumptions \ref{ass-B-mild-int}, \ref{ass-ess}, and that $V$ satisfies Assumptions \ref{ass-V-form small} and \ref{ass-bounded infinity}. Then 
\begin{equation} \label{eq-ess}
 \sigma_{\rm ess} (H_{A,V}) = [0, \infty).
\end{equation} 
\end{corollary}
\begin{proof}
  Combine Theorem \ref{thm-ess-free magnetic} and Corollary \ref{cor-ess-free magnetic}. 	
\end{proof}


\appendix


\section{Pointwise conditions and the Kato--class}  
\label{sec-pointwise} 
Below we show that  Assumptions \ref{ass-B-mild-int}-\ref{ass-bounded unique continuation} are satisfied under mild explicit  regularity and decay conditions on the magnetic field $B$ and the potential $V$. In particular, we give local $L^p$ conditions, which in a natural way extend the pointwise bounds on the potential 
from in \cite{ag,s2}. 

\subsection{Uniformly local $L^p$ conditions}
\label{ssec:uniform local lp}
Recall that the space $L_{loc,unif}^p(\R^d)$ of uniformly local real-valued 
$L^p$ functions is given by (measurable) functions  
$f:\R^d\to \R\cup\{-\infty,\infty\}$ such that for $1\le p<\infty $ 
\begin{equation}
	\|f\|_{L^p_{loc,unif}}\coloneqq \sup_{x\in\R^d} \left( \int_{|x-y|\le 1} |f(y)|^p \, \mathrm{d}y \right)^{1/p} <\infty\,  ,
\end{equation}
with the obvious replacement for $p=\infty$, 
${L^\infty_{\rm{loc},\rm{unif}}}(\R^d) = L^\infty(\R^d)$.  

We note that unlike the $L^p$ spaces, the spaces 
$L_{\rm{loc},\rm{unif}}^p(\R^d)$ are nested in the sense 
that for $1\le q\le p\le \infty$ one has 
$L_{\rm{loc},\rm{unif}}^p(\R^d)\subset L_{\rm{loc},\rm{unif}}^q(\R^d)\subset L_{\rm{loc},\rm{unif}}^1(\R^d)$.

\begin{proposition} \label{prop-uniform local Lp form tiny}
 Let  $p= 1$ if $d=1$ and $p>d/2$ when $d\ge 2$. 
 If  $|\wti B|^2 \in L_{\rm{loc},\rm{unif}}^p(\R^d)$ 
 then Assumptions \ref{ass-B-mild-int} and \ref{ass-B-rel-bounded} are satisfied. 
 If  $V\in  L_{\rm{loc},\rm{unif}}^p(\R^d)$ then 
 Assumption \ref{ass-V-form small} is satisfied. 
 Moreover, assume that one can split $V=V_1+V_2$, where 
 the distributional derivative $x\cdot\nabla V_2$ is 
 given by a function and     
 $V_1, x^2V_1^2, x\cdot\nabla V_2 \in L_{\rm{loc},\rm{unif}}^p(\R^d)$  
 then Assumptions \ref{ass-V-split form bounded} and \ref{ass-bounded unique continuation} are satisfied.  
\end{proposition}
 
 Before we prove this, we give a simple additional 
 pointwise condition on   $|\wti B|, V_1, V_2$ which 
 guarantees that the remaining assumptions on being 
 bounded at infinity are satisfied.

 We say that $V$ is bounded at infinity, if there exists a 
 compact set $K\subset\R^d$ such that 
 $V\in L^\infty(\R^d\setminus K)$. 
 We also say that $V$ goes to zero pointwise at infinity, 
 if it is bounded at infinity and  
 $\limsup_{x\to\infty}|V(x)|=0$.  
\begin{proposition} \label{prop-pointwise} 
 Assume that $V$ goes to zero pointwise at infinity  
 and that $V$ splits as $V=V_1+V_2$ where 
 $V_1$ goes to zero pointwise at infinity. Then Assumption \ref{ass-V-vanishing infinity} is satisfied. 
 
 Moreover, if the distribution $x\cdot\nabla V_2$ is given 
 by a function and $|\wti B|$, 
 $x^2V_1^2, x\cdot\nabla V_2$ 
 are bounded from above at infinity, then  Assumption 
 \ref{ass-bounded infinity} is satisfied and we have 
 the bounds 
\begin{equation} \label{beta-omega-old}
\beta \, \leq \, \limsup_{|x|\to\infty} |\wti B(x)|, \quad  \omega_1 \, \leq \, \limsup_{|x|\to\infty} |x\, V_1(x)|^2, \quad \omega_2\, \leq \, 
\limsup_{|x|\to\infty}\,  x\cdot \nabla V_2(x)\, \, .
\end{equation} 
\end{proposition}

\begin{remark}
  So, under the assumptions of Propositions 
  \ref{prop-uniform local Lp form tiny} and 
  \ref{prop-pointwise} all our Assumptions \ref{ass-B-mild-int}-\ref{ass-bounded unique continuation} are satisfied 
  and the upper bounds from \eqref{beta-omega-old} 
  hold for $\beta, \omega_1$, and $\omega_2$.  
  Of course, the above pointwise conditions are in general way too strong. 
  Below we show how some of the assumptions of Proposition 
  \ref{prop-pointwise} can be relaxed for potentials 
  $V\in L_{\rm{loc},\rm{unif}}^p(\R^d)$, or even potentials in 
  the Kato-class. 

\end{remark}
\begin{proof}[Proof of Proposition \ref{prop-pointwise}]
 Let $W$ be bounded at infinity and   
 set $M\coloneqq \limsup_{x\to\infty} W(x)$. Given 
 $\delta>0$ there exists $R=R_\delta <\infty$ such that 
 \begin{equation}
 	\la \varphi, W\varphi \ra \le (M+\delta) \|\varphi\|^2 
 \end{equation}
 for any $\varphi\in L^2(\R^d)$ with 
 $\supp(\varphi)\subset \U_{R_\delta}^c$. From this observation, the 
 claims of Proposition \ref{prop-pointwise} follow.
\end{proof}
\begin{proof}[Proof of Proposition \ref{prop-uniform local Lp form tiny}] 
Of course, magnetic fields exist only in dimensions $d\ge 2$. 
Nevertheless, for any  $d\ge 1$ and  $w\in\R^d$, Assumption \ref{ass-B-mild-int} follows  from $|\wti B|^2 \in L_{\rm loc}^{p}(\R^d)$, with $p=1$ 
if $d=1$ and  $p >d/2$ when $d\ge 2$, by a simple application 
of H\"older's inequality.

In the following, let $p\ge 1$ for $d=1$ and $p>d/2$ when 
$d\ge 2$.  
It is well-know, at least for specialists, that a potential 
$V\in L^p(\R^d)$ is infinitesimally form bounded with 
respect to $P^2$. That is, for any choice $\alpha_0>0$ there 
exists $C_\veps<\infty$ such that 
\begin{equation}\label{eq-form-tiny-non-mag}
	|\la \varphi, V\varphi \ra| 
	\le \la\varphi, |V|\varphi \ra 
	\le \veps\|P\varphi\|^2 +C_\veps \|\varphi\|^2 
	\qquad \text{for all } \varphi\in\calD(P)\, .
\end{equation}
Using the diamagnetic inequality this also implies 
\begin{equation}\label{eq-form-tiny-magnetic}
	|\la \varphi, V\varphi \ra| 
	\le \la|\varphi|, |V||\varphi| \ra 
	\le \veps\|(P-A)\varphi\|^2 +C_\veps \|\varphi\|^2 
	\qquad \text{for all } \varphi\in\calD(P-A)\, ,
\end{equation}
for any  
$A\in L^2_{\rm{loc}}(\R^d,\R^d)$. 
Less known is the fact that \eqref{eq-form-tiny-non-mag}, 
hence also \eqref{eq-form-tiny-magnetic}, continue to 
hold for  $V\in L_{\rm{loc}, \rm{unif}}^{p}(\R^d)$. 
This follows, for example, from the fact that under the 
above conditions on $p$ in terms of the dimension $d$ 
one knows that $L_{\rm{loc,unif}}^p(\R^d)\subset K_d$, 
where $K_d$ is the Kato--class of potentials, and all 
potentials $V\in K_d$ are infinitesimally form bounded 
w.r.t.\  $P^2$, see \cite{cfks,simon-semigroups-review} .

Given this observation, one sees that Assumption 
\ref{ass-V-form small} is satisfied when 
$V\in L^p_{\rm{loc},\rm{unif}}(\R^d)$ and Assumptions 
\ref{ass-V-split form bounded} and 
\ref{ass-bounded unique continuation} are satisfied when 
we split $V=V_1+V_2$ with  $V_1, x^2V_1^2, x\cdot\nabla V_2 
\in L^p_{\rm{loc},\rm{unif}}(\R^d)$.  
This proves all claims of 
Proposition \ref{prop-uniform local Lp form tiny}.  

\smallskip
However, in order to derive a simple local $L^p$ condition 
for a potential to vanish at infinity w.r.t.\ $P^2$ we 
need a quantitative bound for dependence of 
the constant $C_\veps$ in the bounds 
\eqref{eq-form-tiny-non-mag} and  
\eqref{eq-form-tiny-magnetic} depends on $\veps$ and 
on the norm $\|V\|_{L_{\rm{loc}, \rm{unif}}^{p}}$. 
For this reason, and the convenience of the reader, we sketch 
the derivation of a quantitative version of 
\eqref{eq-form-tiny-non-mag}: 

If $p\ge 1$ for $d=1$ and $p>d/2$ when $d\ge 2$, an 
argument similar to the proof of Theorem X.20 in 
\cite{rs4} shows that there exists a function 
$G:\R_+\times\R_+\to [0,\infty)$, with $G(s_1,s_2)$ separately 
increasing in $(s_1,s_2)\in R_+^2$ and 
$\lim_{s_2\to 0} G(s_1,s_2)=0$ for all $s_1>0$, such that 
\begin{equation}\label{eq-form-tiny-quantitative-1}
	|\la \varphi, V\varphi \ra| 
	\le \la\varphi, |V|\varphi \ra 
	\le \veps\|P\varphi\|^2 +G\big(\veps^{-1}, \|V\|_p\big) \|\varphi\|^2 
	\qquad \text{for all } \varphi\in\calD(P)\, .
\end{equation}
Indeed, H\"older's inequality gives 
$\la \varphi, |V|\varphi\ra \le \|V\|_p \|\varphi\|_{\frac{2p}{p-1}}^2  $. 
Since $\frac{2p}{p-1}\ge 2$ the Hausdorff--Young inequality 
shows 
$\la \varphi, |V|\varphi\ra  \le \|V\|_p \|\widehat{\varphi}\|_q^2 $ with $\widehat{\varphi}$ the Fourier transform 
of $\varphi$ and $\frac{1}{q} = 1-\frac{p-1}{2p} 
= \frac{1}{2} +\frac{1}{2p} \le 1$. Let $t>0$ and write 
$\widehat{\varphi}= (1+t\eta^2)^{-1/2}(1+t\eta^2)^{1/2}\widehat{\varphi}$. 
Since $\frac{1}{q}= \frac{1}{2p} +\frac{1}{2}$ we can 
use again H\"older's inequality to get  
\begin{equation*}
  \|\widehat{\varphi}\|_q^2
    \le 
      \|(1+t\eta^2)^{-1/2}\|_{2p}^2 
      \|(1+t\eta^2)^{1/2}\widehat{\varphi}\|_2^2 
      \le C_{p,q} t^{-d/2p} (2t\|P\varphi\|_2^2 + 2\|\varphi\|_2^2) 
\end{equation*} 
with $C_{p,d}= \|(1+\eta^2)^{-1/2}\|_{2p}^2$. Note that 
$\R^d\ni \eta\mapsto (1+\eta^2)^{-1/2} \in L^{2p}(\R^d)$ 
for any $p\ge 1$ if $d=1$ and $p>d/2$ if $d\ge 2$.  
Altogether we have 
\begin{equation*}
  \la \varphi, |V|\varphi\ra 
    \le 
    	2C_{p,d} \|V\|_p t^{-d/(2p)}
    		(t\|P\varphi\|^2 +\|\varphi\|^2)
\end{equation*}
for any $t>0$ and all $\varphi\in\calD(P)$. Rescaling in 
$t>0$ one sees that a bound of the form 
\eqref{eq-form-tiny-quantitative-1} holds with 
$G(s_1,s_2)= C s_1^{-\frac{d}{2p-d}} s_2^{\frac{2p-d}{2p}}$ 
for some constant $C$ depending only on $d$ and $p$. 

\smallskip
Now we extend this to potentials 
$V\in L^p_{\rm{loc,unif}}(\R^d)$. 
Let $\chi\in\calC^\infty_0(\R^d)$ with $0\le \chi\le 1$ 
and $\chi(x)=1$ for $\|x\|_\infty\le 3/2$ and 
$ \chi(x)=0$ when $\|x\|_\infty \ge 2$. 
For $j\in\Z^d$ define 
$\chi_j(x)= \chi(x-j)$ for $x\in\R^d$. 
Then $\chi_j\in\calC^\infty_0(\R^d)$ for all 
$j\in\Z^d$. Since the supports of the $\chi_j$ have the 
finite intersection property, there exist a 
constant $c>1$ such that 
$  1\le \sum_{j\in\Z^d} \chi_j^2\le c $.
Moreover,  $\sum_{j\in\Z^d} \chi_j^2 \in\calC^\infty(\R^d)$ 
and all partial derivatives of  $\sum_{j\in\Z^d} \chi_j^2 $ 
are bounded functions. 
We define 
\begin{equation}
	\xi_j\coloneqq \frac{\chi_j}{(\sum_{k\in\Z^d}\chi_k^2)^{1/2}}\, .
\end{equation} 
Since $1\le \sum_{k\in\Z^d}\chi_k^2\le c_1 $ the cutoff 
functions $\xi_j$ are well--defined and 
$\xi_j\in\calC^\infty_0(\R^d)$.  
By construction 
\begin{equation}\label{eq-quadratic partition of unity}
	\sum_{j\in\Z^d} \xi_j^2 = 1\, .
\end{equation}
Hence the family of cutoff functions  
$(\xi_j)_{j\in\Z^d}$ is a smooth quadratic partition of 
unity. Using again that the supports of the $\chi_j$ 
have the finite intersection property, it is also  
easy to see that 
there exists a constant $0<L<\infty$ such that 

\begin{equation}\label{eq-bounded localization error}
  \sum_{j\in\Z^d} |\nabla \xi_j|^2 \le L . 
\end{equation}
Lastly, let $K_j=\supp(\xi_j) = \supp(\chi_j)$ and 
notice that there exist $0<\kappa<\infty$ 
such that 
\begin{equation}\label{eq-eqivalent norm}
	\sup_{j\in\Z^d} \|\id_{K_j}V\|_p
		\le 
			\kappa \|V\|_{L_{\rm{loc,unif}}^p} 
\end{equation} 
for all $V\in L_{\rm{loc,unif}}^p(\R^d)$. 
In fact, it is straightforward to show that the two 
norms in \eqref{eq-eqivalent norm} are equivalent.

Given  $\varphi\in\calC^\infty_0(\R^d)$, we have 
$\varphi=\sum_{j\in\Z^d}\xi_j^2\varphi$ because of 
\eqref{eq-quadratic partition of unity}. 
Note also that we can arbitrarily rearrange this sum, and 
similar sums below, since $\supp(\xi_j)\cap\supp(\varphi)\neq \emptyset$ for only finitely many $j\in\Z^d$. 
In particular, we have  $\la \varphi, |V|\varphi\ra = 
\sum_j \la \xi_j\varphi, |V_j|\xi_j\varphi\ra  $ with $V_j=\id_{K_j}V$. Using \eqref{eq-form-tiny-quantitative-1} one gets  
\begin{align}\label{eq-form-tiny-quantitative-localized}
	|\la \varphi, V\varphi\ra| 
		\le \la \varphi, |V|\varphi\ra 
		\le \sum_{j\in\Z^d} 
				\left( 
					\veps \|P(\xi_j\varphi)\|_2^2 
					  + G\big(\veps^{-1}, \|V_j\|_p\big)
					  	\|\xi_j\varphi\|_2^2
					\right)
\end{align}
Because of \eqref{eq-eqivalent norm} and since $G$ is 
increasing in its second variable, we have  
$
	\sup_{j\in\Z^d}G\big(\veps^{-1}, \|V_j\|_p\big) 
	\le G\big(\veps^{-1}, \kappa \|V\|_{L_{\rm{loc,unif}}^p} \big) 
$
for some constant $0<\kappa<\infty$ and all $V\in L_{\rm{loc,unif}}^p(\R^d)$. Moreover, because of  
 \eqref{eq-quadratic partition of unity} we have 
\begin{equation*}	
  \sum_{j\in\Z^d} \|\xi_j\varphi\|^2 
    = \la \varphi, \xi_j^2\varphi \ra 
    = \la \varphi, \sum_{j\in\Z^d} \xi_j^2\varphi \ra 
    = \|\varphi\|^2 \, .
\end{equation*}
The IMS localization formula \ref{thm-ims} together with  \eqref{eq-quadratic partition of unity} and \eqref{eq-bounded localization error} yields    
\begin{align*}
	\sum_{j\in\Z^d} \|P(\xi_j\varphi)\|^2 
		&= \sum_{j\in\Z^d} \la P(\xi_j\varphi), P(\xi_j\varphi)\ra 
		= \sum_{j\in\Z^d} \left(\re\la P(\xi_j^2\varphi), P\varphi\ra + \la \varphi, |\nabla\xi_j|^2\varphi  \ra \right) \\
		&= \re\la P(\sum_{j\in\Z^d} \xi_j^2\varphi), P\varphi\ra + \la \varphi, \sum_{j\in\Z^d} |\nabla\xi_j|^2\varphi  \ra 
		  \le  \|P\varphi\|_2^2 + L\|\varphi\|_2^2 \, . 
\end{align*}
Using  \eqref{eq-form-tiny-quantitative-localized}
we arrive at  
\begin{align}
\label{eq-form-tiny-quantitative-locally-uniform}
	|\la \varphi, V\varphi\ra| 
		\le \la \varphi, |V|\varphi\ra 
		\le  	
			\veps \|P\varphi \|_2^2 
					  + \left(\veps L+ G\big(\veps^{-1}, \kappa \|V\|_{L_{\rm{loc,unif}}^p}\big)\right)
					  	\|\varphi\|_2^2
\end{align}
for all $\varphi\in\calC^\infty_0(\R^d)$ and all  
$\veps>0$, as soon as a local bound of the form 
\eqref{eq-form-tiny-quantitative-1} holds. 
Since $\calC^\infty_0(\R^d)$ is dense in $\calD(P)$ with 
respect to the graph norm, the bound \eqref{eq-form-tiny-quantitative-locally-uniform} extends to  all 
$\varphi\in\calD(P)$. This shows that any potential 
$V\in L_{\rm{loc,unif}}^p(\R^d)$ is infinitesimally 
form bounded w.r.t.\ $P^2$.  
The bound 
\eqref{eq-form-tiny-quantitative-locally-uniform} also holds 
with 
$P$ replaced by $P-A$ and $\varphi\in \calD(P-A)$ 
for any vector potential 
$A\in L^2_{\text{loc}}(\R^d,\R^d)$ thanks to the 
diamagnetic inequality \eqref{diamagnetic-inequality}
\end{proof}

\subsection{Potentials vanishing at infinity}\label{sec:vanishing}
\noindent
Recall the Definitions \ref{def-vanishing}, respectively \ref{def-bounded infinity}, for a potential 
$V$ to vanish, respectively being bounded,  at infinity 
w.r.t.~$(P-A)^2$. 
Assume that $V$ can be split as $V=W_1+W_2$ with 
the quadratic form domains $\calQ(W_1)$ and 
$\calQ(W_2)$ containing, for all large enough $R>0$,  all 
$\varphi\in \calD(P-A)$ with $\supp(\varphi)\subset \U_R^c$.   
Then it is straightforward to see that 
\begin{align*}
	\gamma_R(V) \le \gamma_R(W_1) +\gamma_R(W_2) 
	\text{ and }
	\gamma^+_R(V)  \le \gamma^+_R(W_1) +\gamma^+_R(W_2) 
\end{align*}  
for all large enough $R>0$. Hence 
\begin{align}\label{eq-subadditivity-1}
  0\le \gamma_\infty(V) = \lim_{R\to\infty} \gamma_R(V) 
  \le \gamma_\infty(W_1) +\gamma_\infty(W_2) 
\end{align}
and 
\begin{align}\label{eq-subadditivity-2}
  \gamma^+_\infty(V) = \lim_{R\to\infty} \gamma^+_R(V) 
  \le \gamma^+_\infty(W_1) +\gamma^+_\infty(W_2) \, .
\end{align}
If $W_1$ and $W_2$ vanish at infinty w.r.t.\ $(P_A)^2$, then 
$\gamma_\infty(W_1) =\gamma_\infty(W_2)=0 $. Thus 
$\gamma_\infty(V)=0$, that is, $V$ vanishes at infinity 
w.r.t\ $(P-A)^2$. Moreover, the bound 
\eqref{eq-subadditivity-2} shows that $V$ is bounded   
from above at infinity w.r.t.\ $(P-A)^2$ with upper bound 
$\gamma^+_\infty(W_1) +\gamma^+_\infty(W_2)$. 

\begin{proposition}\label{prop-vanishing-properties}
  \begin{SEalph}
  	\item If $V=W_1+W_2$ and $W_1$ and $W_2$ vanish at infinity, respectively, are bounded from above at infinity, w.r.t.~$(P-A)^2$,  
  		  then $V$ vanishes at infinity, respectively, is bounded from above at infinity, w.r.t.~$(P-A)^2$. Moreover, in the latter case \eqref{eq-subadditivity-2} holds. 

\smallskip
 
  	\item If $V= \nabla\cdot\Sigma$ for some real-valued vector field $\Sigma$  and if $\Sigma^2$  is form bounded respectively vanish at infinity w.r.t.~$(P-A)^2$, then $V$ is form bounded respectively vanishes at infinity  w.r.t.~$(P-A)^2$.
  \end{SEalph}
\end{proposition}
\begin{remarks}
	\itemthm 
	  It is not true, in general, that $\Sigma^2$ bounded at infinity implies that $\nabla\cdot\Sigma$ is bounded at infinity w.r.t~$(P-A)^2$. \\[2pt] 
	\itemthm  The choice $\Sigma(x)= x\x^{-\veps}\sin(e^{1/|x|})= \Oh(\x)^{-\veps}$, for some $\veps>0$, yields a potential 
	$V=\nabla\cdot\Sigma$ with 
	\begin{align}
		V(x)= -|x|^{-1} e^{1/|x|}\x^{-\veps}\cos(e^{1/|x|}) + \Oh(\x^{-\veps})
	\end{align}
	which has a severe singularity at zero. Since $\Sigma^2$ is infinitesimally for bound and vanishing at infinity w.r.t $P^2$, the above result shows that so does $V$. That $V$ vanishes at infinity w.r.t.~$P^2$, which might not be too  surprising, since the singularity is local. \\[2pt]
	\itemthm The choice   $\Sigma(x)= x\x^{-\veps}\sin(e^{|x|})= \Oh(\x)^{-\veps}$, for some $\veps>0$, yields a potential $V=\nabla\cdot\Sigma$ with 
	\begin{align}
		V(x)= |x| e^{|x|}\x^{-\veps}\cos(e^{|x|}) + \Oh(\x^{-\veps})
	\end{align}
	which has again severe oscillations, now at infinity. Nevertheless, it is infinitesimally form bounded and vanishes at infinity  w.r.t.~$P^2$ since $\Sigma^2$ does. In particular, despite the severe oscillations of $V$ at infinity, our Theorem \ref{thm-ess-V vanishing} below shows that the perturbation $V$ does not change the essential spectrum.  
\end{remarks}
 
\begin{proof}
	The first claim was already proven in the discussion just before the proposition.  For the second claim let 
	$\varphi\in\calC^\infty_0$,  and note that Lemma \ref{lem-combescure-ginibre} shows that the distribution 
	$\nabla\cdot\Sigma$ yields the quadratic form  
	\begin{align*}
		\La \varphi, \nabla\cdot\Sigma\varphi \Ra 
			= -2\im \La \Sigma\varphi, P\varphi \Ra 
			= -2\im \La \Sigma\varphi, (P-A)\varphi \Ra 
	\end{align*}
 since $ \La \Sigma\varphi, A\varphi \Ra$ is real. Thus the right hans side above extend to all $\varphi\in\calD(P-A)$ if $\Sigma^2$ is form bounded 
 w.r.t.~$(P-A)^2$ and  
 $|\La \varphi, \nabla\cdot\Sigma\varphi \Ra |\le \|\Sigma\varphi\|\|(P-A)\varphi\|$. So if 
 $\|\Sigma\varphi\|_2^2\le \alpha\|(P-A)\varphi\|_2^2+\gamma\|\varphi\|_2^2$, then 
 \begin{equation}\label{eq-CG trick}
   \begin{split}
 	|\La \varphi, \nabla\cdot\Sigma\varphi \Ra |
 		&\le 2( \alpha\|(P-A)\varphi\|_2^2+\gamma\|\varphi\|_2^2)^{1/2} \|(P-A)\varphi\|
 		\le (\veps^{-1}\alpha+\veps)\|(P-A)\varphi\|_2^2
 			  + \veps^{-1} \gamma \|\varphi\|_2^2 
   \end{split}
\end{equation}
 for all $\veps>0$, which proves that $\nabla\cdot\Sigma$ is form bounded w.r.t. $(P-A)^2$. If $W$ is also form bounded w.r.t\ $(P-A)^2$, then so is their sum $V=\nabla\cdot\Sigma +W$. 
 
 Lastly, because of the first part, we only have to show that $\nabla\cdot\Sigma$ vanishes at infinity as soon as $\Sigma^2$ vanishes at infinity w.r.t~$(P-A)^2$.  So assume that there exist $\alpha_R$ and $\gamma_R$ decreasing with  $\alpha_R, \gamma_R\to 0$ as $R\to\infty$ and  
 \begin{align*}
 	\|\Sigma\varphi\|_2^2 \le \alpha_R\|(P-A)\varphi\|_2^2+\gamma_R\|\varphi\|_2^2
 \end{align*}
 for all $\varphi\in\calD(P-A)$ with $\supp(\varphi)\in \U_R^{\, c}$. Setting $\veps=\max(\alpha_R,\gamma_R)^{1/2}$ in \eqref{eq-CG trick} yields 
 \begin{align*}
 	|\La \varphi, \nabla\cdot\Sigma\varphi \Ra |
 		\le \max(\alpha_R,\gamma_R)^{1/2}\Big(2\|(P-A)\varphi\|_2^2 + \|\varphi\|_2^2\Big) 
 \end{align*}
  for all $\varphi\in\calD(P-A)$ with $\supp(\varphi)\subset \U_R^{\, c}$ and large enough $R$. 
  This shows that $\nabla\cdot \Sigma$ vanishes at infinity w.r.t.~$(P-A)^2$.
\end{proof}

\begin{remark}\label{rem-calculation asympt upper bounds}
  The two bounds \eqref{eq-subadditivity-1} and 
  \eqref{eq-subadditivity-2} also show that 
  \begin{align}\label{modified at infinity-1}
    \gamma_\infty(V) = \inf\{\gamma_\infty(V-W):\,  \gamma_\infty(W)=0\} 
\end{align}
and 
\begin{align}\label{modified at infinity-2}
  \gamma^+_\infty(V) =  \inf\{\gamma^+_\infty(V-W):\,  \gamma^+_\infty(W)=0\}  \, .
\end{align}
As upper bounds these statements follow immediately from 
 \eqref{eq-subadditivity-1} and 
  \eqref{eq-subadditivity-2}. The reverse inequality follows 
  by choosing $W=0$. 
  Thus when trying to calculate the asymptotic bounds 
  $\beta,\omega_1$, $\omega_2$ 
  from \eqref{eq-asymptotoc bounds}, see also 
  Assumption \ref{ass-bounded infinity} one can modify the 
  involved potentials by arbitrary vanishing potentials.  
  Efficient criteria for this are derived in the next two sections.  
\end{remark}

\subsection{A local $L^p$ condition for vanishing at infinity}\label{ssec-cond local vanishing}

We say that $\lambda$ is locally uniformly $L^p$ 
\emph{near infinity}, or $V\in L_{\rm{loc,unif}}^p$ 
\emph{near infinity} if there exists a 
compact set $K\subset \R^d$ such that 
$\id_{K^c}V\in L_{\rm{loc,unif}}^p(\R^d)$. 

In the following we will always assume that 
$p=1$ for $d=1$ and $p>d/2$ for $d\ge 2$.  
If $V\in L_{\rm{loc,unif}}^p$ near infinity, then 
\eqref{eq-form-tiny-quantitative-locally-uniform} and the diamagnetic inequality 
shows that the quadratic form domain $\calQ(V)$ contains all 
$\varphi\in \calD(P-A)$ with 
$\supp(\varphi)\subset \U_R^c = \{x\in\R^d:\, |x|\ge R\}$ 
as soon as $R>0$ is large enough. 

Recalling the notation of 
Definition \ref{def-vanishing}, the bound 
\eqref{eq-form-tiny-quantitative-locally-uniform} also shows that 
for any  $R$ large enough and all $\veps>0$ we have 
\begin{align}
  \alpha_R \le \veps \text{ and }
  \gamma_R \le \veps L+ G\big(\veps^{-1}, \kappa \|V_R\|_{L_{\rm{loc,unif}}^p}\big) \label{eq-alphaR-gammaR bound}
\end{align}
with $V_R= \id_{\U_R^c}V$. 

Now assume that $V$ vanishes at infinity locally uniformly in 
$L^p$, that is,  
\begin{equation}\label{vanishing locally uniformly in Lp}
	\lim_{R\to\infty} \|V_R\|_{L_{\rm{loc,unif}}^p}= 0\, . 
\end{equation}
Since $\lim_{s_2\to 0}G(s_1,s_2)= 0$ for any $s_1>0$ we 
can, for any $n\in\N$, inductively choose $R_n\to \infty$  
such that 
$ G\big(n, \kappa \|V_R\|_{L_{\rm{loc,unif}}^p}\big)
\to 0$ as $n\to\infty$.  Clearly, \eqref{eq-alphaR-gammaR bound} shows that $\alpha_{R_n}\to 0$ and 
$\gamma_{R_n}\to 0$ as $n\to\infty$. Since we can also 
assume, without loss of generality, that $\alpha_R$ and 
$\gamma_R$ are decreasing in $R\ge R_0$,  once they 
exist for some $R_0>0$, this shows that $V$ vanishes at 
infinity w.r.t.\ $(P-A)^2$ as soon as it vanishes at 
infinity locally uniformly in $L^p$. 

With an additional trick it turns out that it is enough 
to only assume that $V$ is locally uniformly $L^p$ at 
infinity and vanished at infinity locally uniformly in $L^1$. 

\begin{proposition} \label{prop-vanishing Lp}
Let $p=1$ for $d=1$ and $p>d/2$ for $d\ge 2$. 
Assume that the potential $W \in L_{\rm{loc,unif}}^p$ near 
infinity and that it vanishes at infinity locally uniformly 
in  $L^1$, that is, with $\id_{\U_R^c}$ the characteristic 
function of $\U_R^c=\{x\in\R^d: |x|\ge R\}$ and 
$W_R\coloneqq \id_{\U_R^c}W$ we have  
\begin{align}\label{eq-vanishing locally uniformly in L1}
	\lim_{R\to\infty}\|W_R\|_{L^1_{\text{loc,unif}}} =0\, .
\end{align}
Then $W$ vanishes at infinity w.r.t.~$P^2$ in the sense of Definition \ref{def-vanishing}. 

Moreover, if $V= \nabla\cdot \Sigma +W$ for some vector field $\Sigma\in L^2_{\text{loc}}$ and a potential  $W\in L^1_{\text{loc}}$ and $\Sigma^2$ and $W$ satisfy the above assumptions, then $V$ also vanishes at infinity w.r.t.~$P^2$ in the sense of Definition \ref{def-vanishing}. 
\end{proposition}
\begin{proof}
The discussion just before Proposition 
\ref{prop-vanishing Lp} shows that $W$ vanishes at infinity 
w.r.t.~$P^2$ in the sense of Definition \ref{def-vanishing}  
if we use the norm 
$\|\cdot\|_{L^p_{\text{loc,unif}}}$ instead of 
the norm $\|\cdot\|_{L^1_{\text{loc,unif}}}$ 
in \eqref{eq-vanishing locally uniformly in L1}. 
If $p=1$, i.e., $d=1$,  then there is nothing to prove. 

So assume $d\ge 2$ and $p>d/2\ge 1$. Pick 
$R_0>0$ so large that 
$W_{R_0}\in L_{\rm{loc,unif}}^p(\R^d)$. 
Since $p>d/2\ge 1$, there exist $1\le d/2 <q<p$. Replacing 
$p$ by $q$ in the discussion just before Proposition 
\ref{prop-vanishing Lp} shows that $W$ vanishes at infinity 
w.r.t\ $(P-A)^2$ as soon as one knows $\lim_{R\to\infty}\|W_R\|_{L_{\rm{loc,unif}}^q}=0 $. 
This is easy.  Since $1\le d/2<q<p$ there exists 
$0<\theta<1$  with $q=\theta 1 +(1-\theta)p$. Thus 
for all $R\ge R_0$  
H\"older's inequality implies 
\begin{align*}
	\|W_R\|_{L_{\rm{loc,unif}}^q}
	  \le 	\|W_R\|_{L_{\rm{loc,unif}}^p}^{1-\theta}
			\|W_R\|_{L_{\rm{loc,unif}}^1}^\theta 
	\le  	\|W\|_{L_{\rm{loc,unif}}^p}^{1-\theta}
			\|W_R\|_{L_{\rm{loc,unif}}^1}^\theta 
			\to 0 \text{ as } R\to\infty
\end{align*}

The second claim of Proposition \ref{prop-vanishing Lp} 
follows from the first and the second part of 
Proposition \ref{prop-vanishing-properties}. 
\end{proof}

\subsection{Vanishing at infinity for potentials in the Kato--class}\label{ssec-kato-class}

To get a replacement for the borderline case $p=d/2$ one can use the Kato--class, which we recall.

\begin{definition}[Kato--class]
	A real-valued and measurable function $V$ on $\R^d$ is in the Kato--class $K_d$ if 
	\begin{align}
		\lim_{\alpha\to 0}\sup_{x\in\R^d} \int_{|x-y|\le \alpha} g_d(x-y) |V(y)|\, dy
		=0
	\end{align}
	where 
	\begin{equation}
		g_d(x)\coloneqq 
		  \left\{ 
		    \begin{array}{lcl}
		    	|x|^{2-d} &\text{if} & d\ge 3 \\
		    	|\ln|x|| &\text{if} & d= 2 
		    \end{array}
		  \right. .
	\end{equation}
	One also defines the Kato--norm 
	\begin{align}
		\|V\|_{K_d}
			&\coloneqq \left\{
				\begin{array}{lcl}
				 \sup_{x\in\R^d} \int_{|x-y|\le 1} |x-y|^{d-2} |V(y)|\, dy \, ,
				 &\text{if  } 
				 & d\ge 3 \\
				 \sup_{x\in\R^2} \int_{|x-y|\le 1/2} |\ln(|x-y|)| |V(y)|\, dy \, ,
				 &\text{if  } 
				 & d=  2
				\end{array}
				 \right. \, .
	\end{align}
 
\end{definition}

It is well-known that any Kato--class potential is infinitesimally form bounded with respect to $P^2$, see e.g.~\cite[Thm.~1.4]{azs}, thus also with respect to $(P-A)^2$ for any vector potential $A\in L^2_{\text{loc}}(\R^d,\R^d)$.  It is also clear that $K_d\subset L^1_{\text{loc,unif}}(\R^d)$ and using H\"older's inequality one easily sees $L^p_{\text{loc,unif}}(\R^d)\subset K_d$ for all $p>d/2$.  

Lastly, we say that a potential $V$ is in the Kato--class outside a compact set, if there exists a compact set $K\subset \R^d$ such that $\id_{K^c}V\in K_d$. Here $\id_{K^c}$ is the characteristic function of the complement of $K$.  
For potentials which are in the Kato--class outside of a compact set we also have a simple condition for vanishing. 

\begin{proposition} \label{prop-vanishing Kato class}
Given a potential $W$ assume that it is in the Kato--class outside a compact set and that it vanishes at infinity locally uniformly in  $L^1$, that is, 
\begin{align}
	\lim_{R\to\infty}\|\id_{\ge R}W\|_{L^1_{\text{loc,unif}}} =0\, .
\end{align}
with $\id_{\ge R}$ the characteristic function of 
$\{x\in\R^d: |x|\ge R\}$. 
Then $W$ vanishes at infinity w.r.t.~$P^2$ in the sense of Definition \ref{def-vanishing}. 

Moreover, if $V= \nabla\cdot \Sigma +W$ for some vector field $\Sigma\in L^2_{\text{loc}}$ and a potential  $W\in L^1_{\text{loc}}$ and $\Sigma^2$ and $W$ satisfy the above assumptions, then $V$ also vanishes at infinity w.r.t.~$P^2$ in the sense of Definition \ref{def-vanishing}. 
\end{proposition}

\noindent In  the proof of Proposition \ref{prop-vanishing Kato class} we need

\begin{lemma}\label{lem-eq-asymptotic-upper-bound-lambda}
	Given a potential $W$ in the Kato-class assume that there exist $R_0>0$ and  
	$\alpha_{R,\lambda}, \gamma_{R,\lambda}\ge 0$ for $R_0>0$ and  $R\ge R_0, \lambda>0$ such that 
	\begin{align}\label{eq-asymptotic-upper-bound-lambda}
		\La \varphi, W\varphi \Ra \le \alpha_{R,\lambda}\|(P-A)\varphi\|_2^2 + \gamma_{R,\lambda}\|\varphi\|_2^2
	\end{align}
    for all $\varphi\in\calD(P-A)$ with $\supp(\varphi)\in \U_R^c$. Moreover, assume that 
    $ R_0\le R\mapsto \alpha_{R,\lambda}, \gamma_{R,\lambda}$ 
    are decreasing for fixed $\lambda>0$ and 
    $\lim_{\lambda\to\infty}  \alpha_{R,\lambda}=0$  
	for fixed  $R\ge R_0$. 
	
	Then $W$ is bounded from above at infinity w.r.t $(P-A)^2$ with asymptotic bound 
	\begin{equation}\label{eq-asymptotic upper bound quantitative}
		\gamma^+_\infty(W)\le \liminf_{\lambda\to\infty}\lim_{R\to\infty} \gamma_{R,\lambda}\, .
	\end{equation}
\end{lemma}
\begin{remark}
	The order of the limits in \eqref{eq-asymptotic upper bound quantitative} is important, since typically one has 
	$\liminf_{\lambda\to\infty}  \gamma_{R,\lambda}=\infty$ for any fixed $R$. 
	
	Given any $\alpha_{R,\lambda}, \gamma_{R,\lambda}$ for which \eqref{eq-asymptotic-upper-bound-lambda} holds,  one can, by a simple monotonicity argument, replace them with $\alpha'_{R,\lambda}\coloneqq \inf_{R_0\le L\le R} \alpha_{L,\lambda} $ and 
	 $\gamma'_{R,\lambda}\coloneqq \inf_{R_0\le L\le R} , \gamma_{L,\lambda}$, i.e., the required 
	 monotonicity in $R$ in Lemma \ref{lem-eq-asymptotic-upper-bound-lambda} is not a restriction. 
\end{remark}
\begin{proof}
  Let $\wti{\gamma}_\lambda= \lim_{R\to\infty}\gamma_{R,\lambda}$. 
  Pick any $\lambda_0>0$ and given $R_n, \lambda_n$ for $n\in\N_0$ choose inductively 
  $ \lambda_{n+1}\ge \lambda_n+1$ with $ \alpha_{R_n,\lambda_{n+1}} \le  \frac{1}{n+1} $
  and then $R_{n+1}\ge R_n +1$ with  $\gamma_{R_{n+1},\lambda_{n+1}}\le \frac{1}{n+1}+ \wti{\gamma}_{\lambda_{n+1}}$. 
  
  Take a subsequence $n_j$ with $\wti{\gamma}_{j}\coloneqq \wti{\gamma}_{n_j}\to \liminf_{n\to\infty} \wti{\gamma}_\lambda$ as $j\to\infty$ and set  
  $\alpha_R\coloneqq \frac{1}{n_j +1}$ and 
  $\gamma_R\coloneqq \frac{1}{n_j +1}+ \wti{\gamma}_j$ for $R\in [R_{n_j}, R_{n_{j+1}})$. With this choice  Definition \ref{def-bounded infinity} is satisfied, 
  so $W$ is asymptotically bounded at infinity 
  w.r.t.~$(P-A)^2$ and  $\gamma_\infty(W)=\lim_{R\to\infty}\gamma_R= \lim_{j\to\infty}\wti{\gamma}_j=\liminf_{\lambda\to\infty}\lim_{R\to\infty} \gamma_{R,\lambda}$.
\end{proof}

\begin{proof}[Proof of Proposition \ref{prop-vanishing Kato class}] Given a locally square integrable magnetic 
vector potential $A$ we abbreviate $H_0=(P-A)^2$ for the free magnetic Schr\"odinger operator defined by quadratic form methods.  
Given a potential $W$ in the Kato--class,  $\varphi\in \calD(P-A)= \calQ(H_0)$, and $\lambda>0$ let 
$f=(H_0+\lambda)^{1/2}\varphi\in L^2$. Then 
\begin{align*}
	|\La\varphi, W\varphi\Ra| 
		&\le \La\varphi, |W|\varphi\Ra
			=  \La f,(H_0+\lambda)^{-1/2} |W|(H_0+\lambda)^{-1/2} f\Ra 
			\le \|(H_0+\lambda)^{-1/2} |W|(H_0+\lambda)^{-1/2} \|_{2\to2} \|f\|_2^2 \\
		& =  \|(H_0+\lambda)^{-1/2} |W|(H_0+\lambda)^{-1/2} \|_{2\to2} 
			\left(  
				\|((P-A)^2\varphi\|_2^2 + \lambda\|\varphi\|_2^2
			\right)
\end{align*}
By duality, $\|(H_0+\lambda)^{-1/2} |W|(H_0+\lambda)^{-1/2} \|_{2\to2} = \||W|^{1/2}(H_0+\lambda)^{-1} |W|^{1/2} \|_{2\to2} $. Assume that $|W|$ is bounded, then for $0\le \re(z)\le 1$ the operator family $T_z= |W|^{z}(H_0+\lambda)^{-1} |W|^{1-z}$ is analytic and bounded. 

Using the diamagnetic inequality and duality we have 
\begin{equation*}
	\||W|(H_0+\lambda)^{-1}\|_{1\to1} 
		= \|(H_0+\lambda)^{-1}|W|\|_{\infty\to\infty}
		= \|(H_0+\lambda)^{-1}|W|\|_{\infty}
		\le \|(P^2+\lambda)^{-1}|W|\|_\infty\, ,
\end{equation*}
which is finite for any $\lambda>0$ and bounded $W$. 
Thus $T_z$ is bounded from $L^1\to  L^1$ for $\re(z)=0$ 
and from  $L^\infty\to  L^\infty$ for $\re(z)=1$ and as 
in \cite{cfks} one can use the Stein interpolation 
theorem \cite{rs1} to see 
\begin{align*}
  \|(H_0+\lambda)^{-1/2} |W|(H_0+\lambda)^{-1/2} \|_{2\to2} 
    \le \|(P^2+\lambda)^{-1}|W|\|_{\infty} \, .
\end{align*}
at least for bounded $W$. 
If $\supp(\varphi)\subset \U_R^c$, one can replace $W$ by 
$W_R=\id_{\ge R}W$. Thus 
\begin{align*}
	|\La\varphi, W\varphi \Ra| = 	|\La\varphi, W_R\varphi \Ra|  
		&\le \alpha_{R,\lambda}\|(P-A)\varphi\|_2^2 + \gamma_{R,\lambda}\|\varphi\|_2^2
\end{align*}
for all $\varphi\in\calD(P-A)$ with  
$\supp(\varphi)\subset \U_R^c$, choosing  
\begin{equation}\label{eq-superdooper}
  \begin{split}
	\alpha_{R,\lambda}
		& = \|(P^2+\lambda)^{-1}|W_R|\|_{\infty}\, ,\\
	\gamma_{R,\lambda}
		& = \lambda\|(P^2+\lambda)^{-1}|W_R|\|_{\infty}\, .
  \end{split}
\end{equation} 
If $W_R$ is unbounded, replace $W_R$ by $\min(|W_R|,n)$ 
and take the limit $n\to\infty$ to see that the above bounds work also for unbounded $W$, as long as the right 
hand side of   \eqref{eq-superdooper} is finite. 

Clearly, $\alpha_{R,\lambda}$ and $\gamma_{R,\lambda}$ are decreasing in $R$ for fixed $\lambda>0$. 
One even has  
$\lim_{\lambda\to\infty} \|(P^2+\lambda)^{-1}|W|\|_{\infty}= 0$ if and only if $W$ is in the Kato--class, which is well--known, see \cite{cfks,simon-semigroups-review}. 
However, we also clearly have 
$\lim_{\lambda\to\infty}\gamma_{R,\lambda}=\|W_R\|_\infty$, which is finite, if and only if $W_R$ is bounded. 
Nevertheless,  if $W_R$ is in the Kato class for some, hence all, large enough  $R$ and  
$
	\lim_{R\to\infty}\|W_R\|_{L^1_{\text{loc,unif}}}=0
$ 
then 
\begin{equation}\label{eq-superdooper-2}
	\lim_{R\to\infty} \|(P^2+\lambda)^{-1}|W_R|\|_{\infty}
		= 0\, ,
\end{equation}
which together with Lemma \ref{lem-eq-asymptotic-upper-bound-lambda} shows $\gamma_\infty(W)=0$. This proves the first part of Proposition \ref{prop-vanishing Kato class}. 
The other claim of Proposition \ref{prop-vanishing Kato class} follows from the above since by Proposition \ref{prop-vanishing-properties} $W=\nabla\cdot\Sigma$ vanishes w.r.t~$(P-A)^2$ as soon as $\Sigma^2$ does.  For the proof of \eqref{eq-superdooper-2}, we claim that 
 for any potential $W$ and any $0<\alpha\le 1$
 \begin{align}\label{eq-superdooper-3}
   \|(P^2+\lambda)^{-1}|W|\|_\infty 
     \lesssim 
         \sup_{x\in\R^d} \int_{|x-y|\le \alpha}g_d(x-y)|W(y)|\, dy 
         +         \frac{e^{-\sqrt{\lambda}\alpha/4}}{\sqrt{\lambda}\alpha}\|W\|_{L^1_{\text{loc,unif}}}	 
 \end{align}
 where the implicit constant depend only on $d$.   
 This clearly proves \eqref{eq-superdooper-2}, since replacing $W$ by $W_R=\id_{\ge R}W$ it yields  
 \begin{align*}
 	\limsup_{R\to\infty}\|(P^2+\lambda)^{-1}|W_R|\|_\infty
 	  \le 
 		     C_{\lambda,d}  \sup_{x\in\R^d} \int_{|x-y|\le \alpha}g_d(x-y)|W_{R_0}(y)|
 \end{align*}
 for any fixed $R_0, \lambda>0$ and all $0\le \alpha\le 1$ as soon as 
 $\lim_{R\to\infty} \|W_R\|_{L^1_{\text{loc,unif}}}=0$.   Since $W_{R_0}$ is in the Kato--class, we can then take the limit 
 $\alpha\to 0$ to get  \eqref{eq-superdooper-2}.  It remains to prove \eqref{eq-superdooper-3}. Note
 \begin{align*}
 	\|(P^2+\lambda)^{-1}|W|\|_{\infty} 
 	  = \sup_{x\in\R^d} \int_{\R^d} G(x,y,\lambda)|W(y)|\, dy
 \end{align*}
 where $G(x,y),\lambda= (P^2+\lambda)^{-1}(x,y)$ is 
 the Green's function, i.e., the kernel of  
 $(P^2+\lambda)^{-1}$. 
 We split the integral above in the two regions $|x-y|\le \alpha$ and $|x-y|>\alpha$.   The bounds 
 \begin{align}
 	G(x,y,\lambda)
 		&\lesssim 
 		  \lambda^{-1}|x-y|^{-d} e^{-\sqrt{\lambda}|x-y|/2}\label{eq-green function bound 1}\\
 	\intertext{and for $|x-y|\le 1/2$ and $\lambda\ge 1$}
 	 G(x,y,\lambda)
 	 	&\lesssim	\left\{  
 	 					\begin{array}{ccc} 
 	 						|x-y|^{2-d}& \text{if} & d\ge 3 \\
 	 						|\ln|x-y||& \text{if} & d= 2 
 	 					\end{array}
 	 				\right. \label{eq-green function bound 2}
 \end{align}
 are well-know. Integrating over shells 
  $\alpha n\le |x-y|<\alpha (n+1) $ leads to 
\begin{align*}
  \sup_{x\in\R^d}& \int_{|x-y|>\alpha} G(x,y,\lambda)|W(y)|\, dy 
    \lesssim 
		 \lambda^{-1}
		   \sum_{n=1}^\infty 
		     e^{-\sqrt{\lambda}\alpha n/2} 
		         \frac{(\alpha (n+1))^d - (\alpha n)^d}{(\alpha n)^d}
		     \|W\|_{L^1_{\text{loc,unif}}} \\
	&\lesssim \lambda^{-1} \sum_{n=1}^\infty 
		     e^{-\sqrt{\lambda}\alpha n/2}
		     \|W\|_{L^1_{\text{loc,unif}}} 
	  = \frac{e^{-\sqrt{\lambda}\alpha/2}}{\lambda(1-e^{-\sqrt{\lambda}\alpha/2})} 
		     \|W\|_{L^1_{\text{loc,unif}}}
	     \lesssim  
	     	\frac{e^{-\sqrt{\lambda}\alpha/4}}{\sqrt{\lambda}\alpha} 
	     	\, \|W\|_{L^1_{\text{loc,unif}}}
\end{align*}
since $0<t\mapsto \frac{te^{-t/2}}{1-e^{-t}}$ is bounded. 
This proves \eqref{eq-superdooper-3}. 

\smallskip

\noindent  We sketch the proof of the bounds \eqref{eq-green function bound 1} and \eqref{eq-green function bound 2}, for the convenience of the reader: 
 The kernel of the heat semigroup is $e^{-P^2t}(x,y)= (4\pi t)^{-d/2}e^{-\frac{|x-y|^2}{4t}}$. 
 Since $(P^2+\lambda)^{-1}= \int_0^\infty e^{-P^2s- \lambda s}\, ds$ 
 we have 
 \begin{align*}
 	G(x,y,\lambda) 
 	  &= \int_0^\infty (4\pi s)^{-d/2} e^{-\frac{|x-y|^2}{4s}} e^{-\lambda s}\, ds\, 
 	    = |x-y|^{2-d}\int_0^\infty (4\pi u)^{-d/2} e^{-\frac{1}{4u}}  e^{-\lambda |x-y|^2u}\, du
 \end{align*}
 Moreover, 
 $\frac{1}{4u}+ \lambda|x-y|^2u\ge \sqrt{\lambda}|x-y| $ 
 for all $u>0$, so 
 \begin{align*}
 	&G(x,y,\lambda) 
 	  \le |x-y|^{2-d}e^{-\sqrt{\lambda}|x-y|/2}\int_0^\infty (4\pi u)^{-d/2} e^{-\frac{1}{8u}}  e^{-\lambda |x-y|^2u/2}\, du
 	  	\lesssim \frac{ |x-y|^{-d}e^{-\sqrt{\lambda}|x-y|/2}}{\lambda}
 \end{align*}
 since $0<t\mapsto te^{-t}$ is bounded and 
 $c_d= \int_0^\infty (4\pi u)^{-d/2}  e^{-\frac{1}{4u}} \, \frac{du}{u} <\infty$ for all $d\ge 1$. This proves \eqref{eq-green function bound 1}. 
 
On the other hand,  
\begin{align*}
	 	G(x,y,\lambda) 
 	  &= |x-y|^{2-d}\int_0^\infty (4\pi u)^{-d/2} e^{-\frac{1}{4u}}  e^{-\lambda |x-y|^2u}\, du
 	   \le \wti{c}_d  |x-y|^{2-d}
\end{align*}
where $\wti{c}_d  = \int_0^\infty (4\pi u)^{-d/2} e^{-\frac{1}{4u}} \, du <\infty $ if $d\ge 3$, which proves \eqref{eq-green function bound 2} when $d\ge 3$. 
If $d=2$, then for $0<|x-y| \le 1/2$, one has 
\begin{align*}
	G(x,y,\lambda) = (4\pi)^{-1}\int_0^\infty e^{\frac{1}{4u}} e^{-\lambda|x-y|^2u}\, \frac{du}{u}
		\lesssim
		  \int_0^1 e^{\frac{1}{4u}}\, \frac{du}{u}
		  + \int_1^{|x-y|^{-2}} \, \frac{du}{u}
		  + \int_{|x-y|^{-2}}^\infty  e^{-\lambda|x-y|^2u}\, \frac{du}{u}
\end{align*}
Since $\int_0^1 e^{\frac{1}{4u}}\, \frac{du}{u}\lesssim 1$ and 
$ \int_{|x-y|^{-2}}^\infty  e^{-\lambda|x-y|^2u}\, \frac{du}{u} =  \int_{1}^\infty  e^{-\lambda u}\, \frac{du}{u}\le 1
$ for $\lambda\geq 1$, this proves \eqref{eq-green function bound 2}.
\end{proof}

\section{Gronwall type bounds}
\label{sec-app} 

\begin{lemma} \label{lem-gronwall}
Let $T>0$ and let $w, E: [0,T] \to [0,\infty)$. If for some $c>0$  
\begin{equation} \label{gron-hyp} 
w(t) \leq E(t) + c \int_0^t e^{t-s} \, w(s) \, ds, 
\end{equation}
for all $t\in [0,T]$, then 
\begin{equation} \label{gron-eq} 
w(t) \leq E(t) + c \int_0^t e^{(1+c)(t-s)} \, E(s) \, ds \qquad 	\forall\, t\in [0,T]. 
\end{equation}
Moreover, if  
\begin{equation} \label{gron-hyp-2} 
w(t) \leq E(t) + c \int_0^t e^{s-t} \, w(s) \, ds, 
\end{equation} 
for all $t\in[0,T]$, then 
\begin{equation} \label{gron-eq-2} 
w(t) \leq E(t) + c \int_0^t e^{(c-1)(t-s)} \, E(s) \, ds \qquad 	\forall\, t\in [0,T]. 
\end{equation}

\end{lemma} 

\begin{proof}
Put $v(t): =  \int_0^t e^{t-s} \, w(s) \, ds$. Then $v(0)=0$ and, assuming \eqref{gron-hyp}, 
$$
v'(t)  = v(t) +w(t) \leq E(t) +(1+c) v(t) 
$$
Hence
\begin{align*}
\frac{d}{dt} \Big (e^{-(1+c) t}\,  v(t) \Big) & = e^{-(1+c) t} (v'(t) -(1+c) v(t)) \leq e^{-(1+c) t}\,  E(t). 
\end{align*}
It follows that 
\begin{align*}
e^{-(1+c) t}\,  v(t) & = \int_0^t \frac{d}{ds} \Big (e^{-(1+c) s}\,  v(s) \Big)\, ds \leq \int_0^t \, e^{-(1+c) s}\,  E(s)\, ds\, .
\end{align*}
This implies 
$$
v(t)  \leq \int_0^t e^{(1+c)(t-s)} \, E(s) \, ds, 
$$
and \eqref{gron-eq} follows, cf.~\eqref{gron-hyp}.
\end{proof}

\section{Optimizing the threshold}\label{app-optimized threshold}
It is tempting to split the potential $V=V_1+V_2$ at infinity in order to optimize the threshold above which one can exclude existence of eigenvalues. Using $V_1= sV$ and $V_2=(1-s)V$, Theorem \ref{thm-abs} shows the non--existence of eigenvalues with 
\begin{align*}
	E> \frac{1}{4}\Big( \beta + \omega_1 s +\sqrt{(\beta+\omega_1 s)^2 + 2\omega_2(1-s)}  \Big)^2 
	  = \frac{\omega_1^2}{4} (g(s))^2
\end{align*}
where for $0\le s\le 1$ we set 
\begin{align}\label{def-g}
	g(s)\coloneqq b+s + \sqrt{(b+s)^2+2c(1-s)}
\end{align}
with $b=\beta/\omega_1$ and $c= \omega_2/\omega_1^2$. 
The goal is to minimize $g$ over $s\in [0,1]$. 

\begin{lemma}[Bang--Bang type Lemma]\label{lem-bang-bang}
  For $g$ given in \eqref{def-g} we have $\min_{0\le s\le 1}g(s)\ge \min(g(0),g(1))$. More precisely, 
  \begin{equation}
  	\min_{0\le s\le 1} g(s) 
  	= \left\{  
  	    \begin{array}{ccc}
  	      g(0) & \text{if} & c<2b+2 \\
  	      g(1) & \text{if} &  c>  2b+2 
  	    \end{array}
  	  \right.
  	  \, 
  \end{equation}
  and $g$ is constant if $c=2b+2$. 
\end{lemma}
\begin{proof}
  Write $c= 2b+2+r$. Then $(b+s)^2+2c(1-s)= (b+2-s)^2+2r(1-s)$, hence 
  \begin{align*}
    g(s)= b+s+\sqrt{(b+2-s)^2+2r(1-s)}	
  \end{align*}
  for all $0\le s\le 1$. Note that $g$ is clerly constant on $[0,1]$ if $r=0$. On $[0,1]$ the derivative of $g$ is given by 
  \begin{align*}
  	g'(s) = 1+((b+2-s)^2+2r(1-s))^{-1/2}\big( s-(b+2+r) \big)\, .
  \end{align*}
  Fix $0\le s\le 1$. A calculation shows 
  \begin{align*}
  	\left( ((b+2-s)^2+2r(1-s))^{-1/2}\big( s-(b+2+r) \big) \right)^2>1
  \end{align*}
  if and only if $0<r(r+2b+2)=rc$. 
  Since $c\ge 0$, this implies that if $r<0$, i.e., $c<2b+2$, 
  we  have $g'>0$ on $[0,1]$, i.e., $g$ is strictly increasing 
  on $[0,1]$. 
   On the other hand, if $c>2b+2$, then also $c-b> b+2\ge 2$ and $r<0$, so $g'<0$ on $[0,1]$, i.e., $g$ is strictly decreasing on $[0,1]$. This proves the lemma.   
\end{proof}
\begin{corollary}\label{cor-optimized threshold}
  Setting \begin{equation} 
  \begin{split}
	\beta^2 \coloneqq 	\gamma_\infty\big(\wti{B}^2\big), \quad 
	\omega_1^2 \coloneqq 	\gamma_\infty\big((xV)^2\big) , \quad 
	\omega_2 \coloneqq 	\gamma^+_\infty\big(x\cdot\nabla V\big)
  \end{split}
\end{equation}
 the threshold $\Lambda(B,V)$ defined in \eqref{edge} optimized for splitting the potential 
 as $V=V_1+V_2$ with $V_1=sV$, $V_2=(1-s)V$ and $0\le s\le 1$ is given by 
 \begin{align}
 	\wti{\Lambda}(B,V)
 	  = \left\{\begin{array}{ccc}  
 	      \frac{1}{2}\left( \beta^2+ \omega_2+\beta\sqrt{\beta+2\omega_2} \right)
 	        & \text{if}& \omega_2\le 2\omega_1(\beta+\omega_1) \\
 	      (\beta+\omega_1)^2 	        
 	        & \text{if}& \omega_2>2\omega_1(\beta+\omega_1)
 	   	\end{array}\right. 
 \end{align}
\end{corollary}
\begin{proof}
  Given Lemma \ref{lem-bang-bang} this is just a simple calculation.	
\end{proof}

\section{IMS localization formula}\label{app-ims}
In one step in the proof of Lemma 
\ref{lem-kinetic energy bounded} we need a quadratic form 
version of the well-known IMS 
localization formula under minimal assumptions on the quadratic 
form of the magnetic Schr\"odinger operator. This result is not new, see e.g.~\cite[pp.~98, Prop.~4.2]{nr}. For the sake of completeness 
we include a short proof.

\begin{theorem}[IMS localization formula] \label{thm-ims}
  Let $A$ be a locally square integrable magnetic vector 
  potential and $V$ form small w.r.t.~$(P-A)^2$. Then for all 
  bounded real--valued $\xi\in \calC^\infty(\R^d)$ such that 
  $\nabla\xi$ is also bounded and all $\varphi\in\calD(P-A)$, 
  also $\xi\varphi$ and $\xi^2\varphi\in \calD(P-A)$ and  
  \begin{align}\label{eq-ims}
  	\re q_{A,V}(\xi^2\varphi,\varphi) = q_{A,V}(\xi\varphi, \xi\varphi) -\La \varphi,|\nabla\xi|^2\varphi \Ra
  \end{align}
\end{theorem}

\begin{proof}
	As before, one easily checks that $\xi\varphi$ and 
	$\xi^2\varphi$ are in the domain of $P-A$ when $\varphi$ is.
	Moreover, the potential $V$ commutes with the 
	multiplication operator $\xi$, so as quadratic forms 
	$\La \xi^2\varphi, V\varphi \Ra=\La \xi\varphi, V\xi\varphi \Ra$ 
	and we only have to check the kinetic energy term. 
	Since $(P-A)(\xi^2\varphi)=\xi(P-A)(\xi\varphi)+ (P\xi)\xi\varphi$ a 
	short calculation reveals 
	\begin{align*}
		\La (P-A)(\xi^2\varphi)  ,(P-A)\varphi \Ra
		  &= \La (P-A)(\xi\varphi)  ,(P-A)(\xi\varphi) \Ra
		      + \La (P\xi)\varphi  ,(P-A)(\xi\varphi) \Ra \\
		  &\phantom{==}    - \La (P-A)(\xi\varphi)  ,(P\xi)\varphi) \Ra 
			  - \La \varphi,|\nabla\xi|^2  \varphi \Ra\, ,
	\end{align*}
  	so 
	\begin{align*}
		\re q_{A,0}(\xi^2\ \varphi,\varphi) 
			&= \re \La (P-A)(\xi^2\varphi), (P-A)\varphi\Ra  \\
			&= \La (P-A)(\xi\varphi), (P-A)(\xi\varphi)\Ra
				\, +\La \varphi,|\nabla\xi|^2\varphi\Ra 
	\end{align*}
	which proves \eqref{eq-ims}.
\end{proof}

\appendix
\setcounter{section}{0}
\renewcommand{\thesection}{\Alph{section}}
\renewcommand{\theequation}{\thesection.\arabic{equation}}
\renewcommand{\thetheorem}{\thesection.\arabic{theorem}}
%

\renewcommand{\thesection}{\arabic{chapter}.\arabic{section}}
\renewcommand{\theequation}{\arabic{chapter}.\arabic{section}.\arabic{equation}}
\renewcommand{\thetheorem}{\arabic{chapter}.\arabic{section}.\arabic{theorem}}

\bigskip
\noindent 
\textbf{Acknowledgments:}  We thank Rupert Frank and Semjon Wugalter for  useful discussions.  Hynek Kova\v{r}\'{\i}k has been partially supported by Gruppo Nazionale per Analisi Matematica, la Probabilit\`a e le loro Applicazioni (GNAMPA) of the Istituto Nazionale di Alta Matematica (INdAM). 
Dirk Hundertmark has been partially funded by the Deutsche Forschungsgemeinschaft (DFG, German Research Foundation) -- Project-ID 258734477 -- SFB 1173.


\end{document}